\documentclass[a4paper,11pt]{article}
\usepackage{geometry,a4wide}
\usepackage{amsthm,amssymb,amsmath}
\usepackage{thmtools}
\usepackage{hyperref,cleveref}
\usepackage{mathrsfs}
\usepackage{graphicx} 
\usepackage[maxnames=99]{biblatex}
\usepackage{todonotes}
\usepackage{tikz}
\usepackage{caption}

\usetikzlibrary{calc}

\usepackage{tkz-graph}
\usetikzlibrary{arrows}
\usetikzlibrary{decorations.markings}
\usetikzlibrary{decorations.pathmorphing}
\tikzset{snake it/.style={decorate, decoration=snake}}
\usepackage{circuitikz}

\usepackage{thm-restate}
\usepackage[blocks]{authblk}
\usepackage{thmtools}
\usepackage{xpatch}
\makeatletter
\xpatchcmd{\thmt@restatable}
{\csname #2\@xa\endcsname\ifx\@nx#1\@nx\else[{#1}]\fi}
{\ifthmt@thisistheone\csname #2\@xa\endcsname\ifx\@nx#1\@nx\else[{#1}]\fi\else\csname #2\@xa\endcsname\fi}
{}{} 
\makeatother

\MakeRobust{\ref}

\makeatletter
\newcommand{\labeltext}[2]{%
  \@bsphack
  \csname phantomsection\endcsname 
  \def\@currentlabel{#1}{\label{#2}}%
  \@esphack
}
\makeatother

\usepackage[noframe]{showframe}
\usepackage{framed}
\renewenvironment{shaded}{%
  \MakeFramed{\advance\hsize-\width \FrameRestore\FrameRestore}}%
 {\endMakeFramed}
\definecolor{shadecolor}{gray}{0.88}

\title{Quantum Polymorphisms and the Complexity of\\ Quantum Constraint Satisfaction}

\author{Lorenzo Ciardo\textsuperscript{\textdagger}, Gideo Joubert\textsuperscript{\ddag}, and Antoine Mottet\textsuperscript{\ddag}\\
\textsuperscript{\textdagger}TU Graz\\
\textsuperscript{\ddag}Hamburg University of Technology
}

\newtheorem{theorem}{Theorem}
\newtheorem{proposition}[theorem]{Proposition}
\newtheorem{fact}[theorem]{Fact}

\newtheorem{corollary}[theorem]{Corollary}
\newtheorem{lemma}[theorem]{Lemma}
\newtheorem{claim}[theorem]{Claim}

\theoremstyle{definition}
\newtheorem{definition}[theorem]{Definition}

\theoremstyle{remark}
\newtheorem{remark}[theorem]{Remark}
\newtheorem{example}[theorem]{Example}

\newcommand{\ba}{\mathbf{a}}
\newcommand{\bb}{\mathbf{b}}
\newcommand{\bc}{\mathbf{c}}

\newcommand{\bs}{\mathbf{s}}
\newcommand{\bu}{\mathbf{u}}
\newcommand{\bx}{\mathbf{x}}
\newcommand{\bv}{\mathbf{v}}
\newcommand{\by}{\mathbf{y}}
\newcommand{\bw}{\mathbf{w}}
\newcommand{\be}{\mathbf{e}}

\newcommand{\bq}{\mathbf{q}}
\newcommand{\bz}{\mathbf{z}}
\newcommand{\A}{\rel{A}}
\newcommand{\B}{\rel{B}}
\newcommand{\C}{\rel{C}}
\newcommand{\D}{\rel{D}}
\newcommand{\K}{\rel{K}}
\newcommand{\GG}{\rel{G}}
\newcommand{\N}{\mathbb{N}}
\newcommand{\R}{\mathbb{R}}

\newcommand{\Z}{\mathbb{Z}}
\newcommand{\X}{\rel{X}}
\newcommand{\Y}{\rel{Y}}

\newcommand\MIP{\ensuremath{\mathsf{MIP}^*}}

\newcommand\ang[2]{{\ensuremath\langle #1,#2\rangle}}
\newcommand\NP{\mathsf{NP}}
\newcommand\RE{\ensuremath{\mathsf{RE}}}

\DeclareMathOperator\qnoto{\to_{\operatorname{qno}}}
\DeclareMathOperator\Gaif{Gaif}

\DeclareMathOperator\ar{ar}
\DeclareMathOperator\Span{span}

\DeclareMathOperator\dist{dist}

\DeclareMathOperator\CSP{CSP}
\newcommand\YES{\mathrm{YES}}
\renewcommand\NO{\mathrm{NO}}
\DeclareMathOperator\qCSP{CSP_q}
\DeclareMathOperator\qnoCSP{CSP_q^{no}}

\DeclareMathOperator\qPol{qPol}
\DeclareMathOperator\qnoPol{qPol^{no}}
\DeclareMathOperator\qcPol{\ensuremath{\oplus}Pol}
\DeclareMathOperator\Pol{Pol}

\DeclareMathOperator\id{id}
\DeclareMathOperator\comp{comp}
\newcommand\rel[1]{\mathbb{#1}}
\DeclareMathOperator\Hom{Hom}

\newcommand\tuple[1]{\mathbf{#1}}
\newcommand\qto[1][\operatorname{q}]{\to_{#1}}
\newcommand\ignore[1]{}
\newcommand\qc[1]{\mbox{\ensuremath{\oplus{#1}}}}

\newcommand{\minion}[1]{{\mathscr #1}}

\DeclareMathOperator\minimalClone{\rel M}

\begin{document}

\maketitle
\begin{abstract}
\noindent
We introduce the concept of quantum polymorphisms to the complexity theory of quantum constraint satisfaction. Via this notion, we build an algebraic framework of reductions between quantum CSPs, and we establish a Galois connection between quantum polymorphism minions and quantum relational constructions. By leveraging a contextuality property of quantum polymorphisms, we fully characterise the existence of commutativity gadgets for relational structures, introduced by Ji as a method for achieving quantum soundness of classical CSP reductions. Prior to our work, only a partial classification was known for a subclass of Boolean languages
and for non-Boolean languages meeting specific structural conditions [Culf--Mastel, FOCS'25].
As an application of our framework, we prove that the quantum CSPs parameterised by odd cycles and the quantum CSP expressing quantum satisfiability of Siggers clauses are undecidable.
\end{abstract}

\setcounter{page}{0}\thispagestyle{empty}\clearpage

\newpage
\tableofcontents

\setcounter{page}{0}\thispagestyle{empty}\clearpage

\section{Introduction}

The \textit{quantum constraint satisfaction problem} is the decision problem of testing for the existence of a winning quantum strategy for a 2-player 1-round game where the two (cooperating, computationally unbounded) players aim to prove to a classical verifier that a system of constraints can be simultaneously satisfied. 
Here, a quantum strategy consists of a bipartite finite-dimensional quantum state shared by the players and a collection of measurements that the players perform on the state.
Upon receiving from the verifier questions involving local constraints of the system, the players measure their part of the state and answer with partial assignments based on the measurement outcomes. 
The origins of such game---often known as a \textit{non-local game}---come from physics: Bell's\cite{bell1964einstein} and Kochen--Specker's~\cite{kochen1967problem} proofs of the incompatibility of quantum mechanics with local
hidden-variable theories, stemming from  the Einstein--Podolsky--Rosen paradox~\cite{einstein1935can}, depend on the fact that, for certain types of constraints, quantum strategies generate correlations in the player's answers that are unattainable in classical physics.

Limiting the strategies to unentangled quantum states collapses the quantum CSP to its classical counterpart, the \textit{constraint satisfaction problem} (CSP).
The classical CSP is easily shown to be solvable in non-deterministic polynomial time. On the other hand, it was proved in~\cite{slofstra2019set} that the quantum CSP problem is undecidable---in fact,
as a consequence of the \MIP=\RE\ theorem~\cite{ji2021mip}, it remains undecidable even in its \textit{gapped} version.
In both the classical and quantum settings, restricting the constraint systems in the game to conjunctions of clauses belonging to some specific language can make the problem strictly easier. 
For example, for classical CSPs, if all constraints are XOR or  Horn-SAT clauses, the problem is easily seen to be in $\mathsf{P}$, while if they are expressive enough to capture $3$-colouring or $3$-SAT, the problem is $\NP$-complete. Historically, the apparent lack of problems of this type having intermediate complexity  sparked the hypothesis that such dichotomy could, in fact, extend to the whole class of \textit{finite-domain} CSPs.\footnote{The complexity behaviour of infinite-domain CSPs is much wilder (see, e.g.~\cite{SymmetriesNotEnough}). However, a dichotomy has been conjectured for a large class of infinite-domain CSPs with a finite description~\cite{BPP}; see~\cite{SmoothApproximations} for a recent article about the progress on the conjecture.}
`$\operatorname{P}$ vs.\ $\NP$-complete' dichotomies for fragments of CSPs started to emerge, primary examples being Schaefer's~\cite{Schaefer78:stoc} and Hell--Ne\v{s}et\v{r}il's~\cite{HellN90} dichotomies for Boolean and undirected graph CSPs, respectively. 
However, a \textit{unified} framework for approaching the \textit{CSP Dichotomy Conjecture}~\cite{Feder98:monotone} only emerged at the turn of the century, with the foundational line of work initiated by Jeavons, Cohen, and Gyssens that became known as the \textit{(universal-)algebraic approach to CSPs}~\cite{Jeavons97:closure,Jeavons98:algebraic}.
By lifting CSP theory from the realm of Turing machines and reductions to that of algebra, the study of \textit{polymorphism clones}---universal-algebraic objects capturing high-order invariance properties of CSP solutions---provided a clearer viewpoint on the CSP complexity landscape, and it eventually culminated with the positive resolution of the conjecture by Bulatov~\cite{Bulatov17:focs} and, independently, Zhuk~\cite{Zhuk17_FOCS,Zhuk20:jacm} in 2017. 
In the case of quantum CSPs, the complexity landscape is much less understood. While there are various examples of constraint classes for which the quantum CSP is undecidable~\cite{slofstra2019set,ji2021mip,harris2024universality,AtseriasKS19,culf2024re,paddock2025satisfiability,Zeman} or polynomial-time solvable~\cite{kempe2010unique,Ji, AtseriasKS19,ciardo_quantum_minion,paddock2025satisfiability,bulatov2025satisfiability,levene2026unique}, the complexity is unknown for most classes of languages, due to the lack of a general approach for studying them with unified tools.

The main conceptual contribution of this work is the introduction of the notion of \textit{quantum polymorphisms} as a tool to lift the study of algebraic invariants of classical CSP languages to the quantum setting. Via this notion, we initiate an algebraic complexity theory of quantum CSPs modelled on the algebraic approach initiated in~\cite{Jeavons97:closure,Jeavons98:algebraic} and culminated 20 years later with the proof of the (classical) Dichotomy Theorem.

Classical polymorphisms turned out to be the right tool for exploring CSP complexity at the algebraic level. The root of this phenomenon is a connection between relational constructions (expressing the CSP language) and operation systems (expressing the family of polymorphisms). Via this connection, reducibility between CSPs associated with different languages can be explained via inclusion (or algebraic forms of the latter, in particular, so called minion homomorphisms) between their polymorphism sets.
In this work, we prove that quantum polymorphisms completely determine the complexity of quantum CSPs, thus enabling the use of algebra in the complexity analysis of such problems.
Moreover, we link quantum polymorphisms to a quantised version of classical relational constructions that we name \textit{q-definitions} and \textit{q-constructions}, obtained by lifting the classical model-theoretic counterparts to the setting of quantum strategies.

Broadly, computable reductions in the literature on non-local games and quantum CSP complexity typically follow a common conceptual pattern: taking a classical CSP reduction (in particular, a so-called \textit{subdivision} of constraints or a more general pp-definition), and trying to achieve quantum completeness and soundness. 
It is well known that this approach does not work in general. A substantial obstacle to the prospect of using classical gadgets as reductions between quantum CSPs is ultimately caused by the following, non-classical phenomenon: replacing constraints by gadgets increases the size of the contexts---which, in general, destroys simultaneous measurability (see, e.g.,~\cite{ji2021mip,paddock2025satisfiability,AtseriasKS19}). On the other hand, for some families of constraint languages, the issue can be solved by pairing classical gadgets with an extra type of gadget known as a \textit{commutativity gadget}. On a high level, these are specific types of gadgets 
that can be plugged into classical gadgets to enforce simultaneous measurability of the final instance resulting from the reduction. For example, by showing that a specific \textit{triangular prism} construction yields a commutativity gadget for the 3-clique, Ji showed that the classical polynomial-time reduction from 3-SAT to 3-colouring can be lifted to a reduction between the corresponding quantum CSPs~\cite{Ji}.
These commutativity gadgets have no counterpart in the classical setting, as they enforce (by design) no additional constraints on the variables in their scope apart from commutativity. Therefore, they are useless from the classical point of view.

Later works generalised this idea to other types of quantum CSPs. In particular, commutativity gadgets for certain Boolean CSPs were constructed in~\cite{AtseriasKS19}. In a subsequent recent breakthrough~\cite{culf2024re}, Culf and Mastel designed commutativity gadgets for all Boolean constraint languages whose classical CSP is NP-complete (or, more generally, is not preserved under majority~\cite{Schaefer78:stoc}).
In the same work, the authors identified a \textit{sufficient} property of CSP languages---which they called \textit{two-variable (non-)falsifiability}---that guarantees the existence of commutativity gadgets for the corresponding CSPs and thus allows quantising classical reductions. Very recently, Culf, de Bruyn, Vernooij, and Zeman~\cite{Zeman} gave a different, \textit{necessary} condition for the existence of commutativity gadgets, by showing that CSPs with non-classical quantum endomorphisms do not admit commutativity gadgets. 
Despite this recent progress, a full characterisation of CSPs admitting commutativity gadgets was hitherto not available.

We prove that Ji's reductions obtained by pairing classical gadgets with commutativity gadgets are captured by (and are a strict subset of) the reductions coming from quantum polymorphisms. Hence,
as a second main contribution of this work, we give a complete characterisation of the class of CSP languages admitting commutativity gadgets, in terms of a \textit{contextuality} property of their quantum polymorphisms.
In particular, we show that a commutativity gadget exists if, and only if, the minion of quantum polymorphisms coincides with the \textit{quantum closure} of the clone of standard polymorphisms. 
Through this characterisation, we prove that all odd cycles, as well as the Siggers digraph~\cite{kearnes_optimal_2014}, admit commutativity gadgets and, thus, the corresponding quantum CSP is undecidable.\footnote{This answers an open question raised in~\cite{Zeman}.} In addition, we show that whenever a language admits a commutativity gadget, this can be chosen to be a high-enough categorical power of the language itself. We use this to build commutativity gadgets for Boolean structures in a unified manner, and we extend the result in~\cite{culf2024re} to a full classification of Boolean structures admitting commutativity gadgets. As a consequence, we obtain a complexity dichotomy for Boolean quantum CSPs, by showing that they are either in P or undecidable.

\section{Overview of results and techniques}
\label{sec_overview}
The gist of the algebraic approach to CSPs is a correspondence between
$(i)$ algebraic properties of invariants of solution sets of CSPs---namely, polymorphisms---and $(ii)$ relational constructions that enable reductions between different CSPs and, thus, determine tractability or hardness. In
the next three subsections of this overview, 
we lift both objects to the quantum setting, and we use them to characterise the existence of commutativity gadgets. We then apply the polymorphic framework to prove undecidability for classes of quantum CSPs by using a result on the emergence of certain non-orthogonality patterns in contextual quantum polymorphisms.  This section contains an informal overview of the main ideas we use to establish these results. Full details are in the body of the paper (from~\Cref{sec_prelimns} on).

\subsection{Quantum algebraic reductions}
\label{subsec_overview_algebraic_reductions}

Consider a 2-prover 1-round game as illustrated in the Introduction.
Suppose that the admitted constraints in the protocol are encoded in a relational structure $\A$; i.e., the goal of the provers is to convince the verifier that a given system $\X$ (called the \textit{instance}) of $\A$-clauses applied to variables of $\X$ is simultaneously satisfiable. The problem of deciding whether a winning quantum strategy exists for such game is the quantum CSP parameterised by the relational structure $\A$, and it is denoted by $\qCSP(\A)$.
We define a quantum polymorphism of $\A$ as a winning quantum strategy for such protocol, where the instance $\X$ is constrained to be some $n$-ary categorical power $\A^n$ of $\A$.\footnote{If $\A$ is a digraph (i.e., it has a single, binary relation), then $\A^n$ is the standard direct power of digraphs.} In the language of \textit{quantum homomorphisms} between relational structures~\cite{abramsky2017quantum}, this can be rephrased
as follows.

\vspace{-.25cm}
\begin{shaded}
\vspace{-.45cm}
\begin{definition}
\label{defn_quantum_polymorphisms}
    A \emph{quantum polymorphism} of arity $n\in\N$ for a structure $\A$ is a quantum homomorphism $Q\colon\rel A^n\qto\rel A$. We let $\qPol(\A)$ be the set of all quantum polymorphisms of $\A$.
\end{definition}
\vspace{-.45cm}
\end{shaded}
\vspace{-.25cm}

By forcing the admitted strategies in the game to be classical and deterministic, we retrieve the classical polymorphisms of $\A$; i.e., $\Pol(\A)\subseteq\qPol(\A)$. 
The role of classical polymorphisms in CSP complexity is well understood: $\Pol(\A)$ completely determines the complexity of $\CSP(\A)$, in the sense that, given two langauges $\A$ and $\B$ such that $\Pol(\A)=\Pol(\B)$,  the two problems $\CSP(\A)$ and $\CSP(\B)$ are polynomial-time equivalent. 
Moreover, not all information contained in $\Pol(\A)$ is relevant complexity-wise. In fact, the complexity of $\CSP(\A)$ is fully determined by the identities of a restricted form satisfied by $\Pol(\A)$, namely the so-called \textit{minor} identities.
This leads to the notion of a \textit{function minion} (or \textit{minor-closed function class}),
which provides an algebraic tool for comparing the complexity of CSPs across different languages. Function minions capture precisely the minor identities satisfied by polymorphisms, thereby isolating the complexity-relevant structure of $\Pol(\A)$.

\begin{theorem}[\cite{BOP18}]
\label{thm_classical_minion_homo_reduction}
    Let $\A$ and $\B$ be relational structures. If there exists a minion homomorphism from $\Pol(\A)$ to $\Pol(\B)$, then $\CSP(\B)$ is logspace-reducible to $\CSP(\A)$.
\end{theorem}

Quantum polymorphisms are not operations over the domain $A$ and, thus, cannot be directly equipped with the structure of a function minion. Nevertheless, we show that they \textit{behave algebraically} like classical polymorphisms, by endowing them with the structure of \textit{abstract minions}. These objects, originally introduced in the contexts of relaxations of \textit{promise CSPs}~\cite{BBKO21,cz23sicomp:clap,cz23soda:minions,CiardoZassociation,cz23stoc:ba,BrakensiekGS23,larrauri2025ineffectiveness}, are designed to retain the algebraic structure of function minions while dispensing from the requirement that their elements be concrete functions. 
To establish that $\qPol(\A)$ is an abstract minion, we exploit the fact that quantum homomorphisms admit a specific composition operation, described in~\cite{abramsky2017quantum} using the framework of Kleisli categories  (see~\Cref{subsec_body_algebraic_structure}). This structure enables an algebraic comparison of quantum polymorphisms across different languages. In particular, we lift~\Cref{thm_classical_minion_homo_reduction} to the following result.

\vspace{-.25cm}
\begin{shaded}
\vspace{-.45cm}
\begin{restatable}{theorem}{thmMainqPolHomoReductions}
\label{thm_main_qPol_homo_reductions}
    Let $\A$ and $\B$ be relational structures. If there exists a minion homomorphism from $\qPol(\A)$ to $\qPol(\B)$, then $\qCSP(\B)$ is logspace-reducible to $\qCSP(\A)$.
\end{restatable}
\vspace{-.45cm}
\end{shaded}
\vspace{-.25cm}

The proof of~\Cref{thm_main_qPol_homo_reductions} (whose full details are given in~\Cref{subsec_body_reductions_via_minion_homo}) is based on a reduction coming from the Long-Code testing for classical constraint systems~\cite{bellare1998free} (see also~\cite[\S3.3]{BBKO21}). We establish quantum completeness and soundness of this reduction using the minor-preservation property induced by the minion homomorphism. 
As a consequence of the above result, quantum polymorphisms---and, more precisely, the minor identities they satisfy---completely characterise the complexity of quantum CSPs. 
Moreover,~\Cref{thm_main_qPol_homo_reductions} enables the transfer of known undecidability results for specific quantum CSPs (for instance, quantum 3SAT~\cite{ji2021mip,mastel2024two} or quantum XOR~\cite{slofstra2019set}) to a much broader class of problems via the analysis of the corresponding quantum polymorphisms.
This draws a nice parallel with the classical setting, where the algebraic study of polymorphisms enabled the use of tools from universal algebra, ultimately playing a central role in the proof of the CSP dichotomy theorem. The next step is to look at the relational side of the picture.
\subsection{Quantum relational constructions}
\label{subsec_overview_relational_reductions}
The fact that a homomorphism $\xi$ between (classical) polymorphism minions $\Pol(\A)$ and $\Pol(\B)$ gives rise to reductions can be viewed in two distinct ways. At the level of the instance, $\xi$ allows a complete and sound reduction from $\CSP(\B)$ to $\CSP(\A)$ by ensuring that satisfying assignment to the first problem correspond to satisfying assignments of the second. At the level of the language, $\xi$ captures the fact that $\A$ can \textit{simulate} $\B$ via a specific type of uniform construction. These two sides are tied together via strong categorical links (see, for example,~\cite{hadek2025categorical}). We now look at the latter side, in the case of quantum CSPs. 

Classically, the right notion of simulation for CSPs is based on \textit{primitive positive (pp) formulae}---first-order formulae using relations in the language, equality, conjunction, and
existential quantification---a concept that originates in model theory and universal algebra. Lifting this concept directly to the case of quantum CSPs presents a seriour obstacle: It is well known that, in general, reductions between classical CSPs based on pp-formulae are not complete and sound in the quantum setting. This is due to the fact that they alter the commutativity contexts of the instance, see e.g.~\cite{Ji,AtseriasKS19,culf2024re}. For this reason, we lift the notion of pp-formula to a quantum counterpart that provably works as a reduction between quantum CSPs. This results in the notion of \textit{q-definition}, which we formally describe in~\Cref{subsec_body_quantum_reductions}. At a high level,
this notion should be viewed as the ``minimal set of conditions required to quantise classical pp-definitions of constraint languages''. One limitation of q- (and pp-) definitions is that any structure that is definable from $\A$ must have domain $A$. In order to be able to reduce between quantum CSPs having distinct domains, we introduce the more general notion of \textit{q-constructions}, modelled around the classical pp-constructions.

The next step is to show that q-definitions and the more expressive q-constructions perform the required tasks of simulating a quantum CSP within another. To that end, we prove that q-constructions induce minion homomorphisms of the corresponding quantum polymorphism minions (and thus, via~\Cref{thm_main_qPol_homo_reductions}, reductions between the corresponding quantum CSPs). 

\vspace{-.25cm}
\begin{shaded}
\vspace{-.45cm}
\begin{restatable}{theorem}{qconstructionimpliesqPolhom}
\label{q_construction_implies_qPol_hom}
    If $\A$ q-constructs $\B$, then there exists a minion homomorphism $\qPol(\A)\to\qPol(\B)$.
\end{restatable}
\vspace{-.45cm}
\end{shaded}
\vspace{-.25cm}

Note that, in the statement above, neither the domains nor the signatures of $\A$ and $\B$ are required to coincide. 
The key to proving the above result is to show that it is possible to perform a ``vectorisation step'' identifying tuples of tuples with vectors (see condition $(ii)$ in~\Cref{defn_q_construction}) at the level of minion homomorphisms of quantum polymorphism minions (and, thus, of reductions between quantum CSPs). While for classical pp-constructions this fact is straightforward, it is not a priori clear that this operation does not affect completeness and soundness in the quantum case. One might expect that an extra commutativity gadget (see later) is needed specifically to quantise this part of the classical construction. We prove that, in fact, no extra gadget is necessary for this particular step to work. To this end,  we make use of a minion-theoretic observation on the extension of partial minion homomorphisms defined on elements having only essential coordinates to global minion homomorphisms.

In the case of q-definitions, we additionally prove a quantum counterpart of the classical operation/relation Galois connection for pp-definitions~\cite{bodnarchuk1969galois,geiger1968closed}. 
\vspace{-.25cm}
\begin{shaded}
\vspace{-.45cm}
\begin{restatable}{theorem}{propgaloisoracular}
\label{prop_galois_oracular}
    Let $\A$ and $\B$ be relational structures over the same domain. Then $\A$ q-defines $\B$ if, and only if, $\qPol(\A)\subseteq\qPol(\B)$. 
\end{restatable}
\vspace{-.45cm}
\end{shaded}
\vspace{-.25cm}

\subsection{Non-contextual quantum polymorphisms and commutativity gadgets}
\label{subsec_overview_commutativity_gadgets}

We have, at this point, two perspectives on reductions between quantum CSPs---one algebraic, and one relational.
The next step is to obtain a direct link between such reductions and their classical counterparts, with the goal of achieving a systematic method for lifting classical hardness reductions to quantum undecidability reductions.
To that end, we make use of
the notion of commutativity gadgets.
At a high level, the idea is that, while every q-definition is in particular a pp-definition, 
the converse is not true in general, as the commutativity contexts of the measurements of the quantum strategy for the reduced instance are in general unrelated to those for the initial instance.
Commutativity gadgets were introduced in~\cite{Ji} to overcome this obstacle for certain constraint languages, by artificially introducing the required commutativity relations between the PVMs.

Recall that quantum strategies for a non-local game consist of projection-valued measurements (PVMs) with outcome set given by the language $\A$. In general,  the correlation associated with the measurements of a quantum strategy may locally satisfy all constraints with perfect probability even if such correlation cannot be associated with any system of measurements simultaneously performed on all constraints. This is ultimately a phenomenon of quantum contextuality, occurring whenever the PVMs of the strategy do not give rise to a global context on the whole instance---i.e., whenever some of the PVMs do not commute with each other.
This phenomenon can also occur for quantum polymorphisms: we say that a quantum polymorphism $Q\colon\A^n\qto\A$ is \textit{contextual} if at least one pair $(Q_1,Q_2)$ of PVMs in $Q$ does not commute. If, on the other hand, all PVMs commute, we say that $Q$ is \textit{non-contextual}. As a straightforward application of the spectral theorem, non-contextual quantum polymorphisms can be equivalently viewed as probability distributions over classical deterministic polymorphisms. %
A main conceptual contribution of our work is to show that the existence of a commutativity gadget for a language can be fully characterised by the non-contextuality of its quantum polymorphisms, as stated next.

\vspace{-.25cm}
\begin{shaded}
\vspace{-.45cm}
\begin{theorem}\label{gadget-characterization_overview_friendly}
Let $\rel A$ be a relational structure. The following are equivalent:
\begin{enumerate}
    \item $\rel A$ has a commutativity gadget.
    \item $\rel A^n$ is a commutativity gadget for all large enough $n\in\N$.
    \item All quantum polymorphisms of $\A$ are non-contextual.
\end{enumerate}
\end{theorem}
\vspace{-.45cm}
\end{shaded}
\vspace{-.25cm}

In order to establish the 
$3.\Rightarrow 2.\Rightarrow 1.$ implications in the
theorem above, we use the fact that any pair of elements in $A$ can be \textit{generated} as the image of some polymorphism of high-enough arity $n$. Then, we prove that, for such $n$, the $n$-th categorical power of $\A$ yields a commutativity gadget for $\A$ provided that all quantum homomorphisms $\A^n\qto\A$ (i.e., all $n$-ary quantum polymorphisms of $\A$) are non-contextual. In order to prove the $1.\Rightarrow 3.$ implication, the high-level idea is to start with a quantum polymorphism of $\A$, and compose it with a tensor product of quantum homomorphisms from the 
commutativity gadget to $\A$. 
To make this step work, we make use of the composition operation introduced in~\cite{abramsky2017quantum}
to prove that quantum homomorphisms can be expressed as Kleisli morphisms for a graded monad.

Via the characterisation of commutativity gadgets in~\Cref{gadget-characterization_overview_friendly}, we are able to connect classical pp-constructions to both the algebraic reductions in~\Cref{subsec_overview_algebraic_reductions}, and the relational constructions from~\Cref{subsec_overview_relational_reductions}. As a consequence, we prove that if $\A$ pp-constructs $\B$ and all quantum polymorphisms of $\A$ are non-contextual, then $\A$ q-constructs $\B$ and, thus, $\qCSP(\B)$ reduces to $\qCSP(\A)$. In particular, this gives the following. 

\vspace{-.25cm}
\begin{shaded}
\vspace{-.45cm}
\begin{restatable}{theorem}{cornpcompleteplusnocontextualmeansundecidable}
\label{cor_np_complete_plus_no_contextual_means_undecidable}
    Let $\A$ be a structure such that $\CSP(\A)$ is not in $\mathsf{P}$ and all quantum polymorphisms of $\A$ are non-contextual. Then, $\qCSP(\A)$ is undecidable.
\end{restatable}
\vspace{-.45cm}
\end{shaded}
\vspace{-.25cm}

This result yields a broadly applicable undecidability criterion for quantum CSPs.\footnote{We point out that the condition $\CSP(\A)\not\in\mathsf{P}$ in~\Cref{cor_np_complete_plus_no_contextual_means_undecidable} can be replaced by the algebraic condition of $\Pol(\A)$ being trivial (or, equivalently, by $\A$ pp-constructing, say, 3SAT). In this way, the result does not trivialise if $\mathsf{P}=\NP$; see~\Cref{cor_np_complete_plus_no_contextual_means_undecidable_ALGEBRAIC}.} In particular, it captures and extends the known undecidability results established in~\cite{AtseriasKS19} and in~\cite{culf2024re}: We show in~\Cref{subsec_comparison_TVF} that any relational structure that does not meet the so-called \textit{two-variable falsifiability (TVF)} property (and, thus, is captured by the undecidability criteria  of~\cite{AtseriasKS19,culf2024re}) cannot have contextual quantum polymorphisms
and, thus, is captured by~\Cref{cor_np_complete_plus_no_contextual_means_undecidable}. As we shall see, the latter result is strictly more expressive, as it applies to various classes of CSP languages that are not captured by~\cite{AtseriasKS19,culf2024re}. 

\subsection{Contextuality and bifurcations}
\label{subsec_overview_contextuality_bifurcations}

A consequence of the machinery developed so far is that we are now able to systematically investigate the complexity landscape of quantum CSPs  by studying the properties of their quantum polymorphisms---in particular, whether they are non-contextual.
Recall that we defined a quantum homomorphism to be contextual when at least two of the measurements of the corresponding quantum strategy do not commute and, thus, cannot be simultaneously performed in at least one reference system.
Before proceeding further, it shall be useful to address the following question: \textit{What information can we derive from the fact that a quantum homomorphism is contextual}?

Informally, the idea here is to ``trade non-commutativity for non-orthogonality''.
At first glance, this trade appears unfavourable as any two non-commuting projectors are, in particular, non-orthogonal. However, it turns out that a \textit{single} instance of non-commutativity among projectors arising in a perfect quantum strategy generates a structured pattern of \textit{multiple} non-orthogonality relations. We call such non-orthogonality patterns \textit{(contextuality) bifurcations}. (For the formal definition, see~\Cref{sec_bifurcations}; two illustrations of such patterns are given in~\Cref{fig_two_bifurcation} and~\Cref{fig_omega_bifurcation}.)
Specifically, we prove the following.

\vspace{-.25cm}
\begin{shaded}
\vspace{-.45cm}
\begin{restatable}{theorem}{thmbifurcations}
\label{thm_bifurcations}
    Let $\X$ and $\A$ be relational structures such that $\Gaif(\X)$ is a connected graph, and let $Q\colon\X\qto\A$ be a contextual quantum homomorphism. Then there exists a bifurcation for $Q$ whose length is at most the diameter of $\Gaif(\X)$.
\end{restatable}
\vspace{-.45cm}
\end{shaded}
\vspace{-.25cm}

In practice, the strategy for applying~\Cref{cor_np_complete_plus_no_contextual_means_undecidable} to show that the quantum polymorphisms of a structure $\A$ are non-contextual (and hence that $\A$ admits commutativity gadgets) proceeds as follows. First, one observes that the support of any bifurcation defines a substructure $\GG \subseteq \A^n \times \A$, since $Q$ is a perfect strategy. Next, one shows that the existence of a contextual quantum polymorphism would yield sufficiently many bifurcations (or bifurcations of a suitable form) to force $\GG$ to contain edges not present in $\A^n \times \A$, leading to a contradiction.

\begin{figure}
\begin{center}
    \includegraphics[width=\textwidth]{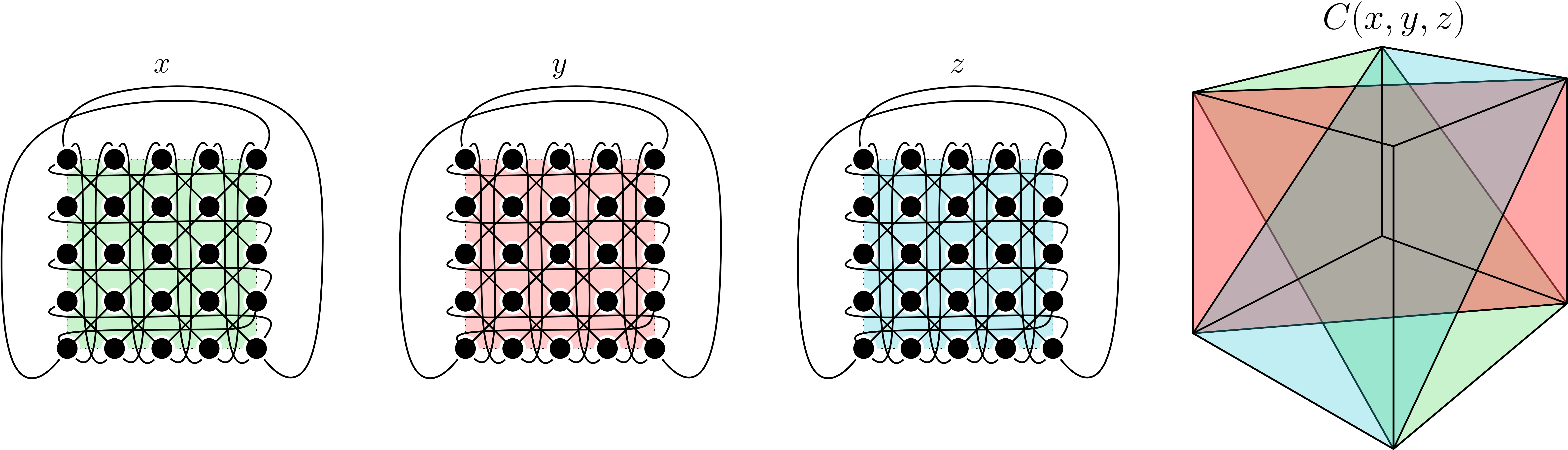}
    \caption{A reduction proving undecidability of $\qCSP(\rel C_5)$ by a reduction from the ternary Boolean 1-in-3 predicate. Each variable of an instance of 1-in-3-SAT is replaced by a cloud of $5^2$ vertices as shown, isomorphic to the graph $(\rel C_5)^2$.
    Each constraint $C(x,y,z)$ is then replaced by a cloud of $5^3$ vertices inducing the graph $(\rel C_5)^3$, and the variable-clouds are identified with subsets of the constraint-cloud as depicted in the image.}
    \label{fig:reduction-C5}
    \end{center}
\end{figure}

\subsection{Quantum polymorphisms for Siggers and odd cycles}
\label{subsec_overview_siggers_and_cycles}

The machinery developed so far allows transferring undecidability across classes of quantum CSPs, in a similar way as NP-hardness propagates across classical CSPs via gadget reductions associated with pp-definitions and pp-constructions. To showcase this strategy, we prove undecidability of the quantum CSPs parameterised by two simple types of languages.
In~\Cref{sec_odd_cycles}, we prove the following.

\vspace{-.25cm}
\begin{shaded}
\vspace{-.45cm}
\begin{restatable}{theorem}{thmquantumoddcyclesisundecidable}
\label{thm_quantum_odd_cycles_is_undecidable}
    Let $\rel C_m$ be the undirected cycle of size $m$. Then $\qCSP(\rel C_m)$ is undecidable for each odd $m\geq 3$.
\end{restatable}
\vspace{-.45cm}
\end{shaded}
\vspace{-.25cm}
This resolves an open question raised in~\cite{Zeman}, where it was mentioned that the known techniques for proving undecidability of quantum CSPs do not capture this class of graphs.

The \textit{Siggers digraph} is the $3$-vertex digraph obtained by making one edge of a directed triangle undirected. 
The prominence of this language in algebraic CSP theory lies in the fact that the corresponding \textit{loop condition} gives rise to \textit{Siggers-identities}~\cite[Theorem 2.2]{kearnes_optimal_2014}, which are known to exactly capture the borderline between tractability and hardness for classical CSPs~\cite{BKW17} (like other types of similar identities, see e.g.~\cite{kozik2015characterizations}). In~\Cref{sec_siggers}, we prove the following.

\vspace{-.25cm}
\begin{shaded}
\vspace{-.45cm}
\begin{restatable}{theorem}{thmsiggersundecidable}
\label{thm_siggers_undecidable}
    Let $\A$ be the Siggers digraph. Then $\qCSP(\rel A)$ is undecidable.
\end{restatable}
\vspace{-.45cm}
\end{shaded}
\vspace{-.25cm}

The proofs of these results, given in full detail in the body of the paper, are based on establishing that the quantum polymorphisms of such languages are non-contextual, via a combination of the bifurcation result (\Cref{thm_bifurcations}) and ad-hoc arguments exploiting the particular structures of the languages. Then, undecidability follows from~\Cref{cor_np_complete_plus_no_contextual_means_undecidable}.

\subsection{Dichotomy for Boolean languages}
\label{subsec_overview_boolean}

As a third application of our framework to concrete classes of languages, we consider the case of Boolean structures, and we obtain a complete classification of the Boolean structures admitting a commutativity gadget.

\vspace{-.25cm}
\begin{shaded}
\vspace{-.45cm}
\begin{restatable}{theorem}{thmbooleancommgadgetclassification}
\label{thm_boolean_comm_gadget_classification}
    Let $\A$ be a Boolean relational structure. Then the following are equivalent:
    \begin{itemize}
        \item[$(i)$] $\A$ admits a commutativity gadget;
        \item[$(ii)$] $\A$ has no majority polymorphism or is not two-variable falsifiable.
    \end{itemize}
\end{restatable}
\vspace{-.45cm}
\end{shaded}
\vspace{-.25cm}

We point out that the $(ii)\Rightarrow(i)$ implication was shown in~\cite{culf2024re}; here, we provide a different proof based on the analysis of quantum polymorphisms of Boolean structures. We show that the polymorphisms of 1-in-3-SAT, are non-contextual and then use quantum compatible pp-definitions to show this for other Boolean structures. We note that some of these arguments make use of
\textit{non-oracular} quantum polymorphisms, see~\Cref{subsec_overview_non_oracular}.

To obtain the converse implication $(i)\Rightarrow(ii)$, we identify a structure $\minimalClone$ whose
quantum polymorphism minion
$\qPol(\minimalClone)$
has the property of being contained in the quantum polymorphism minion of any Boolean structure that is invariant under majority and is two-variable falsifiable, see~\Cref{fig:O_b^4} and~\Cref{thm_DM_quantum_minimality}. Then, we prove that the minion $\qPol(\minimalClone)$ is contextual, as it contains contextual quantum polymorphisms of any arity $\geq 4$. On the other hand, we show that  $\minimalClone$  does not admit contextual quantum polymorphisms of any arity $n\leq 3$ (see~\Cref{thm_DM_contextuality}). In particular, this means that all quantum endomorphisms of $\minimalClone$ are non-contextual. Hence, the necessary condition for the existence of commutativity gadgets given in~\cite{Zeman} is provably insufficient to establish that Boolean TVF structures that are preserved by majority cannot admit a commutativity gadget, and the full power of our quantum polymorphism characterisation is needed.

\begin{figure}
    \centering
    \begin{tikzpicture}[scale=.6,
  every node/.style={circle, draw, minimum size=6mm, inner sep=1pt, font=\tiny},
  >=stealth, ->
]


\node (0000) at (0,0) {0000};

\node[fill=yellow!80] (1000) at (-3,2) {1000};
\node[fill=yellow!60]   (0100) at (-1,2) {0100};
\node[fill=yellow!40] (0010) at (1,2)  {0010};
\node[fill=yellow!20] (0001) at (3,2)  {0001};

\node[fill=red!20]   (1100) at (-5,4) {1100};
\node[fill=green!20] (1010) at (-3,4) {1010};
\node[fill=blue!20]  (1001) at (-1,4) {1001};
\node[fill=blue!20]  (0110) at (1,4)  {0110};
\node[fill=green!20] (0101) at (3,4)  {0101};
\node[fill=red!20]   (0011) at (5,4)  {0011};

\node[fill=yellow!20] (1110) at (-3,6) {1110};
\node[fill=yellow!40] (1101) at (-1,6) {1101};
\node[fill=yellow!60]   (1011) at (1,6)  {1011};
\node[fill=yellow!80] (0111) at (3,6)  {0111};

\node (1111) at (0,8) {1111};

\foreach \y in {1000,0100,0010,0001} \draw (0000)--(\y);

\draw (1000)--(1100);
\draw (1000)--(1010);
\draw (1000)--(1001);
\draw (0100)--(1100);
\draw (0100)--(0110);
\draw (0100)--(0101);
\draw (0010)--(1010);
\draw (0010)--(0110);
\draw (0010)--(0011);
\draw (0001)--(1001);
\draw (0001)--(0101);
\draw (0001)--(0011);

\draw (1100)--(1110);
\draw (1100)--(1101);
\draw (1010)--(1110);
\draw (1010)--(1011);
\draw (1001)--(1101);
\draw (1001)--(1011);
\draw (0110)--(1110);
\draw (0110)--(0111);
\draw (0101)--(1101);
\draw (0101)--(0111);
\draw (0011)--(1011);
\draw (0011)--(0111);

\foreach \x in {1110,1101,1011,0111} \draw (\x)--(1111);
\end{tikzpicture}
    \caption{Illustration of the $4$-th power of the structure $\minimalClone$ yielding the minimal contextual minion. The proof of the $(i)\Rightarrow(ii)$ implication of~\Cref{thm_boolean_comm_gadget_classification} relies on showing that this structure admits a contextual quantum homomorphism to $\minimalClone$; see~\Cref{subsec_body_minimal_contextual_clone} for full details.}
    \label{fig:O_b^4}
\end{figure}

Combining~\Cref{thm_boolean_comm_gadget_classification} with results from~\cite{slofstra2019set,AtseriasKS19}, we obtain the following complexity dichotomy for quantum CSPs parameterised by Boolean languages. We denote by $\rel{XOR}$ the structure encoding Boolean linear systems of equations.

\vspace{-.25cm}
\begin{shaded}
\vspace{-.45cm}
\begin{restatable}{corollary}{corcomplexitydichotomyboolean}
\label{cor_complexity_dichotomy_boolean}
    Let $\A$ be a Boolean relational structure. If $\A$ pp-defines $\rel{XOR}$, $\qCSP(\A)$ is undecidable; otherwise, $\qCSP(\A)$ is solvable in polynomial time.
\end{restatable}
\vspace{-.45cm}
\end{shaded}
\vspace{-.25cm}
We refer the reader to~\Cref{remark_comparison_boolean_paddock_slofstra} for a comparison of this result with the Boolean dichotomies given in~\cite[Theorem~5.11 (a)]{paddock2025satisfiability} and in~\cite[\S6.1]{AtseriasKS19}.

\subsection{Non-oracular quantum polymorphisms}
\label{subsec_overview_non_oracular}

So far, we have only considered the oracular setting for quantum homomorphisms. 
At the level of two-player games, this corresponds to the so-called \textit{constraint-constraint} or \textit{constraint-variable} formulations of the CSP game (see, for example, as opposed to the \textit{variable-variable} variant where the cooperating players both receive as questions two variables of the instance). At the level of quantum homomorphisms, this corresponds to requiring that the projectors in a winning strategy associated with variables in the same constraint (or \textit{context}) should commute. 
We believe this is the natural setting for establishing a theory of reductions between quantum CSPs, for two reasons. First, categorical interpretations of quantum CSPs---primarily, the quantum-monad formulation of~\cite{abramsky2017quantum}---appear to naturally capture the oracular setting for CSPs of arbitrary arities.
Second, in the gapped version of quantum CSPs, the definition of the quantum value of a game in terms of the expected value of the trace of products of projectors in a randomly sampled constraint requires commutativity between such projectors in order for the trace to be a real number, if the arity is at least $3$; see, for example, the discussion in~\cite[\S3.1]{mousavi2025quantum}. On the other hand, in the case of binary CSPs, both notions are natural, and they are known to sometimes give rise to different behaviours.
For example, the non-oracular Unique Games are provably in P~\cite{kempe2010unique}, while it is plausible that the oracular version is undecidable~\cite{mousavi2025quantum}; see also~\cite{karamlou2025quantum}.
See also~\cite[\S4.3]{abramsky2017quantum}, where non-oracular quantum graph homomorphisms in the sense of~\cite{MancinskaR16} are shown to be captured via oracular homomorphisms for specific Boolean constraint systems.

We show that the machinery illustrated so far extends to a large extent to the non-oracular setting. In particular, we define the non-oracular version of quantum polymorphisms, and we show that it precisely captures the existence of commutativity gadgets for non-oracular quantum CSPs in exactly the same way as in the oracular case; see~\Cref{gadget-characterization_non_oracular}. Furthermore, we extend the notion of q-definition to the non-oracular case, and we prove that the relation/operation Galois connection of~\Cref{prop_galois_oracular} extends to the non-oracular setting.

\vspace{-.25cm}
\begin{shaded}
\vspace{-.45cm}
\begin{restatable}{theorem}{corgaloisnonoracular}
\label{cor_1308_2111}
Two structures $\rel A,\rel B$ satisfy $\qnoPol(\rel A)\subseteq\qnoPol(\rel B)$ if, and only if, $\rel A$ q-no-defines $\rel B$.
\end{restatable}
\vspace{-.45cm}
\end{shaded}
\vspace{-.25cm}
We refer to~\Cref{sec_nonoracular} for more details on the properties of quantum non-oracular polymorphisms, and the reductions that are captured by them.

\subsection{Outlook}
\label{subsec_overview_outlook}
Quantum polymorphisms provide a bridge between two types of methods: $(i)$ reductions from the classical algebraic CSP theory, and $(ii)$ the quantum-specific techniques used in the complexity analysis of non-local games---in particular, commutativity gadgets. It is, in our view, quite surprising that both $(i)$ and $(ii)$ can be captured within a single formalism, given the different contexts from which they originate. Establishing this link is a key conceptual contribution of the paper and plays a crucial role in our undecidability results for concrete classes of languages.
In more detail, quantum polymorphisms are ultimately encodings of specific quantum strategies; as such, they exhibit, in general, non-classical behaviours---specifically, non-commutativity of measurements across different contexts. This is precisely what enables them to capture (and characterise the existence of) commutativity gadgets. At the same time, quantum polymorphisms admit the same algebraic structure as their classical counterparts, namely that of minions, and can thus be related to other such structures on a purely algebraic ground.
It is this dual nature that makes quantum polymorphisms able to 
transfer undecidability across different problems over different languages---simulating the way NP-completeness is transferred across classical CSPs.

Among the open directions
arising from this paper, we mention the prospect of obtaining a stronger contextuality bifurcation theorem than~\Cref{thm_bifurcations}, resulting in more complex non-orthogonality patterns generated by contextual homomorphisms. This would allow applying our undecidability criterion~\Cref{cor_np_complete_plus_no_contextual_means_undecidable} to new classes of languages. Similarly, it is an interesting prospect to adapt and extend the techniques used in~\Cref{sec_boolean} for proving the existence of contextual quantum polymorphisms for non-Boolean languages, thus ruling out the existence of commutativity gadgets.
While in the current paper we exclusively deal with the exact version of quantum CSPs---i.e., the decision problem asking to decide whether a \textit{perfect} quantum strategy for the corresponding 2-player game exists---techniques based on the weighted-algebra formalism introduced by Mastel--Slofstra~\cite{mastel2024two} appear to be suited for a translation of our results to the setting of imperfect soundness (see also~\cite{Zeman}, where it is shown that Ji's commutativity gadgets are \textit{robust}, in the sense that they can be used in the imperfect-soundness version of the quantum CSP).
Finally, in the current work we only consider quantum strategies for non-local games, described by measurements of a finite-dimensional bipartite quantum state. An interesting direction for follow-up work is to translate our results to the case of quantum-approximable strategies (i.e., limits of quantum strategies) and commuting-operator strategies (i.e., possibly infinite-dimensional quantum strategies where the two players' measurements do not act on separate tensor factors, but are merely required to commute with each other); see for example~\cite{CleveM14,lupini2020perfect}.

\paragraph{\textit{Structure of the paper}}
The remaining part of the paper is organised as follows.
In~\Cref{sec_prelimns}, we formally define classical and quantum CSPs, and we give a list of some basic results about quantum homomorphisms and operations thereof that shall we useful throughout the paper.~\Cref{sec_quantum_algebraic_constructions} presents the algebraic side of our result: we show in~\Cref{subsec_body_algebraic_structure} that quantum polymorphisms are closed under minor operations and, thus, form an abstract minion, and in~\Cref{subsec_body_reductions_via_minion_homo} we prove that minion homomorphism between quantum polymorphism minions determine complexity reductions between the corresponding quantum CSPs. We then turn to the relational side of classical/quantum bridge in~\Cref{sec_quantum_relational_constructions}, by formally describing q-definitions and q-constructions in~\Cref{subsec_body_quantum_reductions}, as well as establishing a Galois connection with quantum polymorphism.
In~\Cref{subsec_body_commutativity_gadgets} and~\Cref{subsec_body_quantum_pol_comm_gadgets}, we show how the contextuality of quantum polymorphisms of a language fully captures commutativity gadgets, which translates into a fully polymorphic, concrete undecidability criterion for quantum CSPs (\Cref{cor_np_complete_plus_no_contextual_means_undecidable}). This provably extends the known undecidability criteria in the literature on non-local games, as discussed in~\Cref{subsec_comparison_TVF}. In~\Cref{sec_cliques}, we describe the quantum polymorphisms of cliques and obtain a simple proof of the undecidability of the corresponding CSPs, first established in~\cite{Ji,culf2024re,Zeman}.
Next, in~\Cref{sec_bifurcations}, we describe non-orthogonality patterns for systems of orthogonal projectors that we call bifurcations, and we show that they occur in any contextual quantum homomorphism (or quantum polymorphism). We then specialise our analysis to particular classes of quantum CSPs. In particular, by using the polymorphic framework, contextuality bifurcations, and ad-hoc arguments based on the specific types of constraints, we establish in~\Cref{sec_siggers} and~\Cref{sec_odd_cycles} the undecidability of the problems encoding quantum satisfiability of Siggers clauses and quantum homomorphisms to odd cycles, respectively. In~\Cref{sec_nonoracular}, we consider the case of non-oracular quantum CSPs, we define the non-oracular version of quantum polymorphisms, and we show that they capture non-oracular commutativity gadgets in a similar way as in the oracular setting. Finally,~\Cref{sec_boolean} address quantum CSPs parameterised by Boolean languages. By describing the quantum polymorphisms of such languages, we obtain a full classification of the existence of commutativity gadgets, as well as a complete $\mathsf{P}$ vs.\ undecidable complexity dichotomy for Boolean quantum CSPs.

\section{Preliminaries}
\label{sec_prelimns}
In this section, we give formal definitions of most of the notions used throughout the paper, and we list a few simple facts on quantum homomorphisms and their operations.

\subsection{Relational structures}
\label{subsec_prelimns_relational_structures}
We let $\N=\{1,2,\dots\}$ denote the set of positive integers.
For $\ell\in\N$, $[\ell]$ is the set $\{1,2,\dots,\ell\}$.
A \emph{signature} $\sigma$ is a finite set of relation symbols $R$, each with its \emph{arity} $\ar(R)\in\N$.
A \emph{relational structure} $\X$ with signature $\sigma$ (in short, a \textit{$\sigma$-structure} or simply a \textit{structure} when $\sigma$ is implicit) consists of a finite set $X$ and a relation $R^\X\subseteq X^{\ar(R)}$ for each symbol $R\in\sigma$.
All relational structures in this work shall be implicitly assumed to be finite.
A \textit{(classical) homomorphism} $f\colon\X\to\Y$ between two relational structures $\X$ and $\Y$ on a common signature $\sigma$ is a map from the domain $X$ of $\X$ to the domain $Y$ of $\Y$ that preserves all relations; i.e., $f(\bx)\in R^\Y$ for each symbol $R\in\sigma$ and each tuple $\bx\in R^\X$ (where $f(\bx)$ is the entrywise application of $f$ to the entries of $\bx$). We denote the existence of a classical homomorphism from $\X$ to $\Y$ by the notation $\X\to\Y$. Moreover, we let $\Hom(\rel X,\rel Y)$ be the set of all homomorphisms $f:\X\to\Y$.
Given a $\sigma$-structure $\X$, a positive integer $r\in\N$, and a set $S\subseteq X^r$, we
let $(X;S)$ be the structure with domain $X$ and a single relation $S$.
We let the \textit{Gaifman graph} of a relational structure $\X$ be  the loopless undirected graph $\Gaif(\X)$ on vertex set $V(\Gaif(\X))=X$ and edge set $E(\Gaif(\X))$ containing an edge $\{x,y\}$ if, and only if, $x$ and $y$ appear together in a tuple in some relation of $\X$. 

Given a $\sigma$-structure $\Y$, the \emph{constraint satisfaction problem} parameterised by $\Y$ (in short, $\CSP(\Y)$) is the following computational problem:\footnote{This is the \emph{decision} version of $\CSP(\Y)$. In the \emph{search} version, the problem is to find an explicit homomorphism $f:\X\to\Y$ for any input $\X$ such that $\X\to\Y$. The two versions of $\CSP$ are equivalent up to polynomial-time reductions~\cite{bulatov2005classifying}.} Given as input a $\sigma$-structure $\X$, output $\YES$ if $\X\to\Y$, and $\NO$ if $\X\not\to\Y$. 
For $\ell\in\N$, we denote by $\Y^\ell$ the $\ell$-th \emph{direct power} of $\Y$; i.e., $\Y^\ell$ is the $\sigma$-structure whose domain is $Y^\ell$ and whose relations are defined as follows: given $R\in\sigma$  and any $\ell\times \ar(R)$ matrix $M$ whose rows are tuples in $R^\Y$, the columns of $M$ form a tuple in $R^{\Y^\ell}$. A (classical) \emph{polymorphism} of $\Y$ (of arity $\ell$) is a homomorphism from $\Y^\ell$ to $\Y$. By $\Pol(\Y)$ (resp.\ $\Pol^{(\ell)}(\Y)$) we denote the set of all polymorphisms (resp.\ all polymorphisms of arity $\ell$) of $\Y$.
The \textit{core} of a $\sigma$-structure $\X$ is a minimal (under inclusion) induced substructure of $\X$ that is homomorphically equivalent to $\X$. It is easy to check that the core of a structure is unique up to isomorphism, and any endomorphism is an automorphism (see e.g.~\cite{HellNesetrilCore}).%

The \emph{primitive positive (pp) fragment} of first-order logic consists of the formulae that are built out of atomic formulae using conjunctions and existential quantifiers.
If $\rel Y$ is a relational structure and $\varphi(x_1,\dots,x_n)$ is a primitive positive formula in the signature of $\rel Y$ with free variables $x_1,\dots,x_n$, then the \emph{relation defined by} $\varphi$ is the set of $n$-tuples that satisfy the formula in $\rel Y$.
If $R\subseteq Y^n$ is a relation, we say that it is \emph{pp-definable} over $\rel Y$ if there exists a pp-formula defining it.
A classical fact about pp-formulae is the following duality with \emph{gadgets}:
every pp-formula can be understood as a relational structure $\rel X_\varphi$ with distinguished vertices $x_1,\dots,x_n$ such that for every structure $\rel Y$, $(y_1,\dots,y_n)$ is in the relation defined by $\varphi$ in $\rel Y$ if, and only if, there exists a homomorphism $\rel X_\varphi\to\rel Y$ that maps $x_i$ to $y_i$ for each $i\in[n]$.
Suppose $\X$ and $\Y$ are relational structures with possible distinct domains and signatures.
We say that $\Y$ has a \textit{pp-construction} in $\X$ (or that $\X$ \textit{pp-constructs} $\Y$) if $\Y$ is homomorphically equivalent to a structure
having domain $X^d$ for some $d\in\N$, and relations
$R_1,\dots,R_k$, where each $R_i\subseteq (X^d)^{r_i}$ is pp-definable over $\X$, when seen as a relation $\tilde{R_i}\subseteq X^{d r_i}$ by ``vectorising'' tuples of tuples.%

\subsection{Quantum CSPs}
\label{subsec_prelimns_quantum_CSPs}
Let $H$ be a non-trivial finite-dimensional Hilbert space. We say that a linear endomorphism $p$ of $H$ is an \emph{orthogonal projector} (or simply a \textit{projector}) if it is self-adjoint and idempotent (i.e., if $p^2=p=p^*$).  We denote by $\id_H$ (or simply by $\id$ when the Hilbert space is clear) the identity projector over $H$. 
A \emph{projective measurement} (PVM) over a Hilbert space $H$ with outcome set $B$ is a family $(Q_{b})_{b\in B}$ of projectors on $H$ with the property that $\sum_{b\in B}Q_b = \id$.
A \emph{quantum function} from $A$ to $B$ is a family of PVMs $(Q_{a,b})_{b\in B}$, one PVM for each $a\in A$. We shall denote this by the notation $Q\colon A\qto B$ or $Q\colon A\qto[H] B$ if we want to emphasise the Hilbert space.
We say that two PVMs $(Q_{b})_{b\in B}$ and $(Q'_{b'})_{b'\in B'}$ over the same Hilbert space \textit{commute} if $[Q_b,Q'_{b'}]=0$ for each $b\in B$ and each $b'\in B'$.  
Given a quantum function $Q\colon A\qto B$ for two sets $A,B$, we say that $Q$ is \textit{non-contextual} if $[Q_{a,b},Q_{a',b'}]=0$ for each $a,a'\in A$ and each $b,b'\in B$. Otherwise, we say that $Q$ is \textit{contextual}. Note that each classical function $f\colon A\to B$ can be seen as a non-contextual quantum function by setting $Q_{a,b}=\id$ if $f(a)=b$ and $Q_{a,b}=0$ otherwise.

Let $\rel A,\rel B$ be relational structures with the same signature $\sigma$.
A \emph{quantum homomorphism} $Q\colon \rel A\qto\rel B$, first defined at this level of generality in~\cite{abramsky2017quantum}, is given by a non-trivial finite-dimensional Hilbert space $H$ and a quantum function $Q\colon A\qto B$ 
such that the following conditions hold:
\begin{itemize}
    \item[($\mathbf{QH_1}$)\labeltext{$\mathbf{QH_1}$}{quantum_homo_1}]$\displaystyle\prod_{i\in [\ar(R)]}Q_{a_i,b_i}=0\quad\forall R\in\sigma,\ba\in R^\A,\bb\in B^{\ar(R)}\setminus R^\B$;
    \item[($\mathbf{QH_2}$)\labeltext{$\mathbf{QH_2}$}{quantum_homo_2}]$\displaystyle[Q_{a,b},Q_{a',b'}]=0\quad\forall (a,a')\in E(\Gaif(\A)),(b,b')\in B^2$.
\end{itemize}

It was proved in~\cite{abramsky2017quantum} that a quantum homomorphism $\A\qto\B$ exists if and only if the \textit{constraint-variable game} associated with $\A$ and $\B$ admits a winning finite-dimensional tensor-product strategy. In the literature on non-local games, this is often expressed in terms of the existence of
a perfect correlation in the set $C_q$ of quantum correlation matrices, see for example~\cite{slofstra2019set,harris2024universality,ji2021mip,levene2026unique}. Equivalently, the notion of a quantum homomorphism can be phrased as a representation of a certain finitely generated $*$-algebra~\cite{CleveM14, paulsen2016estimating,kim2018synchronous,mastel2024two,paddock2025satisfiability}.
Condition~\eqref{quantum_homo_2} is sometimes phrased as $Q$ expressing perfect \textit{oracular} correlations; in~\Cref{sec_nonoracular} we shall extend our results to the non-oracular case.

\begin{definition}[Quantum CSP]
Fix a $\sigma$-structure $\A$.
    $\qCSP(\rel A)$ is the following computational problem: given as input a $\sigma$-structure $\rel X$, determine whether there exists a quantum homomorphism $\rel X\qto\rel A$.
\end{definition}

We will use the next observation on the length of walks in $\Gaif(\A)$ and $\Gaif(\B)$ for two structures $\A$ and $\B$ that are quantum homomorphic.
This result is a slight modification of~\cite[Lemma~4.12]{MancinskaR16} (which concerns the case where $\A$ and $\B$ are undirected graphs), and its proof closely follows that in~\cite{MancinskaR16}. 

\begin{definition}
\label{defn_essentially_binary}
    Let $\A$ be a relational structure with signature $\sigma$. We say that $\A$ is \textit{essentially binary} if, for every $R\in\sigma$ and every $\ba\in R^\A$, $\ba$ has at most two distinct entries.
\end{definition}

\begin{proposition}
\label{prop_basic_walks_essentially_binary}
    Let $Q\colon \rel A\qto\rel B$ be a quantum homomorphism between two structures $\rel A$ and $\rel B$ such that $\A$ is essentially binary. If $a,a'\in A$ and $b,b'\in B$ are such that there exists an $\ell$-walk from $a$ to $a'$ in $\Gaif(\A)$ but there is no $\ell$-walk from $b$ to $b'$ in $\Gaif(\B)$, then $Q_{a,b}Q_{a',b'}=0$.
\end{proposition}
\begin{proof}
We use induction over $\ell$. For $\ell=0$, the result is clear. 
    Suppose $\ell=1$; i.e., $\{a,a'\}\in E(\Gaif(\A))$ but $\{b,b'\}\not\in E(\Gaif(\B))$. Assume, by contradiction, that $Q_{a,b}Q_{a',b'}\neq 0$.  Since $\A$ is essentially binary, there exist some symbol $R\in\sigma$ and some tuple $\bz=(z_1,\dots,z_r)\in R^\A$ such that $\{z_1,\dots,z_r\}=\{a,a'\}$. Consider now the tuple $\bw=(w_1,\dots,w_r)\in B^r$ such that, for each $i\in[r]$, $w_i=b$ if $z_i=a$, and $w_i=b'$ if $z_i=a'$. 
    Using~\eqref{quantum_homo_2} and the idempotency of projectors, we find that
    \begin{align*}
        \prod_{i\in[r]}Q_{z_i,w_i}
        =
        Q_{a,b}Q_{a',b'}\neq 0.
    \end{align*}
    We then deduce from~\eqref{quantum_homo_1} that $\bw\in R^\B$, so $\{b,b'\}\in E(\Gaif(\B))$; a contradiction.
    Take now $\ell\geq 2$, and suppose the result is true for all smaller values of $\ell$. Let   $a_0,a_1,\dots,a_\ell$ be an $\ell$-walk in $\Gaif(\A)$ such that $a_0=a$ and $a_\ell=a'$. We have
    \begin{align*}
        Q_{a,b}Q_{a',b'}
        =
        \sum_{z\in B}Q_{a,b}Q_{a_{\ell-1},z}Q_{a',b'}.
    \end{align*}
    Applying the inductive hypothesis to the $(\ell-1)$-walk $a_0,\dots,a_{\ell-1}$ and the $1$-walk $a_{\ell-1},a_\ell$ in $\Gaif(\A)$, we deduce that, in order for $z$ to yield a nonzero contribution in the summation above, $z$ must be connected to $b$ via an $(\ell-1)$-walk in $\Gaif(\B)$, and to $b'$ via a $1$-walk in $\Gaif(\B)$. Concatenating the two walks yields an $\ell$-walk in $\Gaif(\B)$ connecting $b$ to $b'$. Hence, if no such walk exists, the product $Q_{a,b}Q_{a',b'}$ must be zero.
\end{proof}

\subsection{Operations on quantum homomorphisms}
\label{subsec_body_quantum_operations}
In this subsection, we describe certain constructions taking as input two quantum homomorphisms and returning another quantum homomorphism as the output.

Given a linear operator $Q$ on $H$ and a linear operator $Q'$ on $H'$, we let $Q\oplus Q'$ be the linear operator on $H\oplus H'$ defined by $(Q\oplus Q')(v+v')=Q(v)+Q(v')$.
We extend this definition to quantum functions $Q\colon X\to_H A$, $Q'\colon X\to_{H'} A$
by letting $Q\oplus Q'\colon X\to_{H\oplus H'} A$ be defined by $(Q\oplus Q')_{x,a}:= Q_{x,a}\oplus Q'_{x,a}$.
Note that this is well defined as the set $\{Q_{x,a}\oplus Q'_{x,a}\}_{a\in A}$ is a PVM over the Hilbert space $H\oplus H'$. Indeed, the fact that its members are projectors onto $H\oplus H'$ immediately follows from the definition. Moreover, it is easy to verify that the projector $\sum_{a\in A}(Q_{x,a}\oplus Q'_{x,a})$ acts like the identity on both $H$ and $H'$---and, hence, on the whole $H\oplus H'$ by linearity.
We now show that this operation preserves quantum homomorphisms.

\begin{lemma}\label{sum-homomorphisms}
    Let $Q\colon\rel X\qto[H]\rel A$ and $Q'\colon\rel X\qto[H']\rel A$ be quantum homomorphisms.
    Then $Q\oplus Q'$ is a quantum homomorphism $\rel X\qto[H\oplus H']\rel A$.
\end{lemma}
\begin{proof}
    Let $(x_1,\dots,x_r)\in R^{\rel X}$ and $(a_1,\dots,a_r)\not\in R^{\rel A}$.
    Then $R_{x_1,a_1}\ldots R_{x_r,a_r} = Q_{x_1,a_1}\ldots Q_{x_r,a_r} \oplus Q'_{x_1,a_1}\ldots Q'_{x_r,a_r} = 0$, so~\eqref{quantum_homo_1} holds.

    For~\eqref{quantum_homo_2}, take a symbol $R\in\sigma$ of arity $r$ and
    let $(x_1,\dots,x_r)\in R^{\rel X}$, $i,j\in[r]$, and $a,a'\in A$.
    Then $[R_{x_i,a},R_{x_j,a'}] = [Q_{x_i,a},Q_{x_j,a'}] \oplus [Q'_{x_i,a},Q'_{x_j,a'}] = 0$, so~\eqref{quantum_homo_2} holds.
\end{proof}

Observe that classical homomorphisms $h\colon A\to B$ can be identified with quantum homomorphisms $h\colon A\to_H B$ over a Hilbert space $H$ of dimension $1$. We will implicitly use this identification in the following.
Now, if $Q\colon \rel A\qto\rel B$ is such that $[Q_{a,b},Q_{a',b'}]=0$ for all $a,a'\in A$ and $b,b'\in B$, then $Q = h_1\oplus\dots\oplus h_d$ for some $d\geq 1$ and $h_1,\dots,h_d\colon \rel A\to\rel B$ classical homomorphisms (where $d$ can be taken to be the dimension of $Q$).
This is a simple consequence of the spectral theorem.

\begin{definition}\label{quantum-closure}
    Let $\rel A,\rel B$ be two $\sigma$-structures.
    The \emph{quantum closure} of $\Hom(\rel A,\rel B)$ is the set $\qc{\Hom(\rel A,\rel B)}$ of all quantum homomorphisms of the form $h_1\oplus\dots\oplus h_d$, where $h_1,\dots,h_d\colon\rel A\to\rel B$.
\end{definition}

Explicitly, the quantum closure of $\Hom(\rel A,\rel B)$ consists of the following quantum homomorphisms.
Let $H$ be a (finite-dimensional) Hilbert space, and let $e_1,\dots,e_d$ be an orthonormal basis thereof.
We denote by $\langle\underline{\hspace{0.2cm}},\underline{\hspace{0.2cm}}\rangle$ the inner product of $H$.
Let $h_1,\dots,h_d\colon\rel A\to\rel B$.
Then one can define a quantum homomorphism by
\begin{align}
\label{eqn_1240_0710}
    Q_{a,b}(x) = \sum_{i=1}^d \mathbf 1_{[h_i(a)=b]}\langle x,e_i\rangle e_i
\end{align}
for all $x\in H$, where $\mathbf 1_{[h_i(a)=b]}$ is $1$ if $h_i(a)=b$ and $0$ otherwise.\\

Suppose now that $Q$ is a linear operator on $H$ %
and $Q'$ is a linear operator on $H'$, 
and let $Q\otimes Q'$ be the linear operator on $H\otimes H'$ defined by setting $(Q\otimes Q')(v\otimes v') = Q(v)\otimes Q'(v')$ for all $v\in H$ and $v'\in H'$ and extending by linearity.
As for the direct sum, we extend this definition to quantum functions $Q\colon X\qto[H] A$ and $Q'\colon X\qto[H'] B$ by $(Q\otimes Q')_{x,(a,b)}:=Q_{x,a}\otimes Q'_{x,b}$ for all $a\in A$ and $b\in B$.
This is a quantum function $Q\otimes Q'\colon X\qto[H\otimes H'] A\times B$.
Like for the direct sum, we now show that the tensor product preserves quantum homomorphisms.

\begin{lemma}\label{tensor-homomorphisms}
    Let $Q\colon\rel X\qto[H]\rel A$ and $Q'\colon\rel X\qto[H']\rel B$ be quantum homomorphisms.
    Then $Q\otimes Q'$ is a quantum homomorphism $\rel X\qto[H\otimes H']\rel A\times\rel B$.
\end{lemma}
\begin{proof}
    Let $R=Q\otimes Q'$.
    Let $x\in X, a\in A, b\in B$.
    Then
    \begin{align*}
        R_{x,(a,b)}(R_{x,(a,b)}(v\otimes v')) &= R_{x,(a,b)}(Q_{x,a}(v)\otimes Q'_{x,b}(v'))\\
        &= Q_{x,a}(Q_{x,a}(v))\otimes Q'_{x,b}(Q'_{x,b}(v'))\\
        &= Q_{x,a}(v)\otimes Q'_{x,b}(v')\\
        &= R_{x,(a,b)}(v\otimes v'),
    \end{align*}
    so each $R_{x,(a,b)}$ is a projector onto $H\otimes H'$. 
    Moreover, if $x\in X$, then
    \begin{align*}
        \sum_{a,b} R_{x,(a,b)} &= \sum_{a,b} Q_{x,a}\otimes Q_{x',b}
        = \sum_a Q_{x,a}\otimes\sum_b Q_{x',b}
        = \id_H\otimes\id_{H'}= \id_{H\otimes H'},
    \end{align*}
    so $R$ is a quantum function.
    The verification that~\eqref{quantum_homo_1} holds is similar.

    For~\eqref{quantum_homo_2}, let $x,y$ appear together in a tuple in a relation of $\rel X$, and let $a,a'\in A, b,b'\in B$.
    Then
    \begin{align*}
        R_{x,(a,b)}R_{y,(a',b')} &= (Q_{x,a}\otimes Q'_{x,b})(Q_{y,a'}\otimes Q'_{y,b'})\\
        &= (Q_{x,a}Q_{y,a'}) \otimes(Q'_{x,b}Q'_{y,b'})\\
        &= (Q_{y,a'}Q_{x,a})\otimes (Q'_{y,b'}Q'_{x,b})\\
        &= R_{y,(a',b')}R_{x,(a,b)},
    \end{align*}
    so~\eqref{quantum_homo_2} holds.
\end{proof}

We shall also make use of the following \textit{composition} operation for quantum homomorphisms.
\begin{definition}\label{composition}
Let $Q\colon A\to_H B$ and $R\colon B\to_{H'} C$ be quantum functions.
We define \[R\bullet Q\colon A\to_{H\otimes H'} C\] by $(R\bullet Q)_{a,c}=\sum_{b\in B} R_{b,c}\otimes Q_{a,b}$ for each $a\in A$, $c\in C$.
\end{definition}

The fact that the composition of quantum homomorphisms is a quantum homomorphism was proved in~\cite{abramsky2017quantum}, and it corresponds to the fact that this operation can be expressed as the composition of Kleisli arrows in the Kleisli category for the quantum monad. (See also~\cite{karamlou2025quantum} for a recent extension of the Kleisli framework to different quantum relaxations of classical morphisms.) 
For the convenience of the reader, we give here a self-contained proof following~\cite{abramsky2017quantum}.
\begin{lemma}[\cite{abramsky2017quantum}]\label{composition-homomorphisms}
    The composition of quantum homomorphisms $Q\colon \rel A\to_H\rel B$ and $R\colon\rel B\to_{H'}\rel C$ is a quantum homomorphism $\rel A\to_{H\otimes H'}\rel C$.
\end{lemma}
\begin{proof}
	We first verify that $S=R\bullet Q$ is a quantum function.
	For $a\in A, c\in C$, we get
	\begin{align*}
	S_{a,c}S_{a,c} &= \sum_{b,b'} (R_{b,c}\otimes Q_{a,b})(R_{b',c}\otimes Q_{a,b'})\\
	&= \sum_{b,b'} (R_{b,c}R_{b',c})\otimes (Q_{a,b}Q_{a,b'})\\
	&= \sum_b R_{b,c}\otimes Q_{a,b} &\text{(since $Q_{a,b}Q_{a,b'}=0$ for $b\neq b'$)}\\
	&= S_{a,c}.
	\end{align*}
    Similarly, 
    \begin{align*}
    S_{a,c}^*=(\sum_b R_{b,c}\otimes Q_{a,b})^*=\sum_b R_{b,c}^*\otimes Q_{a,b}^*=\sum_b R_{b,c}\otimes Q_{a,b}=S_{a,c},
    \end{align*}
    so $S_{a,c}$ is a projector.
	We also have
	\[
	\sum_c S_{a,c} = \sum_{b,c} R_{b,c}\otimes Q_{a,b} = \sum_b \id\otimes Q_{a,b} = \id_H\otimes\id_{H'}=\id_{H\otimes H'}.
	\]
    This means that $S$ is indeed a quantum function. Let us prove it is in fact a quantum homomorphism.
	
	If $(a_1,\dots,a_r)\in R^{\rel A}$ and $(c_1,\dots,c_r)\not\in R^{\rel C}$ for some symbol $R\in\sigma$ of arity, say, $r$, then
	\begin{align*}
		S_{a_1,c_1}\dots S_{a_r,c_r} &= \sum_{b_1,\dots,b_r\in B} (R_{b_1,c_1}\otimes Q_{a_1,c_1})\dots (R_{b_r,c_r}\otimes Q_{a_r,b_r})\\
		&= \sum_{b_1,\dots,b_r\in B} (R_{b_1,c_1}\dots R_{b_r,c_r})\otimes (Q_{a_1,b_1}\dots Q_{a_r,b_r}).
	\end{align*}
	Note that for every $(b_1,\dots,b_r)\in B^r$, either $(b_1,\dots,b_r)\in R^{\rel B}$---in which case $R_{b_1,c_1}\dots R_{b_r,c_r}=0$ since $R$ satisfies~\eqref{quantum_homo_1}---or $(b_1,\dots,b_r)\not\in R^{\rel B}$---and then $Q_{a_1,b_1}\dots Q_{a_r,b_r}=0$ since $Q$ satisfies~\eqref{quantum_homo_1}.
	Thus, the sum above evaluates to $0$ and~\eqref{quantum_homo_1} holds.

    If $a,a'$ appear together in a tuple in a relation of $\rel A$ and $c,c'\in C$,
	then
	\begin{align*}
		S_{a,c}S_{a',c'} &= \sum_{b,b'} (R_{b,c}\otimes Q_{a,b})(R_{b',c'}\otimes Q_{a',b'})\\
		&= \sum_{b,b'} (R_{b,c}R_{b',c'})\otimes (Q_{a,b}Q_{a',b'})\\
		&=\sum_{b,b'} (R_{b',c'}R_{b,c})\otimes (Q_{a',b'}Q_{a,b})&\text{by~\eqref{quantum_homo_2}}\\
		&= S_{a',c'}S_{a,c},
	\end{align*}
	so~\eqref{quantum_homo_2} holds.
\end{proof}

We also generalise the composition of quantum functions to functions with arity greater than 1.
\begin{definition}\label{operad-composition}
Fix positive integers $k_1,\dots,k_n$, and let $W_1,\dots, W_n$ be quantum functions, where $W_i\colon A^{k_i}\qto B$ for all $i\in[n]$. Take also a quantum function $Q\colon B^n\qto C$.
    We define
    \[ Q*(W_1,\dots,W_n)\colon A^{\sum k_i}\qto C\]
    by $(Q*(W_1,\dots,W_n))_{(\tuple a^{(1)},\dots,\tuple a^{(n)}),c} = \sum_{\tuple b\in B^n} Q_{\tuple b,c}\otimes\bigotimes_{i=1}^n (W_i)_{\tuple a^{(i)},b_i}$.
\end{definition}

Note that this is a well-defined quantum function, since
\begin{align*}
    \sum_{c\in C}(Q*(W_1,\dots,W_n))_{(\tuple a^{(1)},\dots,\tuple a^{(n)}),c}
    &=\sum_{c\in C}
    \sum_{\tuple b\in B^n} Q_{\tuple b,c}\otimes\bigotimes_{i=1}^n (W_i)_{\tuple a^{(i)},b_i}\\
    &=
    \id\otimes\bigotimes_{i=1}^n\sum_{b\in B}(W_i)_{\tuple a^{(i)},b}=\id.
\end{align*}

\begin{lemma}\label{operad-composition-homomorphism}
If\/ $W_1,\dots,W_n$ are homomorphisms $\A^{k_i}\qto[H_i]\B$ and $Q\colon\B^n\qto[H]\C $, then $Q*(W_1,\dots,W_n)$ is a homomorphism $\A^{\sum k_i}\qto[H\otimes\bigotimes H_i]\C$.
\end{lemma}
\begin{proof}
    Consider the quantum function $W:A^{\sum k_i}\qto B^n$ defined as follows. For $\ba^{(1)}\in A^{k_1},\dots,\ba^{(n)}\in A^{k_n}$ and $\bb\in B^n$, we set
    \begin{align*}
        W_{((\ba^{(1)},\dots,\ba^{(n)}),\bb)}=\bigotimes_{i\in[n]}(W_i)_{\ba^{(i)},b_i}.
    \end{align*}
    Reasoning as in the proof of~\Cref{tensor-homomorphisms}, we readily check that $W$ is a quantum homomorphism $\A^{\sum k_i}\to_{\bigotimes H_i} \B^n$, and that
    \begin{align*}
        Q*(W_1,\dots,W_n)=Q\bullet W.
    \end{align*}
    Then, the result follows from~\Cref{composition-homomorphisms}.
\end{proof}

We will also make use of the following simple result.
\begin{lemma}
\label{lem_ess_classical_endo_automorphism}
    Let $\rel X$ be a core and let $Q\colon\X\qto\X$ be a quantum homomorphism. If $Q$ is non-contextual, then 
    \begin{align*}
        \sum_{x\in X}Q_{x,y}=\id
    \end{align*}
    for each vertex $y$ of $\X$.
\end{lemma}
\begin{proof}
Since $Q$ is non-contextual, we can write it as $Q=\bigoplus_{i\in[d]}h_i$ for some $d\in\N$ and some classical homomorphisms $h_i:\X\to\X$. By the definition of a core, each $h_i$ is an automorphism of $\X$. In particular, for each $i\in[d]$ there exists a unique $x_i$ such that $h_i(x_i)=y$. It follows that
\begin{align*}
    \sum_{x\in X}Q_{x,y}
    =
    \bigoplus_{i\in [d]}\sum_{x\in X}(h_i)_{x,y}
    =\id,
\end{align*}
as required.
\end{proof}

\section{Quantum algebraic reductions}
\label{sec_quantum_algebraic_constructions}

In this section, we start our analysis of reductions between quantum CSPs by lifting~\Cref{thm_classical_minion_homo_reduction} to the quantum setting via the notion of quantum polymorphisms. First of all, in the next subsection, we collect information on the algebraic structure of the quantum polymorphism set of a constraint language.

\subsection{The algebraic structure of quantum polymorphisms}
\label{subsec_body_algebraic_structure}

The set of classical polymorphisms of a relational structure $\A$ can be endowed with the algebraic structure of a \textit{function minion} over $A$.
A function minion over some set $S$ is a subset  $\minion M$ of the set of functions $\{f:S^n\to S:n\in\N\}$ that is closed under identifying variables, permuting variables, and adding dummy variables. This can be expressed compactly as follows: For each $n,m\in\N$, each map $\pi:[n]\to[m]$, and each $f\in \minion M$ of arity $n$ (meaning that $f$ is a function $S^n\to S$), the function $f_{/\pi}:S^m\to S$ defined by
\begin{align*}
    f_{/\pi}(x_1,\dots,x_m)=f(x_{\pi(1)},x_{\pi(2)},\dots,x_{\pi(n)})
\end{align*}
belongs to $\minion M$.
It is clear from the definition of classical polymorphisms that the set $\Pol(\A)$ is invariant under this type of operations (known as \textit{minors}), so $\Pol(\A)$ is indeed a function minion over $A$. In fact, one can endow $\Pol(\A)$ with the richer structure of a (function) clone, due to the fact that polymorphisms are closed under composition (and composition is associative). Nevertheless, it was observed in~\cite{BOP18} that the complexity of CSPs is entirely determined by the \textit{height-one} or \textit{minor identities} holding in $\Pol(\A)$; i.e., identities of the form $f=g_{/\pi}$. In particular, in order to capture reducibility of CSPs via pp-definitions and pp-constructions, it is sufficient to compare their polymorphism clones in terms of minor-preserving maps; i.e., so-called \textit{minion homomorphisms}. The minion framework proved robust enough to capture reducibility and complexity of variants of the CSP; in particular, the \textit{promise CSP}~\cite{AGH17,BG18,BBKO21}, whose polymorphisms do not compose and, thus, do not form clones. In addition, minions were shown to be well-suited to capture the power of \textit{relaxations} of (promise) CSPs via multiple algorithmic models~\cite{BBKO21,cz23sicomp:clap,cz23soda:minions,CiardoZassociation,cz23soda:aip,BrakensiekGS23,ciardo_quantum_minion,beikmohammadi2025discrete,larrauri2025ineffectiveness}; see also~\cite{brakensiek2025richness} for a recent application to CSP sparsification and non-redundancy.
In these works, the notion of function minions was gradually lifted to the categorical concept of \textit{abstract minions}, retaining the algebraic structure of function minions without the restriction of being applied to sets of functions; see, in particular,~\cite{bgwz20,barto2025multisorted}.

\begin{definition}
\label{defn_abstract_minion}
    An \textit{abstract minion} (or simply a \textit{minion}) $\minion M$ is a functor from the  skeleton category of non-empty
finite sets to the category of non-empty sets.\footnote{The fact that abstract minions can be described as functors was first observed by Barto; see also~\cite{hadek2025categorical} for a categorical treatment of the algebraic approach to constraint satisfaction.} In other words, $\minion M$ is the disjoint union of non-empty sets $\minion M^{(n)}$ for $n\in\N$, equipped with operations $(\cdot)_{/\pi}\colon\minion M^{(n)}\to \minion M^{(m)}$ (the \emph{minor operations}) for each $n,m\in\N$ and each map
$\pi\colon[n]\to [m]$, such that
\begin{itemize}
    \item[$(i)$] $M_{/\id}=M$ 
    for each $n\in\N$ and $M\in\minion M^{(n)}$, where $\id$ is the identity map on the set $[n]$;
    \item[$(ii)$] $(M_{/\pi})_{/\tilde\pi}=M_{/\tilde\pi\circ\pi}$
     for each $n,m,p\in\N$, $M\in\minion M^{(n)}$, $\pi\colon[n]\to[m]$,
    and $\tilde\pi\colon[m]\to[p]$.
\end{itemize} 
\end{definition}
Recall now the notion of quantum polymorphisms (\Cref{defn_quantum_polymorphisms}).
Quantum polymorphisms are not functions over a set, but rather \textit{quantum functions}, i.e., systems of projective measurements. Hence, in order to lift the algebraic structure of $\Pol(\A)$ to the set $\qPol(\rel A)$ of quantum polymorphisms of $\rel A$, it is crucial to adopt the formalism of abstract, as opposed to function, minions.
To prove the following result, we define the minor operations in $\qPol(\A)$ in terms of compositions and tensor products of quantum homomorphisms, as described in~\Cref{subsec_body_quantum_operations}.

 \begin{proposition}
 \label{prop_qPol_abstract_minion}
     For each relational structure $\A$, $\qPol(\A)$ is an abstract minion.
 \end{proposition}
 \begin{proof}
     First, observe that $\qPol(\A)$ is the disjoint union of the non-empty sets $\qPol(\A)^{(n)}$ for $n\in\N$, where 
         $\qPol(\A)^{(n)}$ is the set of $n$-ary quantum polymorphisms of $\A$; i.e., quantum homomorphisms $\A^n\qto\A$. 
         For $1\leq i\leq k$, we let $\pi_i^k$ be the $k$-ary projection $\pi_i^k\colon\A^k\to\A$ defined by $\ba\mapsto a_i$ for each $\ba=(a_1,\dots,a_k)\in A^k$. Note that this is a classical polymorphism and, thus, a quantum polymorphism.
         We define minor operations as follows. For $n,m\in\N$, $\tau\colon [n]\to[m]$, and $Q\colon\A^n\qto\A$, we set
         \begin{align}
         \label{eqn_1148_2603}
             Q_{/\tau}=Q\bullet
             \bigotimes_{j\in[n]}
             \pi_{\tau(j)}^m.
         \end{align}

         We need to check that this definition satisfies the two conditions in~\Cref{defn_abstract_minion}. Unfolding the definitions of composition and tensor product of quantum homomorphisms given in~\Cref{subsec_body_quantum_operations}, we see that, for each $\ba\in A^m$ and $b\in A$,
         \begin{align*}
             (Q_{/\tau})_{\ba,b}
             &=
             \left(Q\bullet
             \bigotimes_{j\in[n]}
             \pi_{\tau(j)}^m\right)_{\ba,b}
             =
             \sum_{\bc\in A^n}Q_{\bc,b}\otimes\left(\bigotimes_{j\in[n]} \pi_{\tau(j)}^m\right)_{\ba,\bc}
             =
             \sum_{\bc\in A^n}Q_{\bc,b}\otimes\bigotimes_{j\in[n]} (\pi_{\tau(j)}^m)_{\ba,c_j}\\
             &=Q_{\tau(\ba),b}\otimes \id^{\otimes n},
         \end{align*}
         where $\tau(\ba)=(a_{\tau(1)},\dots,a_{\tau(n)})$. 
         
        Therefore, up to the canonical identification of the Hilbert space $H$ associated with $Q$ with $H\otimes \mathbb K^{\otimes n}$ (where $\mathbb K$ is the base field), we obtain that $(Q_{/\tau})_{\tuple a,b} = Q_{\tau(\ba),b}$.
        From this, it follows immediately that $Q_{/\id}=Q$ and that $(Q_{/\tau})_{/\tau'} = Q_{\tau'\circ\tau}$. Therefore, the minor operations in~\cref{eqn_1148_2603} meet the conditions $(i)$ and $(ii)$ of~\Cref{defn_abstract_minion}, and it follows that $\qPol(\A)$ is a well-defined abstract minion.
 \end{proof}

We now introduce a notion that shall be important for the classification of languages admitting commutativity gadgets (see~\Cref{subsec_body_quantum_pol_comm_gadgets}).

\begin{definition}
Let $\A$ be a relational structure. The \textit{quantum closure} of $\Pol(\rel A)$, denoted by $\qcPol(\rel A)$, is the union of the quantum closures of the set of $n$-ary classical polymorphisms $\Hom(\rel A^n,\rel A)$, for all $n\in\N$.
\end{definition}
We observe the following simple fact. 
\begin{proposition}\label{observation}
A quantum polymorphism of a structure $\A$ belongs to $\qcPol(\rel A)$ if, and only if, it is non-contextual.
\end{proposition}
\begin{proof}
    This is an immediate consequence of the spectral theorem.
\end{proof}

Since every classical polymorphism is non-contextual, and every non-contextual polymorphism is in particular a quantum polymorphism, we have the following inclusions for each structure $\A$: 
\begin{align}
\label{eqn_2116_22_03}
    \Pol(\A)\subseteq\qcPol(\rel A)\subseteq\qPol(\A).
\end{align}

Note also that, since the projection maps $\pi_i^k\colon\A^k\to\A$ used in the proof of~\Cref{prop_qPol_abstract_minion} are classical and, thus, non-contextual, $\qcPol(\A)$ is a minion with the same minor operations as the ones for $\qPol(\A)$; i.e., those defined in~\cref{eqn_1148_2603}. 
We next define minion homomorphisms.

\begin{definition}[\cite{bgwz20,barto2025multisorted}]
    A \emph{minion homomorphism} $\xi:\minion M\to\minion N$ between two minions $\minion M$ and $\minion N$ is a natural transformation from $\minion M$ to $\minion N$. I.e., $\xi$ is a map from the underlying set of $\minion M$ to the underlying set of $\minion N$ that
\begin{itemize}
    \item preserves the arity---i.e., $\xi(M)\in\minion N^{(n)}$ if $M\in\minion M^{(n)}$;
    \item preserves the minors---i.e., $\xi(M_{/\pi})=\xi(M)_{/\pi}$ if $M\in\minion M^{(n)}$ and  $\pi:[n]\to[m]$.
\end{itemize}
\end{definition}
\noindent We denote the existence of a minion homomorphism from $\minion M$ to $\minion N$ by the notation $\minion M\to\minion N$.
The following result shows that minion homomorphisms between polymorphism minions of CSPs exactly capture reducibility via pp-constructions.
\begin{theorem}[\cite{BOP18}]
\label{minion_homo_q_construction}
    Let $\A$ and $\B$ be relational structures. Then the following are equivalent:
    \begin{itemize}
        \item $\Pol(\A)\to\Pol(\B)$;
        \item $\A$ pp-constructs $\B$.
    \end{itemize}
\end{theorem}

In particular, a minion homomorphism $\Pol(\A)\to\Pol(\B)$ gives rise to a complexity reduction from $\CSP(\B)$ to $\CSP(\A)$.

Since inclusions between minions are in particular homomorphisms, we immediately obtain from~\eqref{eqn_2116_22_03} that $\Pol(\A)\to\qcPol(\rel A)\to\qPol(\A)$ for each $\A$. In addition, we observe the next fact.  

\begin{proposition}
\label{prop_pol_hom_means_nc_hom}
    Let $\A, \B$ be relational structures such that $\Pol(\A)\to\Pol(\B)$.
    Then,\\ $\qcPol(\A)\to\qcPol(\B)$.
\end{proposition}
\begin{proof}
    Let $\xi\colon\Pol(\A)\to\Pol(\B)$.
    Define $\zeta\colon\qcPol(\A)\to\qcPol(\B)$ as follows.
    For any $Q\in\qcPol(\A)$, there are by definition 1-dimensional Hilbert spaces $H_1,\dots,H_d$ and classical polymorphisms $f_1,\dots,f_d\in\Pol(\A)$ such that $Q=\oplus_{i=1}^d f_i$. 
    Then define $\zeta(Q)$ as $\oplus_{i=1}^d\xi(f_i)$, in the same basis.
    The verification that $\zeta$ is a minion homomorphism is immediate since $Q_{/\pi}=\oplus_{i=1}^d (f_i)_{/\pi}$.
\end{proof}

Although this bears no consequences for the present work, we remark here that $\qPol(\A)$ can also be endowed with the algebraic structure of an \emph{operad}.
The $n$-ary elements of the operad are the $n$-ary quantum polymorphisms of $\A$, the identity is the identity map $\id\colon\A\to\A$, and the composition is given by the $*$ operation of~\Cref{operad-composition}.
\begin{proposition}\label{qPol-operad}
    For each relational structure $\A$, $\qPol(\A)$ is an operad.
\end{proposition}
\begin{proof}
    We first need to check that $Q*(\id,\dots,\id)=Q=\id{}* Q$.
    For the first equality, we get $(Q*(\id,\dots,\id))_{\tuple a,c} = \sum_{\tuple b\in B^n}Q_{\tuple b,c}\otimes\id_{a_1,b_1}\otimes\dots\otimes\id_{a_n,b_n}=Q_{\tuple a,c}\otimes\id^{\otimes n} = Q_{\tuple a,c}$.
    (Note that, as in the proof of~\Cref{prop_qPol_abstract_minion}, we work here with the identification of $H$ and $H\otimes\mathbb K^{\otimes n}$.)
The verification for the right-hand side is similar.

    We now prove operad-associativity; i.e., we prove that if $Q$ is $n$-ary, $W_i$ is $k_i$-ary for $i\in[n]$, and $V_{i,j}$ is $\ell_{i,j}$-ary for $i\in[n]$ and $j\in[k_i]$, then 
    \begin{align*}
        &(Q*(W_1,\dots,W_n))*(V_{1,1},\dots,V_{1,k_1},\dots,V_{n,1},\dots,V_{n,k_n})
        \intertext{and}
        &Q*(W_1*(V_{1,1},\dots,V_{1,k_1}),\dots,W_n*(V_{n,1},\dots,V_{n,k_n}))
    \end{align*}
    coincide.
    The verification is slightly tedious and we omit it to the benefit of the reader; it suffices to remark that, for both sides, the projectors corresponding to the tuples $\tuple a^{i,j}\in A^{\ell_{i,j}}$ (for $i\in[n]$ and $j\in[k_i]$) and $d\in A$, are equal to
    \[ \sum_{\tuple c\in A^n,\tuple b^1\in A^{k_1},\dots,\tuple b^n\in A^{k_n}}Q_{\tuple c,d}\otimes\bigotimes_{i\in[n]} (W_i)_{\tuple b^i,c_i}\otimes\bigotimes_{i\in[n], j\in [k_i]} (V_{i,j})_{\tuple a^{i,j},b^i_j}.\qedhere \]
\end{proof}

It would seem tempting to believe that $\qPol(\A)$, being simultaneously an operad and a minion, is a clone. %
The natural definition of an (abstract) clone from an operad-minion is the following.
Let  $Q$ be $n$-ary, and let $W_1,\dots,W_n$ be $k$-ary.
Define $\sigma\colon[nk]\to[k]$ by $\sigma(i)=i\bmod k$, where we take here the convention that the $\bmod$ operation takes values in $\{1,\dots,k\}$.
Then define $Q\circ(W_1,\dots,W_n)$ to be the $k$-ary operation $(Q*(W_1,\dots,W_n))_{/\sigma}$.
While this definition is valid, and gives rise to the quantum function with projectors $\sum_{\tuple b}Q_{\tuple b,c}\otimes\bigotimes_{i=1}^n (W_i)_{\tuple a,b_i}$,
it fails to be associative.%

\subsection{Complexity reductions via minion homomorphisms}
\label{subsec_body_reductions_via_minion_homo}

The goal of this subsection is to prove that the existence of a minion homomorphism between the quantum polymorphism minions of two structures gives rise to a logspace reduction between the corresponding quantum CSPs. 

\thmMainqPolHomoReductions*
To prove the above result, we quantise the completeness and soundness analysis for the reduction coming from the Long-Code testing of classical constraint systems~\cite{bellare1998free}, see for example~\cite[\S3.3]{BBKO21} and~\cite{chen2017asking}. We observe that,  while in the classical case the completeness of the reduction is readily proved, in the quantum setting both the completeness and the soundness cases are quite technical.
\begin{proof}
    Let $\sigma$ and $\rho$ be the signatures of $\A$ and $\B$, respectively. 
    Without loss of generality, we take the domain of $B$ to be $[n]$.
    Take an instance $\X$ of $\qCSP(\B)$ (i.e., $\X$ is an arbitrary $\rho$-structure). Consider the $\sigma$-structure $\Y$ defined via the following process. 
    \begin{itemize}
        \item[$(i)$] For each $x\in X$, add in $\Y$ a copy of $\A^{n}$. Given a tuple $\ba\in A^{n}$, we denote the corresponding variable of $\Y$ by ${(x,\ba)}$.%
        \item[$(ii)$] For each symbol $R\in\rho$ and each tuple $\bx\in R^\X$, add in $\Y$ a copy of $\A^{|R^\B|}$. Given a tuple $\bc\in A^{|R^\B|}$, we denote the corresponding variable of $\Y$ by ${(R,\bx,\bc)}$. 
        \item[$(iii)$] 
        We identify variables in $\Y$ as follows. For each symbol $R\in\rho$ (of arity, say, $r$), each tuple $\bx\in R^\X$, each index $i\in[r]$, and each tuple $\ba\in A^{n}$ we identify the variable ${(x_i,\ba)}$ with the variable ${(R,\bx,\bc)}$, where $\bc\in A^{|R^\B|}$ is the tuple obtaining by composing $\ba$ with the $i$-th projection $\pi_i\colon R^\B\to B$ mapping each $(b_1,\dots,b_r)\in R^\B$ to $b_i$. 
        In more details, $\bc$ is the tuple $(a_{\tuple r_1(i)},\dots,a_{\tuple r_N(i)})$, where $R^{\B}=\{\tuple r_1,\dots,\tuple r_N\}$ is a fixed enumeration of $R^{\B}$.
    \end{itemize}
    Note that the structure $\Y$ resulting from this process is a $\sigma$-structure and, thus, a well-defined instance of $\qCSP(\A)$.
    Given some $y\in Y$, we shall say that $y$ is a variable of type $(i)$ if it comes from step $(i)$, and we shall say that $y$ is a variable of type $(ii)$ if it comes from step $(ii)$; observe that, because of step $(iii)$, $y$ may be of type $(i)$ and type $(ii)$, simultaneously.
    We claim that the reduction $\X\mapsto\Y$ is complete and sound.
    Henceforth in this proof, with a slight abuse of notation, we identify tuples in $S^{|T|}$ with maps $T\to S$.\\

(\underline{\textit{Completeness.}}) Take a quantum homomorphism $Q\colon\X\qto\B$.
We define a quantum function $W\colon Y\qto A$ as follows.
For each $x\in X$, each $\ba\in A^{n}$, and each $p\in A$, we set
\begin{align}
\label{eqn_1133_2703_1}
    W_{{(x,\ba)},p}= \sum_{\substack{b\in B\\\ba(b)=p}}Q_{x,b}.
\end{align}
Moreover, for each $R\in\rho$, each $\bx\in R^\X$, each $\bc\in A^{|R^\B|}$, and each $p\in A$, we set
\begin{align}
\label{eqn_1133_2703_2}
    W_{{(R,\bx,\bc)},p}= \sum_{\substack{\bb\in R^\B\\\bc(\bb)=p}}\prod_{j\in[\ar(R)]}Q_{x_j,b_j}.
\end{align}

\begin{claim}
    $W$ preserves the identification step $(iii)$.
\end{claim}
\begin{proof}
    Suppose that the variable ${(R,\bx,\bc)}$ is identified with the variable ${(x_i,\ba)}$ in $\Y$; i.e., $\bc=\ba\circ\pi_i$. We find
    \begin{align*}
        W_{{(R,\bx,\bc)},p}
        &=
        \sum_{\substack{\bb\in R^\B\\\ba(b_i)=p}}\prod_{j\in[\ar(R)]}Q_{x_j,b_j}
        =
        \sum_{\substack{\bb\in B^{\ar(R)}\\\ba(b_i)=p}}\prod_{j\in[\ar(R)]}Q_{x_j,b_j}\\
        &=
        \left(\sum_{b_1\in B}Q_{x_1,b_1}\right)\cdot
        \left(\sum_{b_2\in B}Q_{x_2,b_2}\right)\dots
    \left(\sum_{b_{i-1}\in B}Q_{x_{i-1},b_{i-1}}\right)\cdot
    \left(\sum_{\substack{b_{i}\in B\\\ba(b_i)=p}}Q_{x_{i},b_{i}}\right)\\
    &\cdot
    \left(\sum_{b_{i+1}\in B}Q_{x_{i+1},b_{i+1}}\right)\dots
    \left(\sum_{b_{\ar(R)}\in B}Q_{x_{\ar(R)},b_{\ar(R)}}\right)\\
    &=
    \sum_{\substack{b_{i}\in B\\\ba(b_i)=p}}Q_{x_{i},b_{i}}
    =
    W_{{(x_i,\ba)},p}
    \end{align*}
    as required. Note that the second equality crucially uses that $Q$ satisfies the constraint $\bx$, in the sense of~\eqref{quantum_homo_1}.
\end{proof}
\begin{claim}
    $W$ is a quantum map from $Y$ to $A$.
\end{claim}
\begin{proof}
First note that the summands in~\cref{eqn_1133_2703_1} are mutually orthogonal and, thus, their sum is a projector. Similarly, each product $\prod_{j\in[\ar(R)]}Q_{x_j,b_j}$ in~\cref{eqn_1133_2703_1} consists of mutually commuting projectors (by~\eqref{quantum_homo_2}) and, thus, the product is a projector. Moreover, products associated with distinct $\bb\in R^\B$ are orthogonal by~\eqref{quantum_homo_1} and~\eqref{quantum_homo_2}, so the sum is a well-defined projector. 

For variables of type $(i)$, we have 
    \begin{align*}
        \sum_{p\in A}W_{{(x,\ba)},p}= \sum_{p\in A} \sum_{\substack{b\in B\\\ba(b)=p}}Q_{x,b}
        =
        \sum_{b\in B}Q_{x,b}
        =
        \id.
    \end{align*}
    For variables of type $(ii)$, we have 
    \begin{align*}
    \sum_{p\in A}W_{{(\bx,\bc)},p}= \sum_{p\in A}\sum_{\substack{\bb\in R^\B\\\bc(\bb)=p}}\prod_{j\in[\ar(R)]}Q_{x_j,b_j}
    =
    \sum_{\bb\in R^\B}\prod_{j\in[\ar(R)]}Q_{x_j,b_j}
    =
    \sum_{\bb\in B^{\ar(R)}}\prod_{j\in[\ar(R)]}Q_{x_j,b_j}=\id,
\end{align*}
where we are again using~\eqref{quantum_homo_1}. 
\end{proof}

We now proceed to show that the quantum map $Q\colon Y\qto A$ is in fact a quantum homomorphism $\Y\qto\A$. To that end, we need to check that $W$ satisfies the conditions~\eqref{quantum_homo_1} and~\eqref{quantum_homo_2}.

\begin{claim}
    $W$ satisfies~\eqref{quantum_homo_1}.
\end{claim}
\begin{proof}
    Let $S$ be a symbol in $\sigma$ of some arity $s$. Take a tuple $\by\in S^\Y$ and a tuple $\ba\in A^s\setminus S^\A$.
    We need to prove that
    \begin{align}
    \label{eqn_1220_2703}
        \prod_{k\in[s]}W_{y_k,a_k}=0.
    \end{align}
    We now distinguish between two cases.
    \begin{itemize}
        \item Case 1: All variables appearing in $\by$ are of type $(i)$. In this case, there exists some $x\in X$ and some tuples $\bz^{(1)},\dots,\bz^{(s)}\in A^{n}$ such that
        $y_k={(x,\bz^{(k)})}$ for each $k\in[s]$, and        
        the tuple $(\ba^{(1)},\dots,\ba^{(s)})$ belongs to $S^{\A^{n}}$ (where, recall, $S^{\A^{n}}$ is the interpretation of the symbol $S\in\sigma$ in the $n$-wise categorical power $\A^{n}$ of $\A$). In other words, for each $\ell\in B$ it holds that $(z^{(1)}_\ell,\dots,z^{(s)}_\ell)\in S^\A$. In this case, we find
        \begin{align*}
            \prod_{k\in[s]}W_{y_k,a_k}
            =
            \prod_{k\in[s]}W_{{(x,\bz^{(k)})},a_k}
            =
            \prod_{k\in[s]}\sum_{\substack{b_k\in B\\\bz^{(k)}(b_k)=a_k}}Q_{x,b_k}
            =
            \sum_{\substack{\bb\in B^s\\ \bz^{(k)}(b_k)=a_k\;\forall k\in[s]}}\prod_{k\in[s]}Q_{x,b_k}.
        \end{align*}
        Now, because of the idempotency of projectors and the fact that projectors in a PVM having different outcomes are mutually orthogonal, the above expression trivialises unless all $b_k$'s are equal. Hence, we find
        \begin{align*}
            \prod_{k\in[s]}W_{y_k,a_k}
            =
            \sum_{\substack{b\in B\\ \bz^{(k)}(b)=a_k\;\forall k\in[s]}}Q_{x,b}.
        \end{align*}
        We now argue that the above sum is zero. Indeed, for any $b\in B$, the requirement that $\bz^{(k)}(b)=a_k$  $\forall k\in[s]$ implies that the tuple $(z^{(1)}_b,\dots,z^{(s)}_b)$ does not belong to $S^\A$, since we are assuming $\ba\not\in S^\A$. This, however, is a contradiction. It follows that~\cref{eqn_1220_2703} holds in this case.
        \item Case 2: All variables appearing in $\by$ are of type $(ii)$. The argument is similar to the previous case. There exists some symbol $R\in\rho$, some tuple $\bx\in R^\X$, and some tuples $\bz^{(1)},\dots,\bz^{(s)}\in A^{|R^\B|}$ such that
        $y_k={(R,\bx,\bz^{(k)})}$ for each $k\in[s]$, and        
        the tuple $(\ba^{(1)},\dots,\ba^{(s)})$ belongs to $S^{\A^{|R^\B|}}$. Unfolding the definition of relations in categorical powers of structures, this means that, for each $\bb\in R^\B$, it holds that $(z^{(1)}_\bb,\dots,z^{(s)}_\bb)\in S^\A$. Letting $r$ be the arity of $R$, we find
        \begin{align*}
            \prod_{k\in[s]}W_{y_k,a_k}
            &=
            \prod_{k\in[s]}W_{{(R,\bx,\bz^{(k)})},a_k}
            =
            \prod_{k\in[s]}\sum_{\substack{\bb^{(k)}\in R^\B\\\bz^{(k)}(\bb^{(k)})=a_k}}\prod_{j\in[r]}Q_{x_j,b^{(k)}_j}\\
            &=
            \sum_{\substack{(\bb^{(1)},\dots,\bb^{(k)})\in(R^\B)^k\\\bz^{(k)}(\bb^{(k)})=a_k\;\forall k\in[s]}}\prod_{k\in[s]}\prod_{j\in[r]}Q_{x_j,b^{(k)}_j}.
        \end{align*}
        We now observe that the order in which the products are performed in the expression above is irrelevant, due to the fact that, by assumption, $\bx$ belongs to $R^\X$ and $Q$ satisfies~\eqref{quantum_homo_2}.
        In particular, we have
        \begin{align*}
            \prod_{k\in[s]}W_{y_k,a_k}
            &=
            \sum_{\substack{(\bb^{(1)},\dots,\bb^{(k)})\in(R^\B)^k\\\bz^{(k)}(\bb^{(k)})=a_k\;\forall k\in[s]}}\prod_{j\in[r]}\prod_{k\in[s]}Q_{x_j,b^{(k)}_j}.
        \end{align*}
        We now proceed similarly to case 1, by observing that, since each $Q_{x_j}$ is a PVM, the product $\prod_{k\in[s]}Q_{x_j,b^{(k)}_j}$ trivialises unless all $b_j^{(k)}$'s are the same (for any fixed $j\in[r]$). This means that
        \begin{align*}
            \prod_{k\in[s]}W_{y_k,a_k}
            &=
            \sum_{\substack{\bb\in R^\B\\\bz^{(k)}(\bb)=a_k\;\forall k\in[s]}}\prod_{j\in[r]}Q_{x_j,b_j}.
        \end{align*}
        For each $\bb\in R^\B$, we know that the tuple $(z^{(1)}_\bb,\dots,z^{(s)}_\bb)$ must belong to $S^\A$. Hence, such tuple cannot equal $\ba$, which is assumed to be in $A^s\setminus S^\A$. It follows that the expression above is zero, thus proving~\cref{eqn_1220_2703} in this case, too. This concludes the proof of the claim.\qedhere
    \end{itemize}
    \end{proof}
    \begin{claim}
    $W$ satisfies~\eqref{quantum_homo_2}.
\end{claim}
\begin{proof}
    Let $y,y'\in Y$ be such that $\{y,y'\}\in E(\Gaif(\Y))$. This means that there exists some symbol $S\in\sigma$, some tuple $\by\in S^\Y$, and two distinct indices $i,j\in[\ar(S)]$ such that $y=y_i$ and $y'=y_j$. If all entries of $\by$ are variables of type $(i)$, we see from~\cref{eqn_1133_2703_1} that the PVMs $W_{y}$ and $W_{y'}$ are both linear combinations of projectors in the PVM $Q_x$ for some common $x\in X$, and thus they commute.  If all entries of $\by$ are variables of type $(ii)$, we see from~\cref{eqn_1133_2703_2} that both $W_y$ and $W_{y'}$ are linear combinations of projectors in the PVM $Q_\bx$ with outcome set $B^{\ar(R)}$ for some $R\in\rho$ and some $\bx\in R^\X$, where $Q_\bx$ is defined by
    \begin{align*}
        Q_{\bx,\bb}=\prod_{j\in[\ar(R)]}Q_{x_j,b_j}
    \end{align*}
    for each $\bb\in B^{\ar(R)}$. Hence, $W_y$ and $W_{y'}$ commute in this case, too.
\end{proof}
In summary, we have shown that, whenever $\X\qto\B$, it holds that $\Y\qto\A$, which proves completeness of the reduction.\\

(\underline{\textit{Soundness.}}) Let $W\colon\Y\qto\A$ be a quantum homomorphism.
For $x\in X$, let $W^{(x)}$ be the quantum polymorphism of $\A$ defined by $W^{(x)}_{\tuple a,b} := W_{{(x,\tuple a)},b}$.
Similarly, for $R\in\rho$ and $\tuple x\in R^{\X}$, let $W^{(R,\tuple x)}_{\tuple c,b} := W_{{(R,\tuple x,\tuple c)},b}$.

Recall that for $R\in\rho$, we have fixed an enumeration of $R^{\B}$ as $\{\tuple r_1,\dots,\tuple r_N\}$.
Note that if $R$ is $r$-ary and $\pi_i\colon [|R^{\B}|]\to[n]$ is the map $j\mapsto r_j(i)$ for  $i\in[r]$, then $W^{(x_i)} = W^{(R,\tuple x)}_{/\pi_i}$.
Indeed, this follows from the equalities
\begin{align*}
    (W^{(R,\tuple x)}_{/\pi_i})_{\tuple a,b} = W_{{(R,\tuple x,(a_{\tuple r_1(i)},\dots,a_{\tuple r_N(i)}))},b} = W_{(\tuple x,\tuple a),b} = W^{(x)}_{\tuple a,b},
\end{align*}
where the second equality follows from the fact that the variables $(R,\tuple x,a_{\tuple r_1(i),\dots,\tuple r_N(i)})$ and $(x_i,\tuple a)$ of $\Y$ have been identified.

Define $Q\colon X\qto B$ to be the quantum function $Q_{x,b}:=\xi(W^{(x)})_{(1,\dots,n),b}$, where $\xi$ is a minion homomorphism $\qPol(\A)\to\qPol(\B)$ whose existence we are assuming.
We now check that $Q$ is a quantum homomorphism $\C\qto\B$.

First, let $(x_1,\dots,x_r)\in R^{\X}$ and $(b_1,\dots,b_r)\not\in R^{\B}$.
We compute
\begin{align*}
    Q_{x_1,b_1}\cdots Q_{x_r,b_r} &= \prod_{i=1}^r \xi(W^{(x_i)})_{(1,\dots,n),b_i}= \prod_{i=1}^r \left(\xi(W^{(R,\tuple x)})_{/\pi_i}\right)_{(1,\dots,n),b_i}\\
    &= \prod_{i=1}^r \xi(W^{(R,\tuple x)})_{(\tuple r_1(i),\dots,\tuple r_N(i)),b_i} =0,
\end{align*}
where the last equality follows from the fact that $\xi(W^{(R,\tuple x)})$ is a quantum polymorphism of $\B$, and that the $r$-tuple whose entries are $(\tuple r_1(i),\dots,\tuple r_N(i))$ for $i\in[r]$ is in $R^{\B^N}$, while $\tuple b\not\in R^{\B}$ by assumption.
Thus, $Q$ satisfies~\eqref{quantum_homo_1}.

To check~\eqref{quantum_homo_2}, let $(x_1,\dots,x_r)\in R^{\B}$ for some $R\in\rho$ and let $i,j\in[r]$.
Then for all $b,b'\in B$, we have that 
\begin{align*}[Q_{x_i,b},Q_{x_j,b}] &= \left[\left(\xi(W^{(R,\tuple x)})_{/\pi_i}\right)_{(1,\dots,n),b}, \left(\xi(W^{(R,\tuple x)})_{/\pi_j}\right)_{(1,\dots,n),b'}\right]\\
&=\left[\xi(W^{(R,\tuple x)})_{(\tuple r_1(i),\dots,\tuple r_N(i)),b},\xi(W^{(R,\tuple x)})_{(r_1(j),\dots,r_N(j)),b'}\right]=0,
\end{align*}
where the last equality is because $\xi(W^{(R,\tuple x)})$ satisfies~\eqref{quantum_homo_2}.
\end{proof}

\begin{remark}
For the readers familiar with the complexity theory of promise CSPs (in particular, with the reduction in~\cite{BBKO21}), we remark that 
the reduction in~\Cref{thm_main_qPol_homo_reductions} can be cast as a reduction from $\qCSP(\B)$ to the ``quantum minor-condition problem'' $\textsc{qMC}(\qPol(\B), N)$, where $N$ is a constant that depends on $\A$. If there is a minion homomorphism $\qPol(\A)\to\qPol(\B)$, $\textsc{qMC}(\qPol(\B),N)$ trivially reduces to $\textsc{qMC}(\qPol(\A), N)$(where the reduction does nothing).
Finally, $\textsc{qMC}(\qPol(\A), N)$ reduces to $\qCSP(\A)$ by replacing every symbol $f$ of arity $k\leq N$ by a copy of $\A^k$ (steps (i) and (ii) of the reduction above), and identifying vertices according to the minor identities given in the instance (step (iii) of the reduction above).%
\end{remark}

\begin{remark}
We point out that for the reduction in~\Cref{thm_main_qPol_homo_reductions} to be implementable in logspace (and, similarly, for all other logspace reductions we give in this work), we implicitly assume that the quantum CSP instances are given via a \textit{non-succinct} presentation (following the terminology, for example, of~\cite{fu2025succinct}, see also~\cite[\S1.1]{culf2025quantum}).
The fact that the identification of vertices in step (iii) can be achieved in logarithmic space is due to~\cite{Reingold}. 
We remark that in our applications, these reductions are used solely to transfer undecidability, for which it suffices that they be computable. Consequently, the distinction between succinct and non-succinct presentations is immaterial for our purposes. 
\end{remark}

\section{Quantum relational constructions}
\label{sec_quantum_relational_constructions}

In this section, we introduce quantised notions of pp-definitions and pp-constructions that we call \textit{q-definitions} and \textit{q-constructions}, respectively, and we link them to the quantum polymorphisms. Moreover, we show how these quantum constructions can be systematically built by enriching their classical counterparts with the \textit{commutativity gadgets} introduced in~\cite{Ji}. Then, we prove that the existence of commutativity gadgets is entirely determined by a closure property of the quantum polymorphisms of the given language.

\subsection{q-definitions and q-constructions}
\label{subsec_body_quantum_reductions}
Recall the notion of pp-definitions given in~\Cref{subsec_prelimns_relational_structures}. We now lift it to a quantum version. Recall from~\Cref{subsec_prelimns_relational_structures} that, for a finite set $A$ and a subset $S$ of $A^r$, by the notation $(A;S)$ we mean the relational structure having domain $A$ and a single, $r$-ary relation $S$.

\begin{definition}\label{qgadget}
    Let $\rel A$ be a $\sigma$-structure    and let $S\subseteq A^r$ be a set of $r$-tuples of elements of $A$ for some $r\in\N$. Let $R$ be the $r$-ary symbol of the signature of $(A;S)$, and 
    consider the $\{R\}$-structure
    $\rel R$ having domain $[r]$ and relation $R^{\rel R}=\{(1,\dots,r)\}$.
    A \emph{q-definition} for $\rel A$ and $S$ is a $\sigma$-structure $\GG$ with distinguished vertices $g_1,\dots,g_r$ satisfying the following properties:
    \begin{itemize}
        \item[$(\mathbf{q_1})$\labeltext{$\mathbf{q_1}$}{itm:qgadget-extension}] For every quantum homomorphism $Q\colon\rel R\qto (A;S)$, 
        there exists a quantum homomorphism $Q'\colon\GG\qto\rel A$ such that $Q'_{g_i,a} = Q_{i,a}$ for all $i\in[r]$ and all $a\in A$.
        \item[$(\mathbf{q_2})$\labeltext{$\mathbf{q_2}$}{itm:qgadget-restriction}] For every quantum homomorphism $Q\colon\GG\qto\rel A$, 
        the quantum function $Q'$ defined by $Q'_{i,a}:=Q_{g_i,a}$ for all $i\in[r]$ is a quantum homomorphism $\rel R\qto (A;S)$.
    \end{itemize}
\end{definition}
Thus, a q-definition is similar to a classical gadget where we impose (in~\eqref{itm:qgadget-restriction}) the additional property that quantum homomorphisms $Q\colon\GG\qto\A$ must be non-contextual over the special vertices $\{g_1,\dots,g_r\}$ of the gadget.
In particular, in the case that $r=1$, then q-definitions and pp-definitions coincide.

Just as in the classical setting, we have a correspondence between the set of relations with a q-definition over a structure $\A$ and the set of relations $R$ such that $\qPol(\A)\subseteq\qPol(A;R)$.
\begin{theorem}\label{galois-connection-oracular}
    Let $\A$ be a relational structure and let $R\subseteq A^r$.
    Then $R$ has a q-definition over $\A$ if, and only if, $\qPol(\A)\subseteq\qPol(A;R)$.
\end{theorem}
\begin{proof}
    Suppose first that $R$ has a q-definition over $\A$.
    Let $(\GG, g_1,\dots,g_r)$ be a q-definition, and pick a quantum polymorphism $Q\colon\A^n\qto\A$ of some arity $n$. We need to show that $Q$ is an $n$-ary quantum polymorphism of $(A;R)$.

    We first show~\eqref{quantum_homo_1}.
    Let $\tuple a^1,\dots,\tuple a^n\in R$ and $\tuple b\not\in R$.
    For every $i\in[n]$, the classical map $h_i\colon [r]\to A$ defined by $j\mapsto a^i_j$ is a homomorphism $\R\to (A;R)$.
    By~\eqref{itm:qgadget-extension}, there exists for all $i\in[n]$ a homomorphism $H_i\colon\GG\to\A$ such that $H_i(g_j)=h_i(j)$ for all $j\in[r]$.
    Define a quantum function $U\colon [r]\qto A$ by $U_{j,b}:= (Q\bullet(H_1,\dots,H_n))_{g_j,b}$.
    By~\eqref{itm:qgadget-restriction}, this is a quantum homomorphism $\R\qto(A;R)$
    and therefore $\prod_{j=1}^r U_{j,b_j} = 0$.
    But since $U_{j,b_j} = Q_{(a^1_j,\dots,a^n_j),b_j}\otimes\id^{\otimes n}$, we obtain that
    $\prod_{j=1}^r Q_{(a^1_j,\dots,a^n_j),b_j}=0$ and therefore $Q$ satisfies~\eqref{quantum_homo_1}.

    We now check~\eqref{quantum_homo_2}.
    Let $\tuple a^1,\dots,\tuple a^n\in R$, and let $i,j\in[r]$.
    Let $a,b\in A$.
    Using the same notation as above, we have $[Q_{(a^1_i,\dots,a^n_i),a},Q_{(a^1_j,\dots,a^n_j),b}]\otimes\id^{\otimes n}=[U_{i,a},U_{j,b}]=0$.

    We now prove the converse implication. Suppose that $\qPol(\A)\subseteq\qPol(A;R)$.
    Let $\tuple a^1,\dots,\tuple a^n$ be an enumeration of all the tuples in $R$.
    We define the q-definition $\GG$ to be the structure $\A^n$, and its $r$ distinguished vertices to be $(a^1_i,\dots,a^n_i)$ for $i\in[r]$.

    We prove that~\eqref{itm:qgadget-restriction} holds.
    Let $Q\colon\A^n\qto \A$ and define $Q'_{i,b}=Q_{(r^1_i,\dots,r^n_i),b}$, which is a quantum function $[r]\to A$.
    Let $(b_1,\dots,b_r)\not\in R$.
    Then $\prod_{i=1}^r Q'_{i,b_i} = \prod_{i=1}^r Q_{(r^1_i,\dots,r^n_i),b_i}=0$ since $Q\in\qPol(A;R)$ by assumption and by~\eqref{quantum_homo_1}.
    Thus,~\eqref{quantum_homo_1} holds for $Q'$.

    Similarly, if $i,j\in[r]$, $a,b\in A$, then $[Q'_{i,a},Q'_{j,b}]=[Q_{(a^1_i,\dots,a^n_i),a},Q_{(a^1_j,\dots,a^n_j),b}]=0$, since~\eqref{quantum_homo_2} holds for $Q$ as a quantum polymorphism of $(A;R)$.

    Finally, we check that~\eqref{itm:qgadget-extension} holds.
    Let $Q\colon\R\qto (A;R)$ be a quantum homomorphism.
    Since $\R$ consists of a single tuple and $Q$ satisfies~\eqref{quantum_homo_2}, one can write $Q$ as $\oplus_{i=1}^d h_i$, where each $h_i$ is a homomorphism $\R\qto (A;R)$.
    For each $i$, there exists $i(j)\in[n]$ such that $(h_i(1),\dots,h_i(r))=\tuple a^{i(j)}$.
    Define $g_i\colon \A^n\to \A$ to be the $n$-ary projection $\pi_{i(j)}$, and define $Q'\colon\A^n\qto\A$ to be $\oplus g_i$, concluding the proof of~\eqref{itm:qgadget-extension}.
\end{proof}

\begin{definition}
    Let $\A$ and $\B$ be structures on the same domain and potentially different signatures. We say that $\A$ \textit{q-defines} $\B$
    if there exists a q-definition for $\A$ and $R^\B$ for every relation $R^\B$ of $\B$.
\end{definition}

As an immediate corollary of~\Cref{galois-connection-oracular}, we obtain the following. 

\propgaloisoracular*
\begin{proof}
    It is enough to apply~\Cref{galois-connection-oracular} to each of the relations of $\B$.
\end{proof}

Classical gadgets for CSPs that can be defined one from the other via pp-definitions provide logspace reductions between CSPs of different languages.
It follows from the results above that the same holds for q-definitions.

 \begin{proposition}\label{qgadget-reduction}
Let $\A$ and $\B$ be structures such that $\A$ q-defines $\B$.
    Then $\qCSP(\rel B)$ reduces to $\qCSP(\rel A)$ in logspace.
\end{proposition}
\begin{proof}
    Applying~\Cref{prop_galois_oracular}, we know that $\qPol(\A)\subseteq\qPol(\B)$. In particular, the inclusion map gives a minion homomorphism $\qPol(\A)\to\qPol(\B)$. Then, the result follows from~\Cref{thm_main_qPol_homo_reductions}.
\end{proof}

Next, we show that q-definitions are compatible with quantum homomorphisms, in the following sense.
\begin{proposition}\label{qgadget-composition}
    Let $\rel A,\rel B$ be structures, and let $R_A\subseteq A^r$ and $R_B\subseteq B^r$ be relations that admit a common q-definition\/ $\GG$ over $\rel A$ and $\rel B$, respectively. 
    Then every quantum homomorphism $Q\colon\rel A\qto\rel B$ is a quantum homomorphism $(A;R_A)\qto (B;R_B)$.
\end{proposition}
\begin{proof}
    Let $g_1,\dots,g_r$ be the distinguished vertices of $\GG$.

    We prove~\eqref{quantum_homo_1}. Let $(a_1,\dots,a_r)\in R_A$ and $(b_1,\dots,b_r)\not\in R_B$.
    By~\eqref{itm:qgadget-extension}, there exists a quantum homomorphism $Q'\colon\GG\qto\rel A$ such that $Q'_{g_i,a_i}=\id$ for all $i\in[r]$.
    Let $S=Q \bullet Q'$, which is a quantum homomorphism $\GG\qto\rel B$.
    By~\eqref{itm:qgadget-restriction}, we have that the quantum function 
    $Q''$ defined by $Q''_{i,b}:= S_{g_i,b}$ for all $i\in[r]$ and all $b\in B$
    is a quantum homomorphism $\rel R\qto (B;R_B)$.
    Thus, $\prod_{i=1}^rQ''_{i,b_i}=0$ and therefore $\prod_{i=1}^r S_{g_i,b_i}=0$.
    Note that $S_{g_i,b_i}$ is by definition $\sum_{a\in A} Q_{a,b_i}\otimes Q'_{g_i,a} = Q_{a_i,b_i}\otimes\id$.
    Thus, $(\prod_{i=1}^r Q_{a_i,b_i})\otimes\id = 0$, from which we obtain that $\prod_{i=1}^r Q_{a_i,b_i}=0$.

    We now turn to~\eqref{quantum_homo_2}.
    Let $a_1,a_2$ appear together in a tuple in $R_A$ at positions $i$ and $j$ and let $b_1,b_2\in B$ be arbitrary.
    Let $Q'\colon\GG\qto \rel A$ be such that $Q'_{g_i,a_1}=Q'_{g_j,a_2}=\id$, whose existence is asserted by~\eqref{itm:qgadget-extension}.
    Again, let $S\colon\GG\qto\rel B$ be the composition $Q\bullet Q'$ and $Q''\colon\rel R\qto (B;R_B)$ be the quantum homomorphism given by~\eqref{itm:qgadget-restriction}.
    We get 
    $0=[Q''_{i,b_1},Q''_{j,b_2}] = [S_{g_i,b_1},S_{g_j,b_2}] = [Q_{a_1,b_1}\otimes\id, Q_{a_2,b_2}\otimes \id]=[Q_{a_1,b_1},Q_{a_2,b_2}]\otimes\id$
    and therefore $[Q_{a_1,b_1},Q_{a_2,b_2}]=0$.
\end{proof}

The next step is to make q-definitions more expressive by allowing them to act on powers of the initial structure. This results in the notion of q-constructions, which is the quantum counterpart of classical pp-constructions.

\begin{definition}
\label{defn_q_construction}
Let $\A$ be a $\sigma$-structure and let $\B$ be a $\rho$-structure. We say that $\A$ \textit{q-constructs} $\B$ if there exists a $\rho$-structure $\C$ such that 
\begin{itemize}
    \item[$(i)$] the domain of $\C$ is $A^d$ for some $d\in\N$;
    \item[$(ii)$] each relation of $\rel C$ (of arity, say, $r$) is q-definable in $\A$ (identifying $C^r=(A^d)^r$ with $A^{dr}$);
    \item[$(iii)$] $\B$ and $\C$ are quantum-homomorphically equivalent.
\end{itemize}
\end{definition}

We now show that q-constructions give rise to homomorphisms between quantum polymorphism minions---and thus, via~\Cref{thm_main_qPol_homo_reductions}, to reductions between qCSPs. Namely, we prove the following.

\qconstructionimpliesqPolhom*

It turns out that the key to proving this is to show that the ``vectorisation step'' in the second condition of~\Cref{defn_q_construction} yields a minion homomorphism between the corresponding quantum polymorphism minions (and, thus, a reduction between the corresponding quantum CSPs). While for classical CSPs this fact is essentially obvious, in the case of $\qCSP$ one might expect that an additional commutativity gadget is needed to quantise this part of the classical reduction. We shall prove that, in fact, no commutativity gadget is needed.

We start by proving a result about abstract minion homomorphisms. 
Given a minion $\minion M$ and an $n$-ary element $Q\in\minion M^{(n)}$,
we say that a coordinate $j\in[n]$ is \emph{inessential} for $Q$ if there exists an integer $m\in\N$, a map $\sigma\colon[m]\to[n]$, and an element $Q'\in\minion M^{(m)}$ such that 
\begin{itemize}
    \item $Q=Q'_{/\sigma}$, and
    \item $j$ is not in the image of $\sigma$.
\end{itemize}
(cfr. e.g.~\cite[Remark~41]{cz23soda:minions}). A coordinate is \emph{essential} for $Q$ otherwise.
    An \emph{essential part} of $Q$ is a minor $Q_{/\tau}$ obtained via any function $\tau\colon[n]\to[n']$ that is a bijection when restricted to the essential coordinates of $Q$, and such that all coordinates of $Q_{/\tau}$ are essential. This implies in particular that $\tau$ must be surjective.
    Suppose that $Q_{/\tau}$ is an essential part of $Q$, and take a map $\tau'\colon[n']\to[n]$ such that $j=(\tau'\circ\tau)(j)$ holds for every essential coordinate $j$ of $Q$. One readily checks that, in this case, it must hold that $Q=(Q_{/\tau})_{/\tau'}$.
    Such a $\tau'$ exists given that $\tau$ is bijective when restricted to the essential coordinates of $Q$.

    \begin{lemma}\label{minion-essential}
    Let $\minion M,\minion N$ be minions and let $\xi\colon\minion M\to\minion N$ be a partial minion homomorphism defined on every $f\in\minion M$ that only has essential coordinates.
    Then $\xi$ extends to a minion homomorphism $\minion M\to\minion N$.
\end{lemma}
\begin{proof}
    Suppose $\xi$ is defined for elements with only essential coordinates, and let $f\in\minion M$ be arbitrary.
    Let $f_{/\tau}$ be an essential part of $f$, and let $\tau'$ be such that $\tau\circ\tau'=\id$.
    Let $\zeta$ be the extension of $\xi$ defined by $\zeta(f):=(\xi(f_{/\tau}))_{/\tau'}$. Since $f_{/\tau}$ only has essential coordinates, this definition is well posed.
    We prove simultaneously the following two claims: The definition of $\zeta$ does not depend on the choice of an essential part of $f$, and $\zeta$ is a minion homomorphism.
    To this end, suppose that $g=f_{/\sigma}$ holds.
    Let $g_{/\tau_g}$ be an essential part of $g$, and let $f_{/\tau_f}$ be an essential part of $f$.
    Let $\tau'_g$ and $\tau'_f$ be such that $\tau'_g\circ\tau_g=\id$ and $\tau'_f\circ\tau_f=\id$ hold when restricted to the essential coordinates of $g$ and $f$, respectively.
    Then we have that $g_{/\tau_g} = f_{/\tau_g\circ\sigma} = \left[(f_{/\tau_f})_{/\tau'_f}\right]_{/\tau_g\circ\sigma}=(f_{/\tau_f})_{/\tau_g\circ\sigma\circ\tau'_f}$.
    It follows that
    \begin{align*}
        \zeta(g) &= (\xi(g_{/\tau_g}))_{/\tau'_g} = \left[\xi(f_{/\tau_f})_{/\tau_g\circ\sigma\circ\tau'_f}\right]_{/\tau'_g}= \xi(f/_{\tau_f})_{/\sigma\circ\tau'_f}
        =(\xi(f_{/{\tau_f}})_{/\tau'_f})_{/\sigma} = \zeta(f)_{/\sigma}.
    \end{align*}
    This proves that $\zeta$ is a minion homomorphism, and that it is well defined by taking $g=f$ and $\sigma=\id$ in this claim, with potentially two different essential parts $\tau_f,\tau_g$.
\end{proof}

\begin{proof}[Proof of~\Cref{q_construction_implies_qPol_hom}]
    For ease of notation, we assume here that $\B$ has a single relation with symbol $R$; the proof in the case of multiple relations is entirely analogous.
    Let $r$ be the arity of $R$.
    Let $\C=(A^d;R^{\C})$ be as in the definition of a q-construction, and let $\tilde R\subseteq A^{dr}$ be its vectorisation.
    Thus, $(a_{1,1},\dots,a_{1,d},\dots,a_{r,1},\dots,a_{r,d})\in\tilde R$ if, and only if, $(\tuple a_1,\dots,\tuple a_r)\in R^{\C}$.
    We will establish the chain of minion homomorphisms \[\qPol(\A)\to\qPol(A;\tilde R)\to\qPol(\C)\to\qPol(\B).\]

    The first arrow in this chain follows from~\Cref{galois-connection-oracular}.
    For the last arrow in this chain, let $U\colon\C\qto\B$ and $W\colon\B\qto\C$.
    Defining $\xi\colon\qPol(\C)\to\qPol(\B)$ as $\xi(Q):=(U\bullet Q)*(W,\dots,W)$ gives the required minion homomorphism.

    The only non-trivial case is the arrow $\zeta\colon\qPol(A;\tilde R)\to\qPol(\C)$.
    We have shown in~\Cref{minion-essential} that, to define a minion homomorphism $\zeta\colon\qPol(A;\tilde R)\to\qPol(\C)$, it suffices to define it on those elements $\tilde Q$ that only have essential coordinates.
    Let $\tuple a_0$ be an arbitrary tuple in $A^d$.
    Given $\tilde Q\in\qPol(A;\tilde R)$ of arity $n$ and $\tuple a_1,\dots,\tuple a_n\in A^d$, we define $Q:=\zeta(\tilde Q)$ according to two cases.
    Suppose that there is no $i\in[r]$ such that all $\tuple a_1,\dots,\tuple a_n$ are in the projection of $R^{\C}$ onto its $i$th coordinate.
    Then define $Q_{(\tuple a_1,\dots,\tuple a_n),\tuple a_0}=\id$ (which defines the entire PVM for this input).
    Otherwise, define it to be $Q_{(\tuple a_1,\dots,\tuple a_n),\tuple b}=\prod_{k=1}^d \tilde Q_{(a_{1,k},\dots,a_{n,k}),b_k}$ for all $\tuple b\in A^d$.
    
    We now show that $Q$ consists of PVMs.
    If $\tuple a_1,\dots,\tuple a_n$ are as in the first case, then there is nothing to show.
    So let us assume that $\tuple a_1,\dots,\tuple a_n$ all lie in the $i$th projection of $R^{\C}$.
    Let $\GG$ be a q-definition for $\tilde R$ in $\A$, with its distinguished vertices $g_{i,k}$ for $i\in[r]$ and $k\in[d]$.
    By~\eqref{itm:qgadget-extension}, for $j\in[n]$, there exist homomorphisms $h_j\colon\GG\to\A$ such that $h_j(g_{i,k})=a_{j,k}$ for all $k\in[d]$.
    Let $U=\tilde Q\bullet(h_1\otimes\dots\otimes h_n)\colon\GG\to\A$, which can be explicitly computed by
    $U_{g_{i,k},b}=Q_{(a_{1,k},\dots,a_{n,k}),b}$ for all $b\in A$.
    By~\eqref{itm:qgadget-restriction}, we must have $[U_{g_{i,k},b},U_{g_{i,k'},c}]=0$ for all $k,k'\in[d]$ and $b,c\in A$.
    Thus, we get $[Q_{(a_{1,k},\dots,a_{n,k}),b},Q_{(a_{1,k'},\dots,a_{n,k'}),c}]=0$ for all $k,k'\in[d]$ and $b,c\in A$.
    It follows that all the factors in the definition of $Q_{(\tuple a_1,\dots,\tuple a_n),\tuple b}$ commute, so that the operator is indeed a projection.
    A simple computation shows that $\sum_{\tuple b} Q_{(\tuple a_1,\dots,\tuple a_n),\tuple b} = \id$ holds for all $\tuple a_1,\dots,\tuple a_n\in A^d$, so that $Q$ is indeed a quantum function.

    Next, we show that $Q$ is in $\qPol(\C)$, starting with~\eqref{quantum_homo_1}.
    To this end, let $(\tuple a_{1,j},\dots,\tuple a_{r,j})\in R^{\C}$ for all $j\in[n]$ and let $(\tuple b_1,\dots,\tuple b_r)\not\in R^{\C}$.
    Note that, for all $i\in[r]$, we have $Q_{(\tuple a_{i,1},\dots,\tuple a_{i,n}),\tuple b_i}=\prod_{k=1}^d\tilde Q_{(a_{i,1,k},\dots,a_{i,n,k}),b_k}$ (i.e., corresponding to the second case of the definition of $Q$), since $\tuple a_{i,1},\dots,\tuple a_{i,n}$ all belong to the $i$th projection of $R$.
    Thus, we obtain
    \begin{align*}
        \prod_{i=1}^r Q_{(\tuple a_{i,1},\dots,\tuple a_{i,n}),\tuple b_i} &= \prod_{i=1}^r\prod_{k=1}^d \tilde Q_{(a_{i,1,k},\dots,a_{i,n,k}),b_k}=0,
    \end{align*}
    where the second equality follows from~\eqref{quantum_homo_1} applied to $\tilde Q$, since every $dr$-tuple of the form
    $$(a_{1,j,1},\dots,a_{1,j,d},\dots,a_{r,j,1},\dots,a_{r,j,d})$$ is in $\tilde R$ for $j\in[n]$, while the $dr$-tuple $(b_{1,1},\dots,b_{r,d})$ is not.

    We now turn to~\eqref{quantum_homo_2}.
    Using the same notation as in the previous paragraph, we aim to show that
    \begin{align}
    \label{eqn_1531_0104}
    [Q_{(\tuple a_{i,1},\dots,\tuple a_{i,n}),\tuple b}, Q_{(\tuple a_{i',1},\dots,\tuple a_{i',n}),\tuple c}]=0 \end{align}
    for all $i,i'\in[r]$ and all $\tuple b,\tuple c\in A^d$.
    Using the same argument as above, we note that the PVMs for $Q$ are defined as in the second case, since all the tuples belong to the $i$th (resp.\ $i'$th) projection of $R^{\C}$.
    Thus,~\eqref{eqn_1531_0104} implies 
    \[ \left[\prod_{k=1}^d Q_{(a_{i,1,k},\dots,a_{i,n,k}),\tuple b}, \prod_{k'=1}^d Q_{{(a_{i',1,k'},\dots,a_{i',n,k'}),\tuple c}}\right]=0.\]
    We now show that any two factors in this expression commute, for $i,i'\in[r]$ and $k,k'\in[d]$.
    Indeed, by~\eqref{itm:qgadget-extension}, there is for every $j\in[n]$ a homomorphism $h_j\colon\GG\to\A$ such that, for all $i''\in[r]$ and all $k''\in[d]$, it holds that $h_j(g_{i'',k''}) = a_{i'',j,k''}$.
    Define $U:=\tilde Q\bullet(h_1\otimes\dots\otimes h_n)$, so that
    $U_{g_{i'',k''},c} = \tilde Q_{(a_{i'',1,k''},\dots,a_{i'',n,k''}),c}$.
    In particular, for $(i'',k'')=(i,k)$ and $(i'',k'')=(i',k')$, we obtain by~\eqref{itm:qgadget-restriction} the desired claim.
    This concludes the proof of~\eqref{quantum_homo_2}.

    We have now established that $\zeta$ takes values in $\qPol(\C)$, and it remains to prove that it is a minion homomorphism when restricted to polymorphisms having only essential coordinates.
    Suppose that $\tilde W=\tilde Q_{/\sigma}$ holds in $\qPol(A;\tilde R)$, where both $\tilde W$ and $\tilde Q$ only have essential coordinates.
    Let $n$ be the arity of $W$ and $m$ be the arity of $Q$.
    Take $\tuple a_1,\dots,\tuple a_n\in A^d$ and let $\tuple b\in A^d$.
    If there exists an index $i\in[r]$ such that all $\tuple a_1,\dots,\tuple a_n$ belong to the $i$th projection of $R^{\C}$, then the same is true in particular for $\tuple a_{\sigma(1)},\dots,\tuple a_{\sigma(m)}$.    
    Thus, the definition of $\zeta$ gives
    \begin{align*} \zeta(\tilde W)_{(\tuple a_1,\dots,\tuple a_n),\tuple b}=\prod_{k=1}^d \tilde W_{(a_{1,k},\dots,a_{n,k}),b_k}
    =\prod_{k=1}^d \tilde Q_{(a_{\sigma(1),k},\dots,a_{\sigma(m),k}),b_k}=\zeta(\tilde Q)_{(\tuple a_{\sigma(1)},\dots,\tuple a_{\sigma(m)}),\tuple b}.
    \end{align*}
    Suppose otherwise that $\tuple a_1,\dots,\tuple a_n$ are not all in the same projection of $R^{\C}$.
    Then the same must be true for $\tuple a_{\sigma(1)},\dots,\tuple a_{\sigma(m)}$, since $\sigma$ is surjective due to $\tilde W$ having only essential coordinates.
    Then $\zeta(\tilde W)_{(\tuple a_1,\dots,\tuple a_n),\tuple b}$ is the identity projector if $\tuple b=\tuple a_0$, and is $0$ otherwise.
    The same is true for $\zeta(\tilde Q)_{(\tuple a_{\sigma(1)},\dots,\tuple a_{\sigma(m)}),\tuple b}$, and therefore $\zeta(\tilde Q)_{/\sigma}$ is indeed equal to $\zeta(\tilde W)$.
\end{proof}

As an immediate consequence of the above result, we obtain the following.

\begin{corollary}
\label{lem_q_constructions_means_reudction}
    If $\A$ q-constructs $\B$, then $\qCSP(\rel B)$ reduces to $\qCSP(\rel A)$ in logspace.%
\end{corollary}

\begin{proof}
    The result directly follows from~\Cref{q_construction_implies_qPol_hom} and~\Cref{thm_main_qPol_homo_reductions}.
\end{proof}

\subsection{Commutativity gadgets}
\label{subsec_body_commutativity_gadgets}
It is not hard to show that every q-definition yields a classical pp-definition.
Indeed, considering~\eqref{itm:qgadget-extension} for quantum homomorphisms with a Hilbert space of dimension 1 (which corresponds to classical  homomorphisms), we find that the tuples in $R$ are in one-to-one correspondence with the restrictions of homomorphisms $\GG\to\rel A$ to $\{g_1,\dots,g_r\}$.
However, not every pp-definition is a q-definition. 
It turns out that one can build a q-definition out of a pp-definition in a ``uniform'' way by pairing it with an extra construction making sure that the process of replacing constraints with gadgets preserves the existence of quantum homomorphisms. This is the notion of commutativity gadgets from~\cite{Ji}, which we define below. Note that, in our terminology, a commutativity gadget is simply a q-definition for the complete binary relation. 

\begin{definition}\label{comm-gadget}
    Fix a $\sigma$-structure $\rel A$.
    A \emph{commutativity gadget} for $\rel A$ is a q-definition for the relation $A^2$.
    In other words, a commutativity gadget is a $\sigma$-structure $\GG$ with two distinguished elements $u,v\in G$ satisfying the following properties:
    \begin{itemize}
        \item[$(\mathbf{c_1})$\labeltext{$\mathbf{c_1}$}{itm:commgadget-extension}] For any quantum function $Q\colon\{u,v\}\qto A$ such that $Q_{u,a}$ and $Q_{v,b}$ commute for all $a,b\in A$, %
        there exists a quantum homomorphism $Q'\colon\GG\qto\rel A$ extending $Q$.
        \item[$(\mathbf{c_2})$\labeltext{$\mathbf{c_2}$}{itm:commgadget-restriction}] For every quantum homomorphism $Q\colon\GG\qto\rel A$ and all $a,b\in A$, it holds that $[Q_{u,a},Q_{v,b}]=0$.
    \end{itemize}
\end{definition}

\begin{remark}
\label{rem_algebraic_gadgets}
	In~\cite{Zeman}, an \emph{algebraic commutativity gadget} is defined by replacing~\eqref{itm:commgadget-extension} by the following condition:
    \begin{itemize}
        \item[$(\mathbf{c_1'})$\labeltext{$\mathbf{c_1'}$}{itm:commgadget-extension_zeman}] For every $a,b\in A$, there exists a quantum homomorphism $Q\colon\GG\qto\rel A$ such that $Q_{u,a}=\id$ and $Q_{v,b}=\id$.
    \end{itemize}
    Note that this is a special case of~\Cref{itm:commgadget-extension} and, thus, any commutativity gadget is an algebraic commutativity gadget.
	While this appears to be a strictly weaker requirement than the one in~\Cref{comm-gadget}, we shall prove below that the existence of an algebraic commutativity gadget in the sense of~\cite{Zeman} is equivalent to the existence of a commutativity gadget in our sense.
\end{remark}
If $\rel A$ has a commutativity gadget, then we have a natural way of obtaining a complete and sound reduction from $\qCSP(\rel B)$ to $\qCSP(\rel A)$ out of a classical gadget: For every pair of variables $x,y$ in  a relation in an instance of $\qCSP(\rel B)$, glue to $x,y$ a copy of the commutativity gadget, identifying $x$ with $u$ and $y$ with $v$. Condition~\eqref{itm:commgadget-restriction} in~\Cref{comm-gadget} forces the associated projectors to commute in every quantum homomorphism to $\rel A$, and condition~\eqref{itm:commgadget-extension} tells us that no additional constraint is forced upon $x,y$ by gluing the commutativity gadget.
We can view this fact at the relational level, as follows.
\begin{proposition}
\label{prop_pp_plus_comm_equals_q}
    Let $\rel A,\rel B$ be structures such that $(i)$ $\A$ has a commutativity gadget, and $(ii)$ $\A$ pp-defines $\B$.
    Then $\A$ q-defines $\B$.
\end{proposition}
\begin{proof}
Let $S=R^\B$ be a relation of $\B$ of some arity $r$, and let $\GG$ be a classical pp-definition for $R^\B$ in $\rel A$, with distinguished vertices $g_1,\dots,g_r$.
    Let $\rel H$ be a commutativity gadget for $\rel A$, with distinguished vertices $u,v$.
    We let $\GG'$ be the structure that initially consists of the disjoint union of $\GG$ and a copy $\rel H_{x,y}$ of $\rel H$ for each pair of distinct vertices $x,y$ in $G$, and in which we identify each such pair $x,y$ with the vertices $u,v$ from the copy $\rel H_{x,y}$ of $\rel H$.

    We show that $\GG'$ (with distinguished vertices $g_1,\dots,g_r$) is a q-definition for $S$ in $\rel A$.
    Let $Q\colon\rel R\qto (A;S)$ be a quantum homomorphism.
    In particular, $Q$ is non-contextual and can therefore be written as $h_1\oplus\dots\oplus h_d$
    for some homomorphisms $h_k\colon\rel R\to (A;S)$.
    Since $\GG$ is a pp-definition for $S$ in $\rel A$, it follows that there exist homomorphisms $f_k\colon \GG\to\rel A$ such that $f_k(g_i)=h_k(i)$ for all $i\in[r]$.
    Then $Q'=\bigoplus_{k\in[d]} f_k$ is a non-contextual quantum homomorphism $\GG\qto\rel A$.
    In particular, $Q'_{x}$ and $Q'_{y}$ commute for all $x,y\in G$, and it follows from~\eqref{itm:commgadget-extension} that for all $x,y$ there exists a quantum homomorphism
    $Q^{(x,y)}\colon \rel H\qto\rel A$ such that $Q^{(x,y)}_{u}=Q'_{x}$ and $Q^{(x,y)}_{v}=Q'_{y}$.
    We thus obtained a quantum homomorphism $\GG'\qto\rel A$, which proves~\eqref{itm:qgadget-extension}.

    Let now $Q\colon\GG'\qto\rel A$ be a quantum homomorphism.
    For each $x,y\in G$, we have by~\eqref{itm:commgadget-restriction} that $Q_x$ and $Q_y$ commute.
    Thus, $Q$ restricts to a quantum homomorphism $\GG\qto\rel A$ that is non-contextual and can thus be written as $h_1\oplus\dots\oplus h_d$, each $h_k$ being a classical homomorphism $\GG\to\rel A$.
    Since $\GG$ is a classical pp-definition for $R^\B$, we obtain that $(h_k(g_1),\dots,h_k(g_r))\in R^\B$
    for all $k\in[d]$.
    Thus, the quantum function $Q'$ defined by $Q'_{i,a} = \bigoplus_k (h_k)_{g_i,a}$
    for $i\in [r],a\in A$
    is a quantum homomorphism $\rel R\qto (A;S)$.
\end{proof}

Next, we prove the analogue of~\Cref{prop_pp_plus_comm_equals_q} in the more general case of q-constructions.

\begin{restatable}{theorem}{propppconstructionpluscommequalsqconstruction}
\label{prop_pp_construction_plus_comm_equals_q_construction}
    Let $\rel A,\rel B$ be structures such that $(i)$ $\A$ has a commutativity gadget, and $(ii)$ $\A$ pp-constructs $\B$.
    Then $\A$ q-constructs $\B$.
\end{restatable}

\begin{proof}
    It is well known that condition $(ii)$ is equivalent to the fact that $\B$ is homomorphically equivalent to a pp-power $\C$ of $\A$ (for instance, see~\cite{BKW17,BBKO21}). Let $\D$ be the structure with  domain $D=A$ obtained by identifying $(A^d)^r$ with $A^{dr}$ for each $r$-ary relation of $\C$, as in~\Cref{defn_q_construction}. Since $\C$ is a pp-power of $\A$, we have that $\D$ is pp-definable from $\A$. Using that $\A$ admits a commutativity gadget, it follows from~\Cref{prop_pp_plus_comm_equals_q} that $\D$ is in fact q-definable from $\A$. But this precisely means that $\A$ q-constructs $\B$, and we are done. 
\end{proof}

\subsection{Complexity reductions via non-contextual quantum polymorphisms}
\label{subsec_body_quantum_pol_comm_gadgets}

We now have all ingredients to classify the languages for which commutativity gadgets exist in terms of their quantum polymorphism minion---thus proving the following theorem, which is an extended version of~\Cref{gadget-characterization_overview_friendly}.

\begin{theorem}\label{gadget-characterization}
Let $\rel A$ be a relational structure. The following are equivalent:
\begin{enumerate}
    \item $\rel A$ has a commutativity gadget.
    \item $\rel A$ has an algebraic commutativity gadget.
    \item $\rel A^n$ is an (algebraic) commutativity gadget for all large enough $n\in\N$.
    \item $\qPol(\rel A)=\qcPol(\rel A)$.
\end{enumerate}
\end{theorem}

%

%
The next definition shall be useful.

\begin{definition}
    Let $\A$ be a structure. The pairs $(a_1,b_1),\dots,(a_n,b_n)\in A^2$ are called \textit{generators of $A^2$ in $\Pol(\A)$} if for every pair $(c,d)\in A^2$, there exists $g\in\Pol(\rel A)$ of arity $n$ such that $g(a_1,\dots,a_n)=c$ and $g(b_1,\dots,b_n)=d$.
\end{definition}
Note that, if $A$ is finite (as is always the case in the current paper), $A^2$ admits a family of $n=|A|^2$ generators. Indeed, we can simply let $g$ be the $i$-th projection $\rel A^n\to \rel A$, with $i\in [n]$ being the index of the pair $(c,d)$. However, $n$ could be smaller for certain $\rel A$.
This is for example the case for complete graphs (where $n=2$ is enough) or for undirected cycles (where $m$ is enough for an odd cycle of length $2m+1$).
We point out that there are many cases of infinite structures $\rel A$ where $A^2$ is finitely generated under $\Pol(\rel A)$; this is for example the case for the class of $\omega$-categorical structures, which has been widely studied in the context of classical constraint satisfaction problems.
\begin{proposition}\label{obtaining-commgadget}
    Let $\rel A$ be a (not necessarily finite) structure such that $A^2$ admits a family of $n$ generators under $\Pol(\rel A)$.
    Suppose also that $\qPol(\rel A)=\qcPol(\rel A)$.
    Then $\rel A^n$ is a commutativity gadget for $\A$.
\end{proposition}
\begin{proof}
    Let $(a_1,b_1),\dots,(a_n,b_n)$ be generators for $A^2$ under $\Pol(\A)$.
    Let $\GG$ be $\rel A^n$ and the two distinguished elements be $\tuple a=(a_1,\dots,a_n)$ and $\tuple b=(b_1,\dots,b_n)$.
    
    Every quantum homomorphism $\GG\qto\rel A$ is a quantum polymorphism of $\rel A$ and therefore~\eqref{itm:commgadget-restriction} holds by the assumption that $\qPol(\rel A)=\qcPol(\rel A)$.
    Moreover, by the spectral theorem, every quantum function $Q\colon\{\tuple a,\tuple b\}\qto A$ such that $Q(\tuple a)$ commutes with $Q(\tuple b)$ is of the form $h_1\oplus\cdots \oplus h_d$, where $h_i\colon\{\tuple a,\tuple b\}\to A$ is a classical function.
    Using that $(a_1,b_1),\dots,(a_n,b_n)$ are generators for $A^2$, we know that, for each $i$, there exists $g_i\in\Pol(\rel A)$ such that $g_i(\tuple a)=h_i(\tuple a)$ and $g_i(\tuple b)=h_i(\tuple b)$.
    Then $Q$ extends to a quantum homomorphism $Q':=g_1\oplus\cdots\oplus g_d$, thus proving~\eqref{itm:commgadget-extension}.
\end{proof}

We now prove that the existence of a commutativity gadget guarantees that all quantum polymorphisms are classical. In fact, we prove that the same fact holds even for the algebraic commutativity gadgets from~\Cref{rem_algebraic_gadgets}.

\begin{proposition}\label{commgadget-implies-closure}
    Let $\rel A$ be a structure.
    If $\rel A$ has an algebraic commutativity gadget, then $\qPol(\rel A)=\qcPol(\rel A)$.
\end{proposition}
\begin{proof}
    Let $\GG$ be an algebraic commutativity gadget with distinguished elements $u,v\in G$.
    Let $Q\colon\rel A^n\qto\rel A$ be a quantum polymorphism, and choose two tuples $\tuple a,\tuple b\in A^n$ and two elements $c,d\in A$.
    We want to show that $[Q_{\tuple a,c},Q_{\tuple b,d}]=0$, which is enough to conclude by~\Cref{observation}.

    Using~\eqref{itm:commgadget-extension_zeman} we know that, for each $i\in[n]$, there exists a quantum homomorphism $W^{(i)}\colon\GG\qto\rel A$ such that 
    \begin{align*}
    W^{(i)}_{u,a_i}=W^{(i)}_{v,b_i}=\id.
    \end{align*}
    Consider now the quantum homomorphism $W\colon\GG\qto\rel A^n$ given by the
    tensor product 
    \[W=\bigotimes_{i\in[n]}W^{(i)}\]
    and
    observe that $W_{u,\tuple a}=W_{v,\tuple b}=\id$. Take now the composition \[U=Q\bullet W\colon\GG\qto\A.\] Using~\eqref{itm:commgadget-restriction}, we observe that $U$ satisfies $[U_{u,c},U_{v,d}]=0$.
    However, expanding the composition as per~\Cref{composition},
    we have 
	\[
        U_{u,c} = \sum_{\tuple t\in A^n}Q_{\tuple t,c}\otimes W_{u,\tuple t}
        = Q_{\tuple a,c}\otimes\id
    \] and, similarly, $U_{v,d}=Q_{\tuple b,d}\otimes\id$.
    Thus, 
    \[[Q_{\tuple a,c},Q_{\tuple b,d}]\otimes\id = [Q_{\tuple a,c}\otimes\id, Q_{\tuple b,d}\otimes\id]=[U_{u,c},U_{v,d}]=0,\]
    which implies that $[Q_{\tuple a,c},Q_{\tuple b,d}]=0$, and we are done.
\end{proof}

\begin{proof}[Proof of~\Cref{gadget-characterization}]
    The result is a direct consequence of~\Cref{obtaining-commgadget,commgadget-implies-closure} and the fact that any commutativity gadget is an algebraic commutativity gadget (as observed in~\Cref{rem_algebraic_gadgets}).
\end{proof}

As a straightforward consequence of the above, we obtain the following result on reducibility between quantum CSPs.

\begin{restatable}{theorem}{thmcommgadgetsplusminionhomomorphismsgivesreduction}
\label{thm_comm_gadgets_plus_minion_homomorphisms_gives_reduction}
    Let $\A,\B$ be relational structures such that 
 $\qPol(\A)=\qcPol(\A)$ and $\Pol(\A)\to\Pol(\B)$.    
    Then $\qCSP(\rel B)$ reduces to $\qCSP(\rel A)$ in logspace.
\end{restatable}

\begin{proof}
The fact that $\qPol(\A)=\qcPol(\A)$ is equivalent to $\A$ admitting a commutativity gadget by~\Cref{gadget-characterization}, while the fact that $\Pol(\A)\to\Pol(\B)$ is equivalent to $\A$ pp-constructing $\B$ by~\Cref{minion_homo_q_construction}. Using~\Cref{prop_pp_construction_plus_comm_equals_q_construction}, we deduce that $\A$ q-constructs $\B$. Hence, there exists a minion homomorphism $\qPol(\A)\to\qPol(\B)$ by~\Cref{q_construction_implies_qPol_hom}, and the result follows from~\Cref{thm_main_qPol_homo_reductions}.
\end{proof}

We immediately derive from the above theorem the following useful undecidability criterion for quantum CSPs. As we shall see later, this generalises the known undecidability criteria for quantum CSPs.

\begin{corollary}
\label{cor_np_complete_plus_no_contextual_means_undecidable_ALGEBRAIC}
    Let $\A$ be a structure such that $\Pol(\A)\to\Pol(\operatorname{3SAT})$ and all quantum polymorphisms of $\A$ are non-contextual. Then, $\qCSP(\A)$ is undecidable.
\end{corollary}
\begin{proof}
   Since all quantum polymorphisms of $\A$ are non-contextual, we deduce that $\qPol(\A)=\qcPol(\A)$. 
    Hence, by~\Cref{{thm_comm_gadgets_plus_minion_homomorphisms_gives_reduction}}, there is a logspace (in particular, computable) reduction from $\qCSP(\operatorname{3SAT})$ to $\qCSP(\A)$. The result then follows from the undecidability of $\qCSP(\operatorname{3SAT})$~\cite{ji2021mip,mastel2024two}. 
\end{proof}

The result above can be rephrased in the following, convenient form.

\cornpcompleteplusnocontextualmeansundecidable*
\begin{proof}
    Since $\CSP(\A)\not\in\mathsf{P}$, it follows from the algebraic CSP dichotomy theorem~\cite{Bulatov17:focs,Zhuk17_FOCS} that $\A$ pp-constructs $\operatorname{3SAT}$ and thus, by~\Cref{minion_homo_q_construction}, $\Pol(\A)\to\Pol(\operatorname{3SAT})$. To conclude, we apply~\Cref{cor_np_complete_plus_no_contextual_means_undecidable_ALGEBRAIC}.
\end{proof}

Note that~\Cref{cor_np_complete_plus_no_contextual_means_undecidable} is in fact equivalent to~\Cref{cor_np_complete_plus_no_contextual_means_undecidable_ALGEBRAIC} if $\mathsf{P}\neq\NP$. Indeed, under such hypothesis, $\Pol(\A)\to\Pol(\operatorname{3SAT})$ if, and only if, $\CSP(\A)\not\in\mathsf{P}$. On the other hand, if $\mathsf{P}=\NP$,~\Cref{cor_np_complete_plus_no_contextual_means_undecidable} becomes vacuous (unlike~\Cref{cor_np_complete_plus_no_contextual_means_undecidable_ALGEBRAIC}).

\begin{example}
Recall~\Cref{fig:reduction-C5} in~\Cref{sec_overview}. The figure illustrates the undecidability of $\qCSP(\rel C_5)$ by reducing from quantum 1-in-3-SAT; i.e., from $\qCSP(\A)$ where $\A$ is the language encoding 1-in-3-SAT formulae. Formally, $\A$ is the Boolean structure with a single, ternary relation
\begin{align*}
    R^\A=\{(0,0,1),(0,1,0),(1,0,0)\}.
\end{align*}
The undecidability of quantum 1-in-3-SAT can be obtained by reducing from quantum 3SAT, using for example that $\Pol(\A)$ does not contain the majority polymorphism, see~\cite{culf2024re} and~\Cref{sec_boolean}. 
\end{example}

\subsection{Comparison with known undecidability criteria}
\label{subsec_comparison_TVF}
The following undecidability criterion 
follows by combining~\cite[Lemma~10]{AtseriasKS19} with the undecidability of $\qCSP(\rel{XOR})$ from~\cite{slofstra2019set}, as pointed out in~\cite{paddock2025satisfiability}. (For the formal definition of $\rel{XOR}$, see~\Cref{sec_boolean}.)

\begin{theorem}[\protect{\cite{AtseriasKS19,slofstra2019set,paddock2025satisfiability}}]
\label{thm_atserias_undecidability}
    Let $\A$ be a Boolean relational structure such that $(i)$ its signature contains a binary symbol $R$ for which $R^\A=A^2$, and $(ii)$ $\CSP(\A)\not\in\mathsf{P}$. Then, $\qCSP(\A)$ is undecidable.
\end{theorem}

This result was generalised in~\cite{culf2024re} by replacing the requirement $(i)$ on the existence of a complete binary relation with a weaker structural condition called two-variable (non-)falsifiability, and relaxing the requirement on $\A$ being Boolean. 
Given any relation $S\subseteq A^r$ and two indices $i\neq j\in [r]$, we denote by $S_{i,j}$ the binary projection of $S$ onto the coordinates $i,j$; i.e., 
\begin{align*}
    S_{i,j}=
    \bigcup_{(a_1,\dots,a_r)\in S}\{(a_i,a_j)\}.
\end{align*}

\begin{definition}[\cite{culf2024re}]
    Let $\A$ be a relational structure with signature $\sigma$. We say that $\A$ is \textit{two-variable falsifiable (TVF)} if for each symbol $R\in\sigma$ and each pair $i\neq j\in[\ar(R)]$, it holds that $(R^\A)_{i,j}\neq A^2$. 
\end{definition}

\begin{theorem}[\protect{\cite[Theorem~4.14]{culf2024re}}]
\label{thm_culf_mastel_undecidability}
    Let $\A$ be a relational structure such that $(i)$ $\A$ is not TVF, and $(ii)$ $\CSP(\A)\not\in\mathsf{P}$. Then, $\qCSP(\A)$ is undecidable.\footnote{We point out that, just like for our~\Cref{cor_np_complete_plus_no_contextual_means_undecidable}, condition $(ii)$ in both~\Cref{thm_atserias_undecidability} and~\Cref{thm_culf_mastel_undecidability} can be replaced by the algebraic condition $\Pol(\A)\to\Pol(\operatorname{3SAT})$, which does not trivialise if $\mathsf{P}=\NP$.} 
\end{theorem} 

Clearly, any structure satisfying the requirements of~\Cref{thm_atserias_undecidability} also satisfies the requirements of~\Cref{thm_culf_mastel_undecidability}, so the latter is indeed a generalisation of the former.   
We now prove that our undecidability criterion~\Cref{cor_np_complete_plus_no_contextual_means_undecidable} extends the one in~\Cref{thm_culf_mastel_undecidability}. To that end, we prove the following.

\begin{proposition}
\label{prop_nonTVF_implies_non_contextual}
Let $\A$ be a relational structure. If $\A$ is not TVF, then $\qPol(\A)=\qcPol(\A)$.
\end{proposition}

Note that, in fact,~\Cref{cor_np_complete_plus_no_contextual_means_undecidable} is \textit{strictly} more general than~\Cref{thm_culf_mastel_undecidability}. Indeed, there are multiple examples of structures $\A$ such that $(i)$ $\CSP(\A)$ is NP-complete, $(ii)$ $\qPol(\A)=\qcPol(\A)$, but $(iii)$ $\A$ is TVF. For instance, as we shall see in later sections, the Siggers digraph, any odd cycle, and any clique give examples of such $\A$.

\begin{proof}[Proof of~\Cref{prop_nonTVF_implies_non_contextual}]
    Let $\A$ be a $\sigma$-structure that is not TVF, and let $Q\colon\A^n\qto\A$ be a quantum polymorphism of some arity $n\in\N$. Take two tuples $\ba,\bb\in A^n$. We claim that the PVMs $Q_\ba$ and $Q_\bb$ commute.
        Since $\A$ is not TVF, there exists some symbol $R\in\sigma$ (of arity, say, $r$) and two indices $i\neq j\in [r]$ such that $R^\A$ has full binary projection onto the coordinates $i,j$; i.e., $(R^\A)_{i,j}=A^2$. Therefore, we can assign to each $\ell\in [n]$ a tuple $\textbf{t}^{(\ell)}=(t_1^{(\ell)},\dots,t_r^{(\ell)})\in R^\A$ such that 
        \begin{align}
        \label{eqn_1800_1803}(t_i^{(\ell)},t_j^{(\ell)})=(a_\ell,b_\ell). 
        \end{align}
        Let now $T$ be the $n\times r$ matrix whose rows are $\textbf{t}^{(1)},\textbf{t}^{(2)},\dots,\textbf{t}^{(n)}$. By definition of the categorical power $\A^n$, the $r$-tuple consisting of the list of columns of such matrix belongs to $R^{\A^n}$. This means that any pair of columns of $T$ having distinct indices is adjacent in $\Gaif(\A^n)$. Note now that, because of~\eqref{eqn_1800_1803}, the tuples $\ba$ and $\bb$ occur in $T$ as the $i$-th and $j$-th columns, respectively. We deduce that $\ba$ and $\bb$ are adjacent in $\Gaif(\A^n)$. Using~\eqref{quantum_homo_2}, we conclude that $Q_\ba$ and $Q_\bb$ must commute, as claimed.
\end{proof}
We point out that~\Cref{prop_nonTVF_implies_non_contextual} can alternatively be proved by observing that any relation having a full binary projection (and, thus, witnessing that $\A$ is not TVF) can be used to directly give a commutativity gadget for $\A$; then, the result follows from~\Cref{gadget-characterization}.

\section{Quantum polymorphisms of cliques}
\label{sec_cliques}

In this section, we give a first simple example of using quantum polymorphisms to obtain commutativity gadgets, in the setting of cliques. We let $\rel K_m$ denote the $m$-clique. We prove the following.

\begin{theorem}\label{Kn-closure}
    For all $m\geq 3$,
    $\qPol(\rel K_m)=\qcPol(\rel K_m)$.
\end{theorem}

\begin{proof}
    The proof goes by induction on $m$. 
    
    For the case of $m=3$, it follows from~\cite[Lemma 4]{Ji} that $\rel K_3$ has a commutativity gadget.
    Then,~\Cref{gadget-characterization} implies the result.
    
Take now $m\geq 4$.
    Let $Q\colon(\rel K_m)^n\qto\rel K_m$ be a quantum homomorphism over a Hilbert space $H$.
    Pick $\tuple u,\tuple v\in[m]^n$ and $i,j\in [m]$. Our goal is to show that $[Q_{\tuple u,i},Q_{\tuple v,j}]=0$.

Choose $\tuple w\in [m]^n$ in such a way that $(\tuple u,\tuple w)$ and $(\tuple v,\tuple w)$ are both edges of $(\rel K_m)^n$. (Note that this is possible as long as $m\geq 3$.)
Choose also $k\in[m]\setminus\{i,j\}$.
    We first claim that  
    \begin{align}
        \label{eqn_1431_0110}
        [Q_{\tuple u,i}Q_{\tuple w,k},Q_{\tuple v,j}Q_{\tuple w,k}]=0.
    \end{align}
    This is clearly the case if $Q_{\tuple w,k}=0$, so assume this is not the case.
    Define $\rel X$ to be the subgraph of $(\rel K_m)^n$ induced by the vertices $\tuple z$ such that $z_p\neq w_p$ for all $p\in[n]$, namely the subgraph induced by the neighborhood of $\tuple w$.
    Note that $\rel X$ is isomorphic to $(\rel K_{m-1})^n$.
    
    Define $R\colon \rel X\qto\rel K_{m-1}$ by setting $R_{\tuple z,\ell} = Q_{\tuple w,k}Q_{\tuple z,\ell}$ over the (non-trivial) space $Q_{\tuple w,k}(H)$.
    We check that this is a quantum homomorphism.
    Every $R_{\tuple z,\ell}$ is a projector since $[Q_{\tuple w,k},Q_{\tuple z,\ell}]=0$.
    For $\tuple z\in X$, we have $\sum_\ell R_{\tuple z,\ell}=\sum_\ell Q_{\tuple w,k}Q_{\tuple z,\ell}=Q_{\tuple w,k}(\sum_\ell Q_{\tuple z,\ell})=Q_{\tuple w,k}$, which is indeed the identity function on $Q_{\tuple w,k}(H)$ since $Q_{\tuple w,k}$ is a projector.
    Thus, $R$ is a quantum function.
    If $\tuple z,\tuple z'$ are adjacent in $\rel X$, then $\tuple z,\tuple z',\tuple w$ form a triangle in $(\rel K_m)^n$
    and therefore the projectors $Q_{\tuple w,k},Q_{\tuple z,\ell},Q_{\tuple z',\ell'}$ pairwise commute for $\ell,\ell'$ distinct from $k$.
    Thus,~\eqref{quantum_homo_2} holds.
    Finally, if $\tuple z,\tuple z'$ are adjacent in $\rel X$ and $\ell\neq k$, then $R_{\tuple z,\ell}R_{\tuple z',\ell} = Q_{\tuple w,k}Q_{\tuple z,\ell}Q_{\tuple w,k}Q_{\tuple z',\ell}=Q_{\tuple w,k}Q_{\tuple z,\ell}Q_{\tuple z',\ell} = 0$.
    Thus,~\eqref{quantum_homo_1} holds.
    
    By the inductive hypothesis, we have $\qPol(\rel K_{m-1})=\qcPol(\rel K_{m-1})$. Therefore, $[R_{\tuple z,\ell},R_{\tuple z',\ell'}]=0$ for all $\tuple z,\tuple z'\in X$ and all $\ell,\ell'\in [m]\setminus\{k\}$.
    For $\tuple z=\tuple u$ and $\tuple z'=\tuple v$, this gives the claimed~\cref{eqn_1431_0110}.

    Observe now that
    
        \[[Q_{\tuple u,i},Q_{\tuple v,j}] = [Q_{\tuple u,i}(\sum_{k}Q_{\tuple w,k}), Q_{\tuple v,j}(\sum_{\ell}Q_{\tuple w,\ell})]
        = \sum_{k,\ell} [ Q_{\tuple u,i}Q_{\tuple w,k},Q_{\tuple v,j}Q_{\tuple w,\ell}].\]
        All terms with $k=\ell$ in the sum vanish by~\cref{eqn_1431_0110}. Moreover,
    all terms with $k\neq \ell$ in the sum vanish as well, since
    \begin{align*}[ Q_{\tuple u,i}Q_{\tuple w,k},Q_{\tuple v,j}Q_{\tuple w,\ell}] &= Q_{\tuple u,i}Q_{\tuple w,k}Q_{\tuple v,j}Q_{\tuple w,\ell} - Q_{\tuple v,j}Q_{\tuple w,\ell}Q_{\tuple u,i}Q_{\tuple w,k}\\
    &= Q_{\tuple u,i}Q_{\tuple w,k}Q_{\tuple w,\ell}Q_{\tuple v,j}- Q_{\tuple v,j}Q_{\tuple w,\ell}Q_{\tuple w,k}Q_{\tuple u,i} &\\
    &= 0, &
    \end{align*}
    where the second equality is due to the fact that $Q_{\tuple w,k}$ and $Q_{\tuple w,\ell}$ commute with $Q_{\tuple u,i}$ and $Q_{\tuple v,j}$, and the third holds since $Q_{\tuple w,k}Q_{\tuple w,\ell}=0=Q_{\tuple w,\ell}Q_{\tuple w,k}$. This is enough to conclude.
\end{proof}

Note that $\rel K_2$ does have quantum polymorphisms that are not in the quantum closure of $\Pol(\rel K_2)$.
Indeed, let $n=2$ and $H=\C^2$, and define $Q\colon(\mathbb K_2)^2\qto\mathbb K_2$ by the following matrices:
\begin{itemize}
    \item $Q_{(0,0),0}=\begin{pmatrix}1&0\\0&0\end{pmatrix}, Q_{(0,0),1}=\begin{pmatrix}0&0\\0&1\end{pmatrix}$
    \item $Q_{(1,1),0}=\begin{pmatrix}0&0\\0&1\end{pmatrix}, Q_{(1,1),0}=\begin{pmatrix}1&0\\0&0\end{pmatrix}$
    \item $Q_{(1,0),0}=\frac12\begin{pmatrix}1&1\\1&1\end{pmatrix}, Q_{(1,0),1}=\frac12\begin{pmatrix}1&-1\\-1&1\end{pmatrix}$
    \item $Q_{(0,1),0} = \frac12\begin{pmatrix}1&-1\\-1&1\end{pmatrix}, Q_{(0,1),1}=\frac12\begin{pmatrix}1&1\\1&1\end{pmatrix}$.
\end{itemize}
It is easily checked that this is a quantum polymorphism and that $[Q_{(0,0),0},Q_{(1,0),0}]\neq 0$.
By~\Cref{gadget-characterization}, we deduce that $\rel K_2$ does not have a commutativity gadget. In~\Cref{sec_boolean}, we will see that this fact can be alternatively derived via a more general argument, by simply noting that $\K_2$ is preserved by the classical majority polymorphism and is not the full binary relation, see~\Cref{thm_boolean_comm_gadget_classification}.

\begin{proposition}\label{Km-gadget}
    Let $m\geq 3$.
    For every $n\geq 2$ and $\tuple u,\tuple v\in(\mathbb K_m)^n$ such that $u_i=v_i$ for at least one $i\in[n]$ and $u_j\neq v_j$ for at least one $j\in[n]$,  $(\mathbb K_m)^n$ with the distinguished vertices $\tuple u,\tuple v$ is a commutativity gadget for $\mathbb K_m$.
\end{proposition}
\begin{proof}
    We apply~\Cref{obtaining-commgadget}.
    The assumptions on $\tuple u,\tuple v$ imply that $(u_1,v_1),\dots,(u_n,v_n)$ generate $[m]^2$.
    Indeed, any pair $(c,c)$ is the image of $(u_i,v_i)$ under a permutation (which is in particular a unary polymorphism of $\rel K_m$), while any pair $(c,d)$ with $c\neq d$ is the image of $(u_j,v_j)$ under a permutation.
\end{proof}

An immediate consequence of~\Cref{Km-gadget} is the undecidability of the quantum $m$-coloring problem for $m\geq 3$.
\begin{corollary}\label{cliques-undecidable}
    $\qCSP(\rel K_m)$ is undecidable for all $m\geq 3$.
\end{corollary}
\begin{proof}
    The result follows by combining~\Cref{Kn-closure} and~\Cref{cor_np_complete_plus_no_contextual_means_undecidable_ALGEBRAIC}.
\end{proof}

We point out that the case of $m=3$ was proved in~\cite{culf2024re} via the reduction of~\cite{Ji}, and the case of $m\geq 4$ was recently and independently proved in~\cite[Proposition 5.2, Lemma 4.6]{Zeman}.

\begin{figure}
\centering
\resizebox{.4\textwidth}{!}{%
\begin{circuitikz}[scale=.8]
\tikzstyle{every node}=[font=\Large]
\draw [short] (6.25,21.5) -- (12.5,19);
\draw [short] (12.5,19) -- (18.75,21.5);
\draw [short] (12.5,19) -- (6.25,14);
\draw [short] (12.5,19) -- (18.75,14);
\draw [short] (6.25,21.5) -- (18.75,21.5);
\draw [ fill={rgb,255:red,113; green,54; blue,186} ] (6.25,21.5) circle (0.25cm);
\draw [ fill={rgb,255:red,113; green,54; blue,186} ] (12.5,19) circle (0.25cm);
\draw [ fill={rgb,255:red,113; green,54; blue,186} ] (18.75,21.5) circle (0.25cm);
\draw [ fill={rgb,255:red,113; green,54; blue,186} ] (6.25,14) circle (0.25cm);
\draw [ fill={rgb,255:red,113; green,54; blue,186} ] (18.75,14) circle (0.25cm);
\draw [dashed] (9.25,24) -- (9.25,11.75);
\draw [dashed] (15.5,24) -- (15.5,11.5);
\node [font=\LARGE] at (6.25,22.75) {$Q_{x_0}$};
\node [font=\LARGE] at (12.5,22.75) {$Q_{x_1}$};
\node [font=\LARGE] at (18.75,22.75) {$Q_{x_2}$};
\node [font=\Large] at (6.25,20.75) {$a_0$};
\node [font=\Large] at (6.25,13.25) {$a_0'$};
\node [font=\Large] at (12.5,19.75) {$a_1$};
\node [font=\Large] at (18.75,13.25) {$a_2'$};
\draw [short] (6.25,14) -- (6.25,14);
\draw [short] (18.75,21.5) -- (6.25,14);
\draw [short] (6.25,21.5) -- (18.75,14);
\draw [ fill={rgb,255:red,113; green,54; blue,186} ] (6.25,14) circle (0.25cm);
\draw [ fill={rgb,255:red,113; green,54; blue,186} ] (18.75,21.5) circle (0.25cm);
\draw [ fill={rgb,255:red,113; green,54; blue,186} ] (18.75,14) circle (0.25cm);
\draw [ fill={rgb,255:red,113; green,54; blue,186} ] (6.25,21.5) circle (0.25cm);
\node [font=\Large] at (18.75,20.75) {$a_2$};
\end{circuitikz}
}%
\caption{A bifurcation of length $2$. Solid lines mean non-orthogonality.}
\label{fig_two_bifurcation}
\end{figure}

\section{Intermezzo: Contextuality and bifurcations}
\label{sec_bifurcations}
In this section, we present a general result on the emergence of a certain non-orthogonality pattern in any contextual quantum homomorphism. 
In later sections, we will use this result to analyse the quantum polymorphism minions of odd cycles and the Siggers digraph. We believe that the same ideas can be applied in a broader setting to prove non-contextuality of quantum polymorphisms (and, in particular, the existence of commutativity gadgets).

The high-level idea is that, whenever two PVMs do not commute in a quantum homomorphism $Q\colon\X\qto\A$,  a \textit{bifurcation} must appear along some walk in the product $X\times A$, in a sense that we formalise below. 
Recall that the Gaifman graph $\Gaif(\X)$ has vertex set $X$, and its edges are all pairs of vertices in $X$ simultaneously appearing in a tuple in some relation of $\X$. 
\begin{definition}
\label{defn_bifurcation}
    Let $\X,\A$ be relational structures, let $Q\colon \X\qto\A$ be a quantum homomorphism, and let $\omega$ be a positive integer. Take a $4$-tuple $\mathcal B=(\bx,\ba,a_0',a_\omega')$, where
    \begin{itemize}
        \item $\bx=(x_0,\dots,x_\omega)\in X^{\omega+1}$,
        \item $\ba=(a_0,\dots,a_\omega)\in A^{\omega+1}$, and \item $a_0',a_\omega'\in A$.
    \end{itemize}
We say that $\mathcal B$ is 
a \textit{(contextuality) bifurcation} for $Q$ if it meets the following five conditions:

\begin{enumerate}
        \item[$(i)$] $x_0,x_1,\dots,x_\omega$ is a walk in $\Gaif(\X)$;
        \item[$(ii)$] $a_0\neq a_0'$ and $a_\omega\neq a_\omega'$;
        \item[$(iii)$] 
        $Q_{x_i,a_i}Q_{x_j,a_j}\neq 0$ for each $i< j\in\{0,\dots,\omega\}$;
        \item[$(iv)$] 
        $Q_{x_0,a_0'}Q_{x_j,a_j}\neq 0$ for each $ j\in\{1,\dots,\omega\}$;
        \item[$(v)$] 
        $Q_{x_i,a_i}Q_{x_\omega,a_\omega'}\neq 0$ for each $ i\in\{0,\dots,\omega-1\}$.
    \end{enumerate}
    We call $\omega$ the \textit{length} of the bifurcation $\mathcal B$.
\end{definition}

\Cref{fig_two_bifurcation} and~\Cref{fig_omega_bifurcation} illustrate two examples of contextuality bifurcations.
The following is the main result of this section, showing that bifurcations occur in any contextual quantum homomorphism, provided that the left structure is connected.

\thmbifurcations*
We shall derive~\Cref{thm_bifurcations} as a consequence of the next, slightly stronger result. 

\begin{proposition}
\label{prop_bifurcation_detailed}
    Let $\X$ and $\A$ be relational structures  and let $Q\colon\X\qto\A$ be a quantum homomorphism. 
    Let $x_0,x_1,\dots,x_\omega$ be a walk in $\Gaif(\X)$, and choose two values $a,\tilde a\in A$ in a way that
    \begin{itemize}
        \item[$(I)$] $[Q_{x_0,a},Q_{x_\omega,\tilde a}]\neq 0$, but
        \item[$(II)$] for all $i,j\in\{0,\dots,\omega\}$ such that $\{i,j\}\neq\{0,\omega\}$,
        the PVMs $Q_{x_i}$ and $Q_{x_j}$ commute.
    \end{itemize} 
    Then $Q$ has a bifurcation $\mathcal{B}=(\bx,\ba,a_0',a_\omega')$ for some tuple $\ba=(a_0,\dots,a_\omega)\in A^{\omega+1}$ such that $a_0=a$, $a_\omega=\tilde a$ and two values $a_0',a_\omega'\in A$, where $\bx=(x_0,\dots,x_\omega)$.
\end{proposition}

In the proof of~\Cref{prop_bifurcation_detailed}, we shall repeatedly use the following fact about projectors.
\begin{fact}
\label{fact_commmuting_projectors}
    Let $P,Q$ be projectors onto some Hilbert space $H$. Then $[P,Q]=0$ if, and only if, $PQP=PQ=QP$.
\end{fact}
\begin{proof}
    If $[P,Q]=0$, then $PQ=QP=QP^2=PQP$. The other direction is clear.
\end{proof}

\def\nn{2}
\begin{figure}
\centering
\resizebox{.85\textwidth}{!}{%
\begin{circuitikz}
\tikzstyle{every node}=[font=\Large]

\draw (0,2) arc (120:60:1);
\draw (0,2) arc (240:300:2);
\draw (1,2) arc (120:60:1);
\draw (5,2) arc (120:60:1);
\draw (5,2) arc (240:300:2);
\draw (6,2) arc (120:60:1);
\draw (0,2) arc (120:60:5);
\draw (0,2) arc (120:60:6);
\draw (0,2) arc (120:60:7);
\draw (1,2) arc (120:60:4);
\draw (1,2) arc (120:60:5);
\draw (1,2) arc (120:60:6);
\draw (2,2) arc (120:60:3);
\draw (2,2) arc (120:60:4);
\draw (2,2) arc (120:60:5);

\foreach \x in {0,...,\nn}
{
\draw [short] (-1,1) -- (\x,2);
\draw [short] (-1,-1) -- (\x,2);
\draw [short] (2*\nn+4,+1) -- (\x,2);
\draw [short] (2*\nn+4,-1) -- (\x,2);
\draw [short] (-1,1) -- (\x+\nn+3,2);
\draw [short] (-1,-1) -- (\x+\nn+3,2);
\draw [short] (2*\nn+4,+1) -- (\x+\nn+3,2);
\draw [short] (2*\nn+4,-1) -- (\x+\nn+3,2);
\draw [ fill={rgb,255:red,113; green,54; blue,186} ] (\x,2) circle (0.1cm);
\draw [ fill={rgb,255:red,113; green,54; blue,186} ] (\x+\nn+3,2) circle (0.1cm);
}
\node [font=\scriptsize] at (0,2+.2) {$a_{1}$};
\node [font=\scriptsize] at (1,2+.2) {$a_{2}$};
\node [font=\scriptsize] at (2,2+.2) {$a_{3}$};
\node [font=\scriptsize] at (5,2+.2) {$a_{\omega-3}$};
\node [font=\scriptsize] at (6,2+.2) {$a_{\omega-2}$};
\node [font=\scriptsize] at (7,2+.2) {$a_{\omega-1}$};
\draw [thick] (-1,1) -- (2*\nn+4,1);
\draw [thick] (-1,1) -- (2*\nn+4,-1);
\draw [thick] (-1,-1) -- (2*\nn+4,1);

\node [font=\Large] at (\nn+1.5,2) {$\dots$};

\node [font=\scriptsize] at (-1.1,.7) {$a_0$};
\node [font=\scriptsize] at (-1.1,-.7) {$a_0'$};
\node [font=\scriptsize] at (2*\nn+4+.1,.7) {$a_\omega$};
\node [font=\scriptsize] at (2*\nn+4+.1,-.7) {$a_\omega'$};
\node [font=\scriptsize] at (-1,3.3) {$Q_{x_0}$};
\node [font=\scriptsize] at (0,3.3) {$Q_{x_1}$};
\node [font=\scriptsize] at (1,3.3) {$Q_{x_2}$};
\node [font=\scriptsize] at (2,3.3) {$Q_{x_3}$};
\node [font=\scriptsize] at (5,3.3) {$Q_{x_{\omega-3}}$};
\node [font=\scriptsize] at (6,3.3) {$Q_{x_{\omega-2}}$};
\node [font=\scriptsize] at (7,3.3) {$Q_{x_{\omega-1}}$};
\node [font=\scriptsize] at (8,3.3) {$Q_{x_\omega}$};
\draw [dashed] (-.5,3.5) -- (-.5,-1.5);
\draw [dashed] (.5,3.5) -- (.5,-1.5);
\draw [dashed] (1.5,3.5) -- (1.5,-1.5);
\draw [dashed] (2.5,3.5) -- (2.5,-1.5);
\draw [dashed] (4.5,3.5) -- (4.5,-1.5);
\draw [dashed] (5.5,3.5) -- (5.5,-1.5);
\draw [dashed] (6.5,3.5) -- (6.5,-1.5);
\draw [dashed] (7.5,3.5) -- (7.5,-1.5);

\draw [ fill={rgb,255:red,113; green,54; blue,186} ] (-1,1) circle (0.1cm);
\draw [ fill={rgb,255:red,113; green,54; blue,186} ] (-1,-1) circle (0.1cm);
\draw [ fill={rgb,255:red,113; green,54; blue,186} ] (2*\nn+4,1) circle (0.1cm);
\draw [ fill={rgb,255:red,113; green,54; blue,186} ] (2*\nn+4,-1) circle (0.1cm);

\end{circuitikz}
}%
\caption{A bifurcation of length $\omega$. Solid lines mean non-orthogonality.}
\label{fig_omega_bifurcation}
\end{figure}

\begin{proof}[Proof of~\Cref{prop_bifurcation_detailed}]
Let $Q\colon\X\qto\A$ be a quantum homomorphism over some Hilbert space $H$ of dimension $h$, and consider a walk $x_0,x_1,\dots,x_\omega$ in $\Gaif(\X)$ and values $a,\tilde a\in A$ that satisfy conditions $(I)$ and $(II)$ of the proposition.
Note that condition $(I)$ forces $\omega\geq 2$, as $Q$ satisfies~\eqref{quantum_homo_2}.
Moreover, using condition $(II)$, we can
choose an orthonormal basis $\mathcal E=(e_1,\dots,e_h)$ for $H$ that diagonalises all PVMs $Q_{x_\alpha}$ for each $\alpha\in\{0,\dots,\omega-1\}$. 

Take the set $Z=\{i\in [h]:Q_{x_\omega,{\tilde a}}e_i\in\{0,e_i\}\}$, and consider the following endomorphisms of $H$: the projector $P_1$ onto the space $\Span(\{e_i:i\in Z\})$, and the projector $P_2$ onto the space $\Span(\{e_i:i\in [h]\setminus Z\})$. Observe that $P_1P_2=0$ and $P_1+P_2=\id$.  Moreover, calling $a_0=a$ and $a_\omega=\tilde a$, we have
\begin{align}
\label{eqn_1244_0511}
    [P_1,M]=[P_2,M]=0\quad\forall M\in\{Q_{x_\alpha,d}:\alpha\in\{0,\dots,\omega-1\},d\in A\}\cup\{Q_{x_\omega,a_\omega}\}.
\end{align}
Indeed, $P_1$ is diagonal in the same basis $\mathcal E$ that diagonalises each $Q_{x_\alpha}$, so it commutes with all projectors in such PVMs. Furthermore, 
\begin{align*}
    Q_{x_\omega,a_\omega}P_1=\sum_{i\in Z}Q_{x_\omega,a_\omega}e_ie_i^*
    =\sum_{i\in Z}e_ie_i^*Q_{x_\omega,a_\omega}=
    P_1Q_{x_\omega,a_\omega}
\end{align*}
since, for each $i\in Z$, $e_i$ is an eigenvector of $Q_{x_\omega,a_\omega}$. Hence, $P_1$ also commutes with $Q_{x_\omega,a_\omega}$. Finally, $[P_1,M]+[P_2,M]=[\id,M]=0$ for each $M$ because of the bilinearity of the commutator, so~\cref{eqn_1244_0511} holds. 
In particular, the operators $Q_{x_0,a_0}P_1$, $Q_{x_0,a_0}P_2$, $Q_{x_\omega,a_\omega}P_1$, and $Q_{x_\omega,a_\omega}P_2$ are all projectors. Observe now that
\begin{align}
\label{eqn_1548_0411}
\notag
    [Q_{x_0,a_0},Q_{x_\omega,a_\omega}]
    &=
    [Q_{x_0,a_0}P_1+Q_{x_0,a_0}P_2,Q_{x_\omega,a_\omega}P_1+Q_{x_\omega,a_\omega}P_2]\\
    \notag&=
    [Q_{x_0,a_0}P_1,Q_{x_\omega,a_\omega}P_1]+[Q_{x_0,a_0}P_1,Q_{x_\omega,a_\omega}P_2]+[Q_{x_0,a_0}P_2,Q_{x_\omega,a_\omega}P_1]\\&+[Q_{x_0,a_0}P_2,Q_{x_\omega,a_\omega}P_2].
\end{align}
Combining~\cref{eqn_1244_0511} and~\Cref{fact_commmuting_projectors}, we find that the second summand in the right-hand side of~\cref{eqn_1548_0411} is
\begin{align*}
    [Q_{x_0,a_0}P_1,Q_{x_\omega,a_\omega}P_2]
    =
    [P_1Q_{x_0,a_0}P_1,P_2Q_{x_\omega,a_\omega}P_2]
    =
    0
\end{align*}
since $P_1P_2=0$. Similarly, we obtain $[Q_{x_0,a_0}P_2,Q_{x_\omega,a_\omega}P_1]=0$.
Consider now the first summand in the right-hand side of~\cref{eqn_1548_0411}. Both projectors $Q_{x_0,a_0}P_1$ and $Q_{x_\omega,a_\omega}P_1$ are diagonal in the basis $\mathcal{E}$. Hence, they commute, too. We deduce that
\begin{align}
\label{eqn_1552_0411}
    [Q_{x_0,a_0}P_2,Q_{x_\omega,a_\omega}P_2]=[Q_{x_0,a_0},Q_{x_\omega,a_\omega}]\neq 0.
\end{align}

We now claim that the matrices in the set $\{Q_{x_\alpha,d}P_2:\alpha\in[\omega-1],d\in A\}\cup\{Q_{x_\omega,a_\omega}P_2\}$ pairwise commute. The fact that $[Q_{x_\alpha,d}P_2,Q_{x_\alpha',d'}P_2]=0$ for each $\alpha,\alpha'\in[\omega-1],d,d'\in A$ is clear, since $\mathcal E$ diagonalises both. Moreover,
\begin{align*}
    Q_{x_\alpha,d}P_2Q_{x_\omega,a_\omega}P_2
    =
    Q_{x_\alpha,d}Q_{x_\omega,a_\omega}P_2
    =
    Q_{x_\omega,a_\omega}Q_{x_\alpha,d}P_2
    =
    Q_{x_\omega,a_\omega}P_2Q_{x_\alpha,d}P_2,
\end{align*}
where the first and third equalities come from~\cref{eqn_1244_0511} and~\Cref{fact_commmuting_projectors}, while the second holds since the PVMs $Q_{x_\alpha}$ and $Q_{x_\omega}$ are assumed to commute by condition $(II)$.
So, the claim is true, and we can simultaneously diagonalise the matrices in the set $\{Q_{x_\alpha,d}P_2:\alpha\in[\omega-1],d\in A\}\cup\{Q_{x_\omega,a_\omega}P_2\}$ in some orthonormal basis $\mathcal F=(f_1,\dots, f_h)$ of $H$. 

For the next step, it is useful to define the sets $L,M\subseteq [h]$ such that $Q_{x_0,a_0}P_2=\sum_{\ell\in L}e_\ell e_\ell^*$ and $Q_{x_\omega,a_\omega}P_2=\sum_{m\in M}f_mf_m^*$. For $m\in M$, consider also the set $A_m=\{i\in[h]:\ang{e_i}{f_m}\neq 0\}$. Observe that, for any $m\in M$,
\begin{align*}
    Q_{x_0,a_0}P_2f_mf_m^*
    &=
    Q_{x_0,a_0}P_2\left(\sum_{i\in [h]}e_ie_i^*\right)f_mf_m^*\left(\sum_{j\in [h]}e_je_j^*\right)
    =
    \sum_{i,j\in[h]}Q_{x_0,a_0}P_2e_ie_i^*f_mf_m^*e_je_j^*\\
    &=
    \sum_{\substack{i\in L\cap A_m\\j\in A_m}}e_ie_i^*f_mf_m^*e_je_j^*
    \intertext{and}
    f_mf_m^*Q_{x_0,a_0}P_2
    &=
    \left(\sum_{i\in [h]}e_ie_i^*\right)f_mf_m^*\left(\sum_{j\in [h]}e_je_j^*\right)Q_{x_0,a_0}P_2
    =
    \sum_{i,j\in[h]}e_ie_i^*f_mf_m^*e_je_j^*Q_{x_0,a_0}P_2\\
    &=
    \sum_{\substack{i\in A_m\\j\in L\cap A_m}}e_ie_i^*f_mf_m^*e_je_j^*.
\end{align*}
Hence, if $A_m\subseteq L$ or if $L\cap A_m=\emptyset$, it holds that $[Q_{x_0,a_0}P_2,f_mf_m^*]=0$. 
On the other hand, we know from~\cref{eqn_1552_0411} that
\begin{align*}
    0\neq 
    [Q_{x_0,a_0}P_2,Q_{x_\omega,a_\omega}P_2]
    =
    \sum_{m\in M}[Q_{x_0,a_0}P_2,f_mf_m^*],
\end{align*}
whence we deduce that there must exist some $\bar m\in M$ for which $A_{\bar m}\not\subseteq L$ and $A_{\bar m}\not\subseteq [h]\setminus L$. 

We now choose $\xi_1\in A_{\bar m}\cap L$ and $\xi_2\in A_{\bar m}\cap ([h]\setminus L)$. We claim that 
\begin{align}
\label{eqn_0411_1849}
\xi_i\in [h]\setminus Z\quad \mbox{ for both }i=1 \mbox{ and }i=2.
\end{align}
Indeed, suppose for the sake of contradiction that $\xi_i\in Z$ for some $i\in [2]$. This would imply that $P_2e_{\xi_i}=0$, whence in turn we could deduce that 
\begin{align*}
    \sum_{m\in M}f_mf_m^*e_{\xi_i}=Q_{x_\omega,a_\omega}P_2e_{\xi_i}=0.
\end{align*}
It would follow that $\ang{e_{\xi_i}}{f_m}=0$ for each $m\in M$. In particular, $\xi_i\not\in A_{\bar m}$, which is a contradiction. Hence,~\eqref{eqn_0411_1849} is true.

Let now $a_0'\in A$ be such that $Q_{x_0,a_0'} e_{\xi_2}=e_{\xi_2}$. Notice that there exists exactly one such $a_0'$, since $Q_{x_0}$ is a PVM diagonalised by $\mathcal E$. We claim that $a_0'\neq a_0$. Suppose, for the sake of contradiction, that $a_0'=a_0$. Then, it would hold that $Q_{x_0,a_0}e_{\xi_2}=e_{\xi_2}$. On the other hand, we have 
\begin{align}
0
=
\sum_{\ell\in L}e_\ell e_{\ell}^*e_{\xi_2}
=
Q_{x_0,a_0}P_2e_{\xi_2}
=
P_2Q_{x_0,a_0}e_{\xi_2}
=
P_2e_{\xi_2}
=
e_{\xi_2},
\end{align}
where the first equality is true since $\xi_2\not\in L$, the second is true by definition of $L$, the third by~\cref{eqn_1244_0511}, the fourth by the observation above, and the fifth by~\cref{eqn_0411_1849}.
This is a contradiction, 
so the claimed disequality $a_0'\neq a_0$ must hold.

Fix now some $\alpha\in[\omega-1]$, and
recall that the PVM $Q_{x_\alpha}$ is diagonalised by $\mathcal E$. Hence, there exists exactly one element $c_1\in A$ such that $Q_{x_\alpha,c_1}e_{\xi_1}=e_{\xi_1}$, and exactly one element $c_2\in A$ such that $Q_{x_\alpha,c_2}e_{\xi_2}=e_{\xi_2}$. We claim that $c_1=c_2$. Suppose this is not the case. 
It is time to remember that we chose $\mathcal{F}$ in such a way that it diagonalises both $Q_{x_\omega,a_\omega}P_2$ and all the matrices $Q_{x_\alpha,d}P_2$ for $d\in A$.  
Hence, we can write $Q_{x_\alpha,c_1}P_2=\sum_{p\in V_1}f_pf_p^*$ and $Q_{x_\alpha,c_2}P_2=\sum_{p\in V_2}f_pf_p^*$ for some $V_1,V_2\subseteq [h]$. Since $Q_{x_\alpha}$ is a PVM, we have $Q_{x_\alpha,c_1}Q_{x_\alpha,c_2}=0$. Using~\cref{eqn_1244_0511} and~\Cref{fact_commmuting_projectors}, we deduce that
\begin{align*}
    0
    =
    Q_{x_\alpha,c_1}Q_{x_\alpha,c_2}P_2
    =
    Q_{x_\alpha,c_1}P_2Q_{x_\alpha,c_2}P_2,
\end{align*}
so $V_1\cap V_2=\emptyset$. On the other hand, using~\cref{eqn_0411_1849},  we have
\begin{align*}
    e_{\xi_1}&=Q_{x_\alpha,c_1}e_{\xi_1}=Q_{x_\alpha,c_1}P_2e_{\xi_1}=\sum_{p\in V_1}f_pf_p^*e_{\xi_1}\quad\mbox{and}\\
    e_{\xi_2}&=Q_{x_\alpha,c_2}e_{\xi_2}=Q_{x_\alpha,c_2}P_2e_{\xi_2}=\sum_{p\in V_2}f_pf_p^*e_{\xi_2}.
\end{align*}
Since we are assuming that $\{\xi_1,\xi_2\}\subseteq A_{\bar m}$, we must have $f_{\bar m}^*e_{\xi_1}\neq 0$ and $f_{\bar m}^*e_{\xi_2}\neq 0$. But this means that $\bar m\in V_1\cap V_2$, a contradiction.
This proves the claimed equality $c_1=c_2$. Call such element $a_\alpha$.\\

We now make the following claim.
\begin{claim}
\label{claim_1251_1503}
    There exists some $a_\omega'\in A$ such that $a_\omega'\neq a_\omega$ and
\begin{align}
   \label{eqn_1251_2511} \left(\prod_{0\leq i\leq\omega-1}Q_{x_i,a_i}\right)Q_{x_\omega,a_\omega'}\neq 0.
\end{align}
\end{claim}
The truth of the claim would be enough to conclude the proof of the proposition. Indeed, consider the 4-tuple $(\bx,\ba,a_0',a_\omega')$, where $\bx=(x_0,\dots,x_\omega)$ and $\ba=(a_0,\dots,a_\omega)$.

Observe first that
\begin{align*}
    Q_{x_0,a_0}e_{\xi_1}
    =
    Q_{x_0,a_0}P_2e_{\xi_1}
    =
    e_{\xi_1},
\end{align*}
where the first equality is due to~\cref{eqn_0411_1849} and the second to the fact that $\xi_1\in L$ and, for each $\ell\in L$, $e_{\ell}$ is an eigenvector of $Q_{x_0,a_0}P_2$ with eigenvalue $1$. It follows that, for each $i\in[\omega-1]$,
\begin{align*}
    Q_{x_0,a_0}Q_{x_i,a_i}e_{\xi_1}
    =
    Q_{x_0,a_0}e_{\xi_1}
    =
    e_{\xi_1},
\end{align*}
which means that $Q_{x_0,a_0}Q_{x_i,a_i}\neq 0$. Similarly, 
\begin{align*}
    Q_{x_0,a_0'}Q_{x_i,a_i}e_{\xi_2}
    =
    Q_{x_0,a_0'}e_{\xi_2}
    =
    e_{\xi_2},
\end{align*}
so $Q_{x_0,a_0'}Q_{x_i,a_i}\neq 0$.
Moreover, clearly, $Q_{x_i,a_i}Q_{x_j,a_j}\neq 0$ for each $i\neq j\in[\omega-1]$ since $e_{\xi_1}$ is an eigenvector of both with eigenvalue $1$.

Notice now that
\begin{align*}
    Q_{x_\omega,a_\omega}P_2f_{\bar m}=f_{\bar m}
\end{align*}
since $\bar m\in M$ and, for each $m\in M$, $f_{m}$ is an eigenvector of $Q_{x_\omega,a_\omega}P_2$ with eigenvalue $1$. Hence, for each $i\in[\omega-1]$,
\begin{align*}
    e_{\xi_1}^*Q_{x_i,a_i}Q_{x_\omega,a_\omega}P_2f_{\bar m}
    =
    e_{\xi_1}^*f_{\bar m}\neq 0
\end{align*}
since $\xi_1\in A_{\bar m}$. In particular, this means that $Q_{x_i,a_i}Q_{x_\omega,a_\omega}\neq 0$.  Analogously, 
\begin{align*}
     e_{\xi_2}^*Q_{x_0,a_0'}Q_{x_\omega,a_\omega}P_2f_{\bar m}
    =
    e_{\xi_2}^*f_{\bar m}\neq 0
\end{align*}
since $\xi_2\in A_{\bar m}$. Hence, $Q_{x_0,a_0'}Q_{x_\omega,a_\omega}\neq 0$.
Note also that $Q_{x_0,a_0}Q_{x_\omega,a_\omega}\neq 0$, as we know that $[Q_{x_0,a_0},Q_{x_\omega,a_\omega}]\neq 0$.

All in all, the above argument proves conditions $(iii)$ and $(iv)$ of~\Cref{defn_bifurcation}, while condition $(i)$ is clear from the initial assumptions and $(ii)$ follows from the choice of $a_0'$ and from~\Cref{claim_1251_1503}. Moreover, $(v)$ also follows from~\Cref{claim_1251_1503} by noting that condition $(II)$ of the current proposition guarantees that $Q_{x_{\omega},a_\omega'}$ commutes with $Q_{x_i,a_i}$ for each $1\leq i\leq \omega-1$ and, thus, in order for~\eqref{eqn_1251_2511} to hold it must be the case that $Q_{x_i,a_i}Q_{x_\omega,a_\omega'}\neq 0$ for each $i\in\{0,\dots,\omega-1\}$.  
As a consequence, $\mathcal B=(\bx,\ba,a_0',a_\omega')$ is a bifurcation for $Q$, as required.\\

We are left to prove~\Cref{claim_1251_1503}.
Suppose, for the sake of contradiction, that the claim is false. Then, the fact that $Q_{x_\omega}$ is a PVM would imply that
\begin{align*}
    \prod_{0\leq i\leq\omega-1}Q_{x_i,a_i}
    &=
    \prod_{0\leq i\leq\omega-1}Q_{x_i,a_i}\id
    =
    \sum_{d\in A}\prod_{0\leq i\leq\omega-1}Q_{x_i,a_i}Q_{x_\omega,d}\\
    &=
    \prod_{0\leq i\leq\omega-1}Q_{x_i,a_i}Q_{x_\omega,a_\omega}.
\end{align*}
But then, we would have
\begin{align*}
    \prod_{0\leq i\leq\omega-1}Q_{x_i,a_i}Q_{x_\omega,a_\omega}e_{\xi_1}
    =
    \prod_{0\leq i\leq\omega-1}Q_{x_i,a_i}e_{\xi_1}
    =
    e_{\xi_1}
\end{align*}
so, in particular, \[\|\prod_{0\leq i\leq\omega-1}Q_{x_i,a_i}Q_{x_\omega,a_\omega}e_{\xi_1}\|=1.\] However, we know from~\cref{eqn_0411_1849} that $\xi_1\not\in Z$, which by the definition of $Z$ implies that $Q_{x_\omega,a_\omega}e_{\xi_1}\neq e_{\xi_1}$. Since $Q_{x_\omega,a_\omega}$ is a projector, this means that $\|Q_{x_\omega,a_\omega}e_{\xi_1}\|<1$. Since \[\prod_{0\leq i\leq\omega-1}Q_{x_i,a_i}\] is also a projector, it follows that 
\begin{align*}
    \|\prod_{i\in[\omega-1]}Q_{x_i,a_i}Q_{x_\omega,a_\omega}e_{\xi_1}\|<1,
\end{align*}
a contradiction. Hence,~\Cref{claim_1251_1503} is true, and the proof of~\Cref{prop_bifurcation_detailed} is complete.
\end{proof}

Having established~\Cref{prop_bifurcation_detailed},
we can now prove the main result of this section.
\begin{proof}[Proof of~\Cref{thm_bifurcations}]
    Let $Q\colon\X\qto\A$ be a contextual quantum homomorphism. Consider the set $S\subseteq X^2$ containing all pairs $(x,\tilde x)$ such that the PVMs $Q_{x}$ and $Q_{\tilde x}$ do not commute. Since $Q$ is contextual, $S$ is non-empty. Choose $(x,\tilde x)\in S$ in a way that $\dist_{\Gaif(\X)}(x,\tilde x)$ is minimum over all pairs in $S$; call such distance $\omega$. (Observe that this is well defined since $\Gaif(\X)$ is assumed to be connected.)
    Take $a,\tilde a\in A$ so that $[Q_{x,a},Q_{\tilde x,\tilde a}]\neq 0$.
    Let $x=x_0,x_1,\dots,x_\omega=\tilde x$ be a shortest path in $\Gaif(\X)$ connecting $x$ and $\tilde x$. By the minimality of $\dist_{\Gaif(\X)}(x,\tilde x)$, the PVMs $Q_{x_i}$ and $Q_{x_j}$ commute for each $i,j\in\{0,\dots,\omega\}$ such that
    $\{i,j\}\neq\{0,\omega\}$. Hence, conditions $(I)$ and $(II)$ of~\Cref{prop_bifurcation_detailed} are both met. It follows that $Q$ has a bifurcation of length $\omega=\dist_{\Gaif(\X)}(x,\tilde x)\leq\operatorname{diam}(\Gaif(\X))$. 
\end{proof}

\begin{remark}
    Observe that the proofs of both~\Cref{prop_bifurcation_detailed} and~\Cref{thm_bifurcations} do not make use of the fact that the given quantum homomorphism $Q\colon\X\qto\A$ satisfies condition~\eqref{quantum_homo_1}. Hence, both statements hold for any quantum function $Q\colon X\qto A$ satisfying~\eqref{quantum_homo_2}.  
\end{remark}

As a simple application of~\Cref{thm_bifurcations}, we can prove the following non-contextuality result.

\begin{proposition}
\label{cor_non_contextual_quantum_homos_cycles}
    Let $\X,\Y$ be relational structures such that $\Gaif(\X)$ and $\Gaif(\Y)$ are odd cycles of the same length. Then every quantum homomorphism $\X\qto\Y$ is non-contextual.
\end{proposition}
In particular, it follows that quantum endomorphisms of odd cycles are non-contextual.
\begin{corollary}
\label{cor_quantum_endo_odd_cycles_ess_classical}
For each odd $m\geq 3$, every quantum homomorphism $Q\colon\C_m\qto\C_m$ is non-contextual.
\end{corollary}
We point out that a different proof of~\Cref{cor_quantum_endo_odd_cycles_ess_classical} was obtained in~\cite{Zeman} via a result from~\cite{banica2007quantum} on the quantum symmetry of vertex-transitive graphs.
\begin{proof}[Proof of~\Cref{cor_non_contextual_quantum_homos_cycles}]
Suppose  first that $\X$ is not essentially binary---i.e., recalling~\Cref{defn_essentially_binary}, $\X$ contains some relation $R^\X$ such that at least one tuple $\bx\in R^\X$ contains at least three distinct entries---call them $x,y,z$. Then the set $\{x,y,z\}\subseteq X$ induces a $\K_3$ in $\Gaif(\X)$, and the assumption on $\Gaif(\X)$ implies that $X$ cannot contain any other vertex. Hence, in this case, the fact that each quantum homomorphism from $\X$ to $\Y$ is non-contextual immediately follows from~\eqref{quantum_homo_2}.

Suppose now that $\X$ is essentially binary and assume, for the sake of contradiction, that $\X,\Y$ admit some contextual quantum homomorphism $Q\colon\X\qto\Y$. Let $m$  be the number of vertices in $\Gaif(\X)=\Gaif(\Y)$.
Applying~\Cref{thm_bifurcations}, we deduce that 
$Q$ has a bifurcation $\mathcal{B}=(\bx,\ba,a_0',a_\omega')$ of length $\omega\leq \frac{m-1}{2}$. Let as usual $\bx=(x_0,\dots,x_\omega)$ and $\ba=(a_0,\dots,a_\omega)$, and consider the pairs $(x_0,a_0),(x_0,a_0'),(x_1,a_1),(x_2,a_2)\in X\times Y$. By~\Cref{defn_bifurcation}, we know that $Q_{x_0,a_0}Q_{x_1,a_1}\neq 0$, $Q_{x_0,a'_0}Q_{x_1,a_1}\neq 0$, $Q_{x_1,a_1}Q_{x_2,a_2}\neq 0$, $Q_{x_0,a_0}Q_{x_2,a_2}\neq 0$, and $Q_{x_0,a_0'}Q_{x_2,a_2}\neq 0$. 
Using~\Cref{prop_basic_walks_essentially_binary}, we deduce that $\{a_0,a_1\}, \{a_0',a_1\}$, and $\{a_1,a_2\}$ are all edges of $\Gaif(\Y)$. 
Note also that $x_0$ and $x_2$ lie on an $(m-2)$-walk in $\Gaif(\X)$. Since no walk of such length connects a vertex to itself in $\Gaif(\Y)$,~\Cref{prop_basic_walks_essentially_binary} ensures that $a_0\neq a_2$ and $a_0'\neq a_2$. Finally $a_0\neq a_0'$ again by~\Cref{defn_bifurcation}. We deduce that $a_1$ has degree at least three in $\Gaif(\Y)$, a contradiction.
\end{proof}

\section{Quantum polymorphisms of the Siggers digraph}
\label{sec_siggers}
In this section, we investigate the complexity of the quantum CSP parameterised by the Siggers digraph---i.e., the digraph obtained by taking a directed triangle and making one of its directed edges undirected, see~\Cref{fig_siggers_overview}. The significance of this directed graph is that it is the smallest digraph whose classical CSP is NP-complete.
\begin{figure}[h!]
\centering
\begin{tikzpicture}[
    every node/.style={circle, draw, minimum size=8mm}
]

\node (A) at (0,0) {0};
\node (B) at (3,0) {2};
\node (C) at (1.5,2.6) {1};

\draw[->] (A) to[bend left=15] (B);
\draw[->] (B) to[bend left=15] (A);

\draw[->] (C) -- (B);
\draw[->] (A) -- (C);

\end{tikzpicture}
\caption{The Siggers digraph.}
\label{fig_siggers_overview}
\end{figure}
We prove the following result.

\thmsiggersundecidable*

Throughout this section, we denote the Siggers digraph by $\rel A$ (and, as usual, we denote its vertex set $\{0,1,2\}$ by $A$). In order to deduce the undecidability of $\qCSP(\A)$ from the algebraic machinery developed in the previous sections, we show that all quantum polymorphisms of the Siggers digraph are non-contextual by making use of contextuality bifurcations.

\begin{theorem}\label{thm_noncontectual_A}
    Every quantum polymorphism $Q\colon\rel A^n\qto \rel A$ is non-contextual.
\end{theorem}

The key to establishing this result is the analysis of  the properties of quantum homomorphisms from the path of length 5, which we shall denote $\rel P_5$, to $\rel A$. Throughout this section, we label the vertices of $\rel P_5$ as follows.
\[
\begin{tikzpicture}[
    every node/.style={circle, draw, minimum size=8mm}
]

\node (A) at (0,0) {u};
\node (B) at (2,0) {a};
\node (C) at (4,0) {b};
\node (D) at (6,0) {c};
\node (E) at (8,0) {d};
\node (F) at (10,0) {v};

\draw[->] (A) -- (B);
\draw[->] (B) -- (C);
\draw[->] (C) -- (D);
\draw[->] (D) -- (E);
\draw[->] (E) -- (F);

\end{tikzpicture}
\]
We shall make use of the concept of ``modular shift'' of a PVM, as defined next.
\begin{definition}\label{defn_cyc_shift_A}
Let $Q=(Q_{i})_{i\in S}$ and $Q'=(Q'_{i})_{i\in S}$ be two PVMs whose outcome set is labelled as $S=\{0,1,\dots,n-1\}$. We say that $Q'$ is a \textit{$k$-modular shift} of $Q$ (for some $k\in S$) if 
\begin{align*}
    Q'_{i}=Q_{i+_n k}
\end{align*}
for each $i\in S$, where ``$+_n$'' denotes the addition mod $n$. Note that two PVMs $Q$ and $T$ commute if, and only if, any modular shifts $Q'$ and $T'$ also commute.
\end{definition}
For the Siggers digraph $\A$, the modular shifts of a PVM with outcome set $A$ are constructed by permuting the projectors associated with the vertices along the directed triangle underlying $\A$. 
The main thrust of our argument is captured by the following theorem.
\begin{theorem}\label{thm_A_cyc_shift_extension}
    Let $Q\colon \rel P_5\qto \rel A$ be a quantum homomorphism, and suppose that every modular shift of the PVMs associated with $u$ and $v$
    can be extended to a quantum homomorphism from $\rel P_5$ to $\rel A$. Then, $[Q_{u,i},Q_{v,j}]=0$ for each $i,j\in A$.
\end{theorem}

To prove~\Cref{thm_A_cyc_shift_extension}, it shall be useful to collect a few simple observations about contextuality bifurcations in quantum homomorphisms from $\rel P_5$ to $\rel A$. Henceforth, given a bifurcation
$\mathcal{B}=(\bx,\by, p,q)$ of length $\omega$, we shall denote $\bx=(x_0,\dots,x_\omega)$ and $\by=(y_0,\dots,y_\omega)$.
Moreover, given a tuple $\ba=(a_0,a_1,\dots,a_n)$, we let $\ba^R=(a_n,a_{n-1},\dots,a_0)$ be its reversed tuple.

\begin{lemma}
\label{lem_trivial_bifurcation_siggers_A}
    Let $Q\colon\rel P_5\qto\A$ be a quantum homomorphism, and let $\mathcal{B}=(\bx,\by, p,q)$ be a bifurcation of length $\omega$ for $Q$. Then
    there exists a bifurcation $\mathcal{B}'=(\bx',\by', p,q)$ of length $\omega'$ such that 
    \begin{itemize}
        \item $\bx'$ is either a directed path or a reversed path in $\rel P_5$;
        \item $\bx'$ is a subtuple of $\bx$;
        \item $x'_0=x_0$;
        \item $x'_{\omega'}=x_\omega$.
    \end{itemize}
\end{lemma}
\begin{proof}

    We first consider an arbitrary bifurcation $\mathcal{B}$ in a quantum homomorphism $Q\colon \rel P_5\qto \rel A$ and show that $\bx$ or $\bx^R$ is a path is $\rel P_5$ when $\bx$ contains no repeated vertices. 
    To see why, note that $\Gaif(\rel P_5)$ is the undirected path of length 5, so any walk $\bx$ that does not repeat vertices is a directed subpath in $\Gaif(\rel P_5)$. If we then consider $\bx$, we can see that either $(x_i,x_{i+1})$ is an edge of $\rel P_5$ for each $0\le i\le \omega-1$, or $(x_{i+1}, x_i)$ is an edge for each such $i$. It is thus the case that either $\bx$ or $\bx^R$ is a directed walk in $\rel P_5$.

    If $\mathcal{B}$ repeats some vertices, there exist $i<j$ (with $i\neq 0$ or $j\neq \omega$) such that $x_i=x_j$ but $x_k\neq x_\ell$ for each $k\in \{0,\dots, i-1\}$ and $\ell\in \{j+1,\dots, \omega\}$. We then define $\mathcal{B}'=(\bx',\by', p,q)$ by setting
    \begin{itemize}
        \item $\bx'=(x_0,\dots,x_{i-1},x_i,x_{j+1},\dots,x_\omega)$;
        \item $\by'=(y_0,\dots,y_{i-1},y_i,y_{j+1},\dots,y_\omega)$.
    \end{itemize}
    It is immediate that $\mathcal{B}'$ inherits properties $(ii)-(v)$ of~\Cref{defn_bifurcation} from $\mathcal{B}$. Furthermore, $\bx'$ is a walk in $\Gaif(\rel P_5)$ since the same holds for $\bx$.
    The proof is then complete since $\bx'$ does not repeat any vertices, so either $\bx'$ or ${(\bx^R)}'$ is a directed walk in $\rel P_5$.
\end{proof}

\begin{lemma}
\label{lem_trivial_bifurcation_siggers_B}
    Let $Q\colon\rel \X\qto\A$ be a quantum homomorphism, and let $\mathcal{B}=(\bx,\by, p,q)$ be a bifurcation of length $\omega$ for $Q$. Then $\mathcal{B}'=(\bx^R,\by^R, q,p)$ is also a bifurcation for $Q$. 
\end{lemma}
\begin{proof}
    We show that $\mathcal{B}'$ satisfies the conditions in~\Cref{defn_bifurcation}. For $(i)$, simply recall that $\Gaif(\X)$ is an undirected graph; so, if $\bx$ is a walk in $\Gaif(\X)$, then the same holds for $\bx^R$. $(ii)$ is immediate from $\mathcal{B}$ being a bifurcation. $(iii)$, $(iv)$, and $(v)$ are also clear from the fact that, for $A,B$ self-adjoint, $AB=0$ if, and only if, $BA=0$.
\end{proof}
Consider the following property:
\begin{itemize}
    \item[$(\mathbf{BP})$] \labeltext{$(\mathbf{BP})$}{bifurcation_P5}Bifurcation $\mathcal{B}=(\bx,\by,q,p)$ is such that $\bx$ is a walk in $\rel P_5$
\end{itemize}

Combining~\Cref{lem_trivial_bifurcation_siggers_A} and~\Cref{lem_trivial_bifurcation_siggers_B}, we deduce that we can refine any bifurcation $\mathcal{B}$ for some quantum homomorphism $Q\colon\rel P_5\qto\A$ into a (possibly different) bifurcation $\mathcal{B}'$ for $Q$ that runs between the same $x_0$ and $x_\omega$, but has property~\ref{bifurcation_P5}. Conversely, if there is no \ref{bifurcation_P5} bifurcation on a given $\rel P_5$ path, then there is no bifurcation between its start and endpoint.

\begin{lemma}
\label{lem_trivial_bifurcation_siggers_C}
    Let $Q\colon\rel P_5\qto\A$ be a quantum homomorphism, and let $\mathcal{B}=(\bx,\by, p,q)$ be a bifurcation of length $\omega$ for $Q$ with property~\ref{bifurcation_P5}. Then 
    \begin{itemize}
        \item $y_1=2$ and $y_{\omega-1}=0$;
        \item both $y_0$ and $p$ are in $\{0,1\}$;
        \item both $y_\omega$ and $q$ are in $\{1,2\}$.
    \end{itemize}
    In particular, $\omega\geq 3$. 
\end{lemma}

\begin{proof}
    Consider $Q$ and $\mathcal{B}$ as described. By the definition of a bifurcation, $Q_{x_0,y_0}Q_{x_1,y_1}\neq 0$, $Q_{x_0,p}Q_{x_1,y_1}\neq 0$, and $p\neq y_0$. Hence, using~\eqref{quantum_homo_1}, $(y_0,y_1)$ and $(p,y_1)$ are edges in $\rel A$ since $(x_0,x_1)$ is an edge in $\rel P_5$ by \ref{bifurcation_P5}.
    It follows that $y_1=2$, since 2 is the only vertex in $\rel A$ with in-degree $\ge 2$. This also implies that $y_0,p\in \{0,1\}$ by the edge relation of $\rel A$. 
    Repeating the same argument for $x_{\omega-1}$ (this time considering the out-degree), we obtain that $y_{\omega-1}=0$ and $y_\omega,q\in\{1,2\}$.
    Finally, the fact that $y_1\neq y_{\omega-1}$ implies that $\omega\geq 3$.
\end{proof}

The following will also be useful.

\begin{lemma}
\label{lem_trivial_bifurcation_siggers_D}
    For any quantum homomorphism $Q\colon \rel P_5\qto \rel A$ and any two vertices $x,y\in P_5$ such that $\dist_{\Gaif{(\rel P_5)}}(x,y)\leq 2$, the PVMs $Q_x$ and $Q_y$ commute.
\end{lemma}
\begin{proof}
    Let $d=\dist_{\Gaif{(\rel P_5)}}(x,y)$. If $d\leq 1$, the fact is clear from~\eqref{quantum_homo_2}. Let then $d=2$, and suppose for the sake of contradiction that $Q_x$ and $Q_y$ do not commute. It would follow from~\Cref{prop_bifurcation_detailed} that $Q$ admits a bifurcation running on a path $x,\dots, y$ in $\Gaif(\rel P_5)$ which induces a \ref{bifurcation_P5} bifurcation of length $2$ (or $0$), thus contradicting~\Cref{lem_trivial_bifurcation_siggers_C}.
\end{proof}

The next step towards proving~\Cref{thm_A_cyc_shift_extension} is the following result.
\begin{lemma}\label{prop:P5_QB2_commutes}
    For any quantum homomorphism $Q\colon \rel P_5\qto \rel A$, the projector $Q_{b,2}$ commutes with every projector of $Q$.
\end{lemma}
\begin{proof}
Choose $\alpha\in P_5$ and $i\in A$, and 
assume, for the sake of contradiction, that $[Q_{b,2},Q_{\alpha,i}]\neq 0$. By~\Cref{lem_trivial_bifurcation_siggers_D}, it must be that $\alpha=v$. Observe that the walk $b,c,d,v$ in $\rel P_5$ and the values $2,i\in A$ satisfy the conditions of~\Cref{prop_bifurcation_detailed}. Therefore, $Q$ admits a \ref{bifurcation_P5} bifurcation $\mathcal B=(\bx,\by,p,q)$ such that $y_0=2$. By~\Cref{lem_trivial_bifurcation_siggers_C}, this is impossible.
\end{proof}

We can thus consider any quantum homomorphism $Q\colon \rel P_5\qto \rel A$, and split it into the direct sum $S\oplus T$, where, for each $x\in P_5$ and each $i\in A$, 
\begin{align*}
    S_{x,i}&= Q_{b,2}Q_{x,i}\quad\mbox{and}\\
    T_{x,i}&=(Q_{b,0}+Q_{b,1})Q_{x,i}=(\id-Q_{b,2})Q_{x,i}.
\end{align*}
(Note that, formally, the projector, say,\ $S_{x,i}= Q_{b,2}Q_{x,i}$ acts on the subspace given by the range of $Q_{b,2}$. We will leave this implicit to avoid adding extra notation for injections and projections onto such space.)

\begin{claim}\label{SiggersC1}
    $S_{u,1}=0$.
\end{claim}
\begin{proof}
Using~\Cref{prop:P5_QB2_commutes}, we can expand $S_{u,1}$ as 
    \[S_{u,1}=
    Q_{b,2}Q_{u,1}
    =
    Q_{u,1}Q_{b,2}=\sum_{i\in A}Q_{u,1}Q_{a,i}Q_{b,2}=0+0+0=0,\]
    where the penultimate equality follows from applying~\eqref{quantum_homo_1} to the products $Q_{u,1}Q_{a,0}$, $Q_{u,1}Q_{a,1}$, and $Q_{a,2}Q_{b,2}$.
\end{proof}
\begin{claim}\label{SiggersC2}
$T_{v,0}T_{u,2}=0$.
\end{claim}

\begin{proof} 
Using~\Cref{prop:P5_QB2_commutes}, we can expand the product as
    \begin{align*}
    T_{v,0}T_{u,2}
    &=
    (\id-Q_{b,2})Q_{v,0}(\id-Q_{b,2})Q_{u,2}
    =
    Q_{v,0}(\id-Q_{b,2})^2Q_{u,2}
    =
    Q_{v,0}(\id-Q_{b,2})Q_{u,2}\\
    &=Q_{v,0}(Q_{b,0}+Q_{b,1})Q_{u,2}=Q_{v,0}Q_{b,0}Q_{u,2}+Q_{v,0}Q_{b,1}Q_{u,2}.
    \end{align*}
    We now argue that both terms of the above sum are zero. First,
    \[Q_{b,0}Q_{u,2}=\sum_{i\in A}Q_{b,0}Q_{u,2}Q_{a,i}=Q_{b,0}Q_{u,2}Q_{a,0}=Q_{b,0}Q_{a,0}Q_{u,2}=Q_{a,0}Q_{b,0}Q_{u,2}=0,\]
    where the third and fourth equalities follow from~\eqref{quantum_homo_2};
    therefore, $Q_{v,0}Q_{b,0}Q_{u,2}=0$. Second, using again~\eqref{quantum_homo_2},
    \begin{align*}Q_{v,0}Q_{b,1}&=\sum_{i,j\in A}Q_{d,i}Q_{v,0}Q_{b,1}Q_{c,j}=Q_{d,2}Q_{v,0}Q_{b,1}Q_{c,2}=Q_{v,0}Q_{d,2}Q_{c,2}Q_{b,1}\\&=Q_{v,0}Q_{c,2}Q_{d,2}Q_{b,1}=0,
    \end{align*}
    so $Q_{v,0}Q_{b,1}Q_{u,2}=0$. 
    It follows that $T_{v,0}T_{u,2}=0$, as required.
\end{proof}
The main conclusion that we draw from these observations is that both $S$ and $T$ satisfy a specific orthogonality property: Letting $W\in\{S,T\}$, it holds that $W_{v,0}W_{u,i}=0$ for some $i\in A$. In particular, there exists some modular shift $W'_u$ of $W_u$ such that $W_{v,0}W'_{u,0}=0$.

We now proceed to prove the following.

\begin{proposition}\label{prop:P5_v0u0=0} 
 Let $Q\colon \rel P_5\qto \rel A$ be a quantum homomorphism, and assume that $Q_{u,0}Q_{v,0}=0$. Then, $[Q_{\alpha,i},Q_{v,0}]=0$ for each $\alpha\in P_5$ and each $i\in A$.
\end{proposition}
\begin{proof}
First of all, using~\eqref{quantum_homo_1} we note that
\[\label{eqn:Q_x2->Q_y0}Q_{d,2}=Q_{d,2}\sum_{i\in A}Q_{v,i}=Q_{d,2}Q_{v,0}=\sum_{i\in A}Q_{d,i}Q_{v,0}=Q_{v,0}.\]
Hence, it is sufficient to prove that $[Q_{\alpha,i},Q_{d,2}]=0$.  By~\Cref{lem_trivial_bifurcation_siggers_D}, $[Q_\alpha, Q_{d,2}]=0$ for all $\alpha\in\{b,c,d,v\}$. Therefore, we only need to consider the cases of $\alpha=u$ and $\alpha=a$.
Observe that
\begin{align*}
    Q_{a,1}Q_{d,2}&=Q_{a,1}Q_{v,0}=\sum_{i\in A}Q_{u,i}Q_{a,1}Q_{v,0}
    \\&=Q_{u,0}Q_{a,1}Q_{v,0} =Q_{a,1}Q_{u,0}Q_{v,0}=0,
\end{align*}
where the last equality is true by the assumption of this proposition. Again, from~\eqref{quantum_homo_1}, we find
$Q_{a,2}=Q_{b,0}$,
which commutes with $Q_{d,2}$ by~\Cref{lem_trivial_bifurcation_siggers_D}. Thus, the PVM $Q_a$ commutes with $Q_{d,2}$, since $Q_{a,0}=\id-Q_{a,1}-Q_{a,2}$.

The only remaining case is when
$\alpha$ equals $u$.
$Q_u$ commutes with $Q_{d,2}$ since $Q_{u,0}Q_{d,2}=Q_{u,0}Q_{v,0}=0$ and $Q_{u,2}=Q_{a,0}$, which commutes with $Q_{v,0}$. The proof is then concluded since $Q_{u,1}=\id-Q_{u,0}-Q_{u,2}$ and so, $[Q_u,Q_{v,0}]=0$.\qedhere 
\end{proof}

Using~\Cref{prop:P5_v0u0=0}, we can now prove~\Cref{thm_A_cyc_shift_extension}, conditioned on $Q_{v,0}Q_{u,0}=0$.

\begin{proposition}\label{prop:P5_contextuality_disallowance} 
Let $Q\colon \rel P_5\qto \rel A$ be a quantum homomorphism such that $Q_{v,0}Q_{u,0}=0$, and suppose that every modular shift of the PVMs associated with $u$ and $v$ can be extended to a quantum homomorphism from $\rel P_5$ to $\rel A$. Then, $[Q_{u,i},Q_{v,j}]=0$ for each $i,j\in A$.
\end{proposition}

\begin{proof}
    Since $Q_{v,0}$ commutes with every projector in $Q$ (from~\Cref{prop:P5_v0u0=0}), we can decompose $Q$ as $S\oplus T$ as above, by defining for each $x\in P_5$ and each $i\in A$ 
    \begin{align*}
        S_{x,i}&= Q_{v,0}Q_{x,i}\quad\mbox{and}\\
        T_{x,i}&=(\id-Q_{v,0})Q_{x,i}.
    \end{align*} 
    The argument then follows by showing that, in each component, the $u$- and $v$-indexed PVMs commute when their modular shifts can be extended as stated in the proposition. 

    Firstly, $S_{u,i}$ always commutes with $S_{v,i}$ (since $S_{v,0}=\id$), so the implication vacuously holds. Now we focus on $T$. Assign $T_{u,i}=U_i$ and $T_{v,i}=V_i$, and assume that every assignment $T'\colon \{u,v\}\qto A$ defined by the modular shifts $T'_{u,i}=U_{i+_3k}$, $T'_{v,i}=V_{i+_3 l}$ can be completed to a quantum homomorphism $T'\colon \rel P_5\qto \rel A$ for every $k,l\in A$.

    Remember that $V_0=(\id-Q_{v,0})Q_{v,0}=0$. To show that $U_k$ must commute with $V_1$ and $V_2$, consider the modular shifts defined by $T'_{u,i}=U_{i+_3 (2-k)}$ and $T'_{v,i}=V_{i+_3 1}$. By our assumption, $T'_u$ and $T'_v$ can be extended to a quantum homomorphism $T'\colon \rel P_5\qto \rel A$. We will show that there can be no bifurcation between $U_k$ and any $v$-indexed projector in this homomorphism. 
    
    Firstly, note that $U_k=T'_{u,2}$, while 
    $T'_{u,2}=T'_{a,0}$
using~\eqref{quantum_homo_1}. Hence, we consider the path $a,b,c,d,v$ and show $T'_{a,0}$ commutes with $T'_v$. By~\Cref{lem_trivial_bifurcation_siggers_D}, the PVMs $T'_b$, $T'_c$, and $T'_d$ commute with each other since $\dist_{\Gaif{(\rel P_5)}}(x,y)\leq 2$ for all $x,y\in\{b,c,d\}$.
    
    Moreover, $T'_a$, $T'_b$, and $T'_c$ commute by~\Cref{lem_trivial_bifurcation_siggers_D}, and 
    \[T'_{d,1}=T'_{d,1}\sum_{i\in A}T'_{v,i}=0+0+T'_{d,1}T'_{v,2}=0\]
    since $T'_{v,2}=V_0=0$. 
By~\Cref{lem_trivial_bifurcation_siggers_C} and~\Cref{prop_bifurcation_detailed}, $T'_{a,y_0}T'_{d,1}\neq 0$ for any bifurcation along $a,b,c,d$. We conclude that there is no such bifurcation by the fact that $T'_{d,1}=0$. Thus, there is no \ref{bifurcation_P5} bifurcation between $T'_a$ and $T'_d$ and, so, they commute by~\Cref{prop_bifurcation_detailed}.
Similarly, $T'_c$, $T'_d$, and $T'_v$ commute by~\Cref{lem_trivial_bifurcation_siggers_D}, and $T'_{v,2}=V_0=0$. Thus, $T'_b$ and $T'_v$ commute by the same argument.

    In summary, we have established that $T'_i$ and $T'_j$ commute    for all $i,j\in\{a,b,c,d,v\}$ with $\{i,j\}\neq \{a,v\}$. \Cref{lem_trivial_bifurcation_siggers_C} shows that there cannot be a bifurcation using $a,b,c,d,v$ since $T'_{v,2}=0$. Hence, by~\Cref{prop_bifurcation_detailed}, there is no \ref{bifurcation_P5} bifurcation using the walk $u,a,b,c,d,v$ and thus $[T'_{u,2},T'_{v,1}]=0$. 
    Observe now that $T'_{v,2}=0$ implies $T'_{v,0}=\id-T'_{v,1}$. Therefore, $T'_{u,2}$ commutes with the entire PVM $T'_v$, and the proof is complete. 
\end{proof}

We can now prove the full statement of~\Cref{thm_A_cyc_shift_extension}.

\begin{proof}[Proof of~\Cref{thm_A_cyc_shift_extension}]
    Consider any quantum homomorphism $Q\colon \rel P_5 \qto \rel A$ such that the modular shifts of the $u$ and $v$-indexed PVMs extend.
    Using~\Cref{prop:P5_QB2_commutes}, we can decompose $Q= S\oplus T$ and by this construction, the modular shifts of the $u$ and $v$ PVMs in $S$ and $T$ can also be extended to quantum homomorphisms.  
    
    Letting $V\in \{S,T\}$, we know by~\Cref{SiggersC1} and~\Cref{SiggersC2} that there is some $W\colon \rel P_5\qto \rel A$ where
    \begin{enumerate}
        \item $W$ is an extension of some modular shifts of $V_u$ and $V_v$ (so, $W_u$ and $W_v$ are modular shifts of $V_u$ and $V_v$) and
        \item $W_{v,0}W_{u,0}=0$.
    \end{enumerate}

    Any modular shifts $W'_u$ and $W'_v$ of $W_u$ and $W_v$ are, in particular, modular shifts of $V_u$ and $V_v$. Thus, they can be extended to a quantum homomorphism $W'\colon \rel P_5\qto \A$ and we can apply~\Cref{prop:P5_contextuality_disallowance} to $W$.  We conclude that $W_u$ and $W_v$ commute. Thus,
    $V_u$ and $V_v$ commute for $V\in\{S,T\}$, and it follows that $Q_u$ and $Q_v$ commute.
\end{proof}

Now consider any quantum polymorphism $Q\colon \rel A^n\qto \rel A$, and fix two PVMs $Q_u$ and $Q_v$, associated with arbitrary elements $u,v\in A^n$. %
We show that for any modular shifts $Q'_u$ and $Q'_v$, there exists an appropriate restriction $\overline{Q}\colon \rel P_5\qto \rel A$ of $Q$ that extends $Q'_u$ and $Q'_v$.

\begin{proposition}\label{prop_directed_triangle_to_A}
    For any quantum polymorphism $Q\colon\rel A^n\qto \rel A$ and $x\in A^n$, there exist $y$, $z$ in $A^n$ such that $Q_y$ and $Q_z$ are modular shifts of $Q_x$.
\end{proposition}
\begin{proof}
    Take some $x\in A^n$ and let $y=x+_3\{1\}^n$ and  $z=x+_3\{2\}^n$, 
    where $+_3$ is the componentwise addition modulo 3 and $\{i\}^n$ is the $n$-tuple that with $i$ in each entry. Then the directed triangle $\rel T_x=(\{x,y,z\},\{(x,y),(y,z),(z,x)\})$ is a subgraph of $\rel A$, so every $Q\colon \rel A^n\qto \rel A$ restricts to $\overline{Q}\colon \rel T_x\qto \rel A$. 
    
    We now show that, in any quantum homomorphism $S\colon \rel T_x\qto \rel A$, $S_y$ and $S_z$ are the non-trivial modular shifts of $S_x$. To see this, note that
\begin{align*}
    S_{x,1}&=\sum_{i\in A}S_{z,i}S_{x,1}\sum_{i\in A}S_{y,i}=S_{z,0}S_{x,1}S_{y,2},\\
    S_{y,2}&=S_{y,2}\sum_{i\in A}S_{z,i}=S_{y,2}S_{z,0}=\sum_{i\in A}S_{y,i}S_{z,0}=Q_{z,0},
    \intertext{and}
    S_{y,2}S_{z,0}&=S_{z,0}S_{y,2}=S_{z,0}\sum_{i\in A}S_{x,i}S_{y,2}=S_{z,0}S_{x,1}S_{y,2},
    \end{align*}
    from which it follows that $S_{x,1}=S_{y,2}=S_{z,0}$. 
    To conclude, we simply repeat the argument for $S_{y,1}$ and $S_{z,1}$ and notice that $S_x$, cannot have more than 2 distinct, non-trivial, modular shifts (since $|A|=3$).
\end{proof}
\begin{proof}[Proof of~\Cref{thm_noncontectual_A}]
    We begin by noting that one can find a path of length 5 between any two vertices in $\rel A$. Thus, one can find a path of length 5 between any two vertices in $\rel A^n$.

    Fix any $u,v\in A^n$ and $Q\colon \rel A^n\qto \rel A$; we consider the PVMs $Q_u$ and $Q_v$. We can find $a,b,c,d\in A^n$ such that $u,a,b,c,d,v$ is a walk in $\rel A^n$ thus, we can restrict $Q$ to $\overline{Q}\colon \rel P_5 \qto \rel A$ and obtain a quantum homomorphism such that $\overline{Q}_u=Q_u$ and $\overline{Q}_v=Q_v$.

    Now consider any modular shifts $Q'_u$ and $Q'_v$ of $Q_u$ and $Q_v$, respectively. By~\Cref{prop_directed_triangle_to_A}, there exist $u',v'\in A^n$ such that $Q'_u=Q_{u'}$ and $Q'_v=Q_{v'}$. Hence, we can, again, restrict $Q$ to find $\overline{Q'}\colon \rel P_5 \qto \rel A$ such that $Q_{u'}=\overline{Q'}_u$ and $Q_{v'}=\overline{Q'}_v$.
    
    An application of~\Cref{thm_A_cyc_shift_extension} then yields that $Q_u$ and $Q_v$ commute for every $Q\colon \rel A^n \qto \rel A$ and $u,v\in A^n$, which means that $\qPol(\rel A)=\qcPol(\rel A)$. 
\end{proof}

We can finally conclude the proof of the main result of this section.

\begin{proof}[Proof of~\Cref{thm_siggers_undecidable}]
It is known that $\A$ pp-constructs 3SAT~\cite{BartoKozikNiven}.
Moreover, all quantum polymorphisms of $\A$ are non-contextual by~\Cref{thm_noncontectual_A}.
It follows from~\Cref{cor_np_complete_plus_no_contextual_means_undecidable_ALGEBRAIC} that $\qCSP(\A)$ is undecidable.
\end{proof}

\section{Quantum polymorphisms of odd cycles}

\label{sec_odd_cycles}
The goal of this section is to prove that the quantum CSP parameterised by an odd cycle $\rel C_m$ is undecidable. This resolves an open question raised in~\cite{Zeman}.
\thmquantumoddcyclesisundecidable*
To prove this result, we will show that all quantum polymorphisms of odd cycles are non-contextual---and, thus, odd cycles admit commutativity gadgets.

\begin{theorem}
\label{thm_odd_cycles_qPol_equals_qcPol}
   $\qPol(\rel C_m)=\qcPol(\rel C_m)$ for each odd $m\geq 3$. 
\end{theorem}
We start with the following observation. Henceforth in this section, we fix an odd integer $m\in\mathbb{N}$.
\begin{lemma}
    \label{lem_07101220}
    Let $Q$ be a quantum endomorphism of $\rel C_m$ and let $x,y,z$ be a $2$-walk in $\rel C_m$. Then $Q_{x,a}Q_{y,b}Q_{z,c}=Q_{y,b}Q_{z,c}$ for each $a,b,c\in\Z_m$.
\end{lemma}
\begin{proof}
    If $x=z$ the result is trivial, so we may assume that $x\neq z$. We know from~\Cref{cor_quantum_endo_odd_cycles_ess_classical} that there exist $d\in\N$ and $h_1,\dots, h_d\in\Hom(\rel C_m,\rel C_m)$ such that $Q=\bigoplus_{i\in[d]}h_i$. Using~\cref{eqn_1240_0710}, we can write
    \begin{align*}
        Q_{x,a}Q_{y,b}Q_{z,c}
        &=
        \left(\sum_{\substack{i\in[d]\\h_i(x)=a}}\be_i\be_i^*\right)
        \left(\sum_{\substack{j\in[d]\\h_j(y)=b}}\be_j\be_j^*\right)
        \left(\sum_{\substack{k\in[d]\\h_k(z)=c}}\be_k\be_k^*\right)
        =
        \sum_{\substack{i\in[d]\\h_i(x)=a\\ h_i(y)=b\\ h_i(z)=c}}\be_i\be_i^*.
    \end{align*}
    Let $i\in[d]$ give a non-trivial contribution to the sum above. Since $x,y,z$ form a path in $\rel C_m$, any choice of $(b,c)\in\Z_m^2$ for which $h_i(y)=b$ and $h_i(z)=c$ gives precisely one choice of $a\in\Z_m$ such that $h_i(x)=a$. It follows that
    \begin{align*}
        Q_{x,a}Q_{y,b}Q_{z,c}
        &=
        \sum_{\substack{i\in[d]\\ h_i(y)=b\\ h_i(z)=c}}\be_i\be_i^*
        =
        Q_{y,b}Q_{z,c},
    \end{align*}
    as needed.
\end{proof}

Observe that the non-contextuality of \textit{unary} quantum polymorphisms of odd cycles---i.e., quantum homomorphisms $\rel C_m\qto\rel C_m$---is clear from~\Cref{cor_quantum_endo_odd_cycles_ess_classical}. Hence, in order to prove~\Cref{thm_odd_cycles_qPol_equals_qcPol}, we need to extend non-contextuality to higher arities.
Henceforth in this section, we fix some arity $n\geq 2$. Furthermore, whenever a plus or minus sign is applied to an element of $\Z_m$ or to a tuple of such elements, it shall be meant modulo $m$.

\begin{lemma}
\label{lemma_commutator_cycle_same_answer}
    Fix two vertices $\bx,\by$ of $\rel C_m^n$ such that $\dist_{\rel C_m^n}(\bx,\by)=2$. For any $a\in \Z_m$, it holds that
    \begin{align*}
        [Q_{\bx,a},Q_{\by,a}]=0.
    \end{align*}
\end{lemma}
\begin{proof}
In this proof, we shall abbreviate $Q_{\ba,a}$ to $Q_{\ba}$ for each $\ba\in [m]^n$.
Let $\bx\to\bz\to\by$ be a $2$-path connecting $\bx$ and $\by$ in $\rel C_m^n$. Let also $\bu=\bz-\bx$ and $\bv=\by-\bz$, and notice that $\bu,\bv\in\{\pm 1\}^n$. 
\ignore{
If $\bu=\bv$, notice that the vertices $\bx,\bx+\bu,\bx+2\bu=\by,\bx+3\bu,\dots,\bx+(m-1)\bu$ form a $\rel C_m$ within $\rel C_m^n$, and thus the result follows from~\Cref{prop_zeman_endo_auto_odd_cycles}. Hence, we can safely assume that $\bu\neq \bv$, so there exists some coordinate $\alpha\in [n]$ such that $u_\alpha\neq v_\alpha$.}
Observe that the vertices $\by,\by+\bv,\by+2\bv,\dots,\by+(m-1)\bv$ induce a $\rel C_m$ within $\rel C_m^n$. The restriction of $Q$ to such vertices is a quantum endomorphism of $\rel C_m$. Therefore, combining~\Cref{cor_quantum_endo_odd_cycles_ess_classical} and~\Cref{lem_ess_classical_endo_automorphism}, we obtain that
\begin{align*}
    \sum_{i\in\Z_m}Q_{\by+i\bv}=\id.
\end{align*}
Consider the two paths 
\begin{align*}
&\bx\to\bx+\bu\to\by\to\by+\bv\to\by+2\bv\to\dots\to\by+(m-4)\bv,\\
&\bx\to\bx+\bu\to\by-2\bv\to\by-3\bv\to\dots\to\by-(m-2)\bv,
\end{align*}
and notice that both have length $m-2$. Since there is no closed walk in $\rel C_m$ of odd length less than $m$, it follows from~\Cref{prop_basic_walks_essentially_binary} that $Q_\bx Q_{\by+i\bv}=0$ for each $i\in\Z_m\setminus\{0,-2\}$. 
Therefore,
\begin{align*}
    0&=
    Q_{\by+\bv}Q_{\bx}
    =
    \left(\id-\sum_{i\in\Z_{m}\setminus\{1\}}Q_{\by+i\bv}\right)Q_{\bx}
    =
    Q_{\bx}-(Q_{\by}+Q_{\by-2\bv})Q_{\bx}.
\end{align*}
Hence, $Q_{\bx}=(Q_{\by}+Q_{\by-2\bv})Q_{\bx}$. Since projectors are self-adjoint, it follows that
\begin{align*}
    [Q_{\bx},Q_{\by}+Q_{\by-2\bv}]=0.
\end{align*}
Consider now the automorphism $\xi$ of $\rel C_m^n$ that acts as the reflection around $x_\alpha$ for each coordinate $\alpha\in [n]$ for which $u_\alpha=v_\alpha$, and acts as the identity for all other coordinates. Note that both $\bx$ and $\bx+\bu-\bv=\by-2\bv$ are fixed points of $\xi$, while $\xi(\by)=2\bx-\by$. Since the previous part of the argument is invariant under (classical) automorphisms of $\rel C_m^n$, we deduce that 
\begin{align*}
[Q_{\bx},Q_{2\bx-\by}+Q_{\by-2\bv}]=0.
\end{align*}
By the bilinearity of the commutator, it follows that
\begin{align*}
    [Q_{\bx},Q_{\by}]=[Q_{\bx},Q_{2\bx-\by}].
\end{align*}
We now use the same argument after applying the automorphism of $\rel C_m^n$ given by $\ba\mapsto \ba+\by-\bx$. We obtain
\begin{align*}
     [Q_{\bx},Q_{\by}]=[Q_{\bx},Q_{2\bx-\by}]
     &=
     [Q_{3\bx-2\by},Q_{2\bx-\by}]
     =
     [Q_{3\bx-2\by},Q_{4\bx-3\by}]
     =
     \dots\\
     &=
     [Q_{m\bx-(m-1)\by},Q_{(m+1)\bx-m\by}]
     =
     [Q_{\by},Q_{\bx}],
\end{align*}
whence we conclude that $[Q_\bx,Q_\by]=0$, as required.
\end{proof}

\begin{lemma}
\label{lem_1146_0710}
    Fix two vertices $\bx,\by$ of $\rel C_m^n$ such that $\dist_{\rel C_m^n}(\bx,\by)=2$. For any $a,b\in \Z_m$, it holds that
    \begin{align*}
        [Q_{\bx,a},Q_{\by,b}]=0.
    \end{align*}
\end{lemma}
\begin{proof}
    We can assume without loss of generality that $b=a\pm 2$ $\mod m$, as otherwise $Q_{\bx,a}Q_{\by,b}=0$ by~\Cref{prop_basic_walks_essentially_binary}. Up to symmetry, we can further suppose that $b=a+2$ $\mod m$. Let $\bu$ and $\bv$ be as in the proof of~\Cref{lemma_commutator_cycle_same_answer}, and observe that
    \begin{align*}
        Q_{\bx,a}Q_{\by,a+2}
        &=
        Q_{\bx,a}\left(\sum_{b\in\Z_m}Q_{\by-\bv,b}\right)\left(\sum_{c\in\Z_m}Q_{\by-2\bv,c}\right)Q_{\by,a+2}\\
        &=
        \sum_{b,c\in\Z_m}Q_{\bx,a}Q_{\by-\bv,b}Q_{\by-2\bv,c}Q_{\by,a+2}.
    \end{align*}
    Since $[Q_{\by-\bv,b},Q_{\by-2\bv,c}]=0$,
    for a product in the summation above to be non-zero, it must be the case that $b$ is a neighbour of both $a$ and $a+2$; i.e., $b=a+1$. Consider now $Q_{\by-2\bv,c}$. For the product to be non-zero, since $\by-2\bv$ is a $\rel C_m^n$-neighbour of $\by-\bv$, it must hold that $c$ is a $\rel C_m$-neighbour of $b$; i.e., $c=a$ or $c=a+2$. Suppose $Q_{\by-2\bv,a+2}\neq 0$. We know from~\Cref{lemma_commutator_cycle_same_answer} that $[Q_{\by-2\bv,a+2},Q_{\by,a+2}]=0$ and, hence, the two projectors can be diagonalised simultaneously. Hence, if it was the case that $Q_{\by-2\bv,a+2}Q_{\by,a+2}\neq 0$, the ranges of $Q_{\by-2\bv,a+2}$ and of $Q_{\by,a+2}$ would have non-trivial intersection. But this would imply that $\sum_{i\in\Z_m}Q_{\by+i\bv,a+2}\neq\id$, contradicting~\Cref{cor_quantum_endo_odd_cycles_ess_classical} and~\Cref{lem_ess_classical_endo_automorphism}. It follows that the only choice for $c$ giving a non-trivial contribution is $c=a$, so
    \begin{align*}
    Q_{\bx,a}Q_{\by,a+2}
        =
        Q_{\bx,a}Q_{\by-\bv,a+1}Q_{\by-2\bv,a}Q_{\by,a+2}.
    \end{align*}
    By a similar argument, we have
    \begin{align*}
        Q_{\by-\bv,a+1}Q_{\by-2\bv,a}
        &=
        Q_{\by-\bv,a+1}Q_{\by-2\bv,a}Q_{\by,a+2}.
    \end{align*}
    Hence, we obtain
    \begin{align*}
         Q_{\by,a+2}Q_{\bx,a}Q_{\by,a+2}
        &=
        Q_{\by,a+2}Q_{\bx,a}Q_{\by-\bv,a+1}Q_{\by-2\bv,a}Q_{\by,a+2}\\
        &=Q_{\by,a+2}Q_{\bx,a}Q_{\by-\bv,a+1}Q_{\by-2\bv,a}\\
        &=Q_{\by,a+2}Q_{\by-\bv,a+1}Q_{\by-2\bv,a}Q_{\bx,a}
        \intertext{(where we are using that $[Q_{\bx,a},Q_{\by-\bv,a+1}]=0$ as $\bx\sim\by-\bv$ and $[Q_{\bx,a},Q_{\by-2\bv,a}]=0$ by~\Cref{lemma_commutator_cycle_same_answer})}
        &=Q_{\by,a+2}Q_{\by-2\bv,a}Q_{\by-\bv,a+1}Q_{\bx,a}\\
        &=
        (Q_{\bx,a}Q_{\by,a+2})^*
        =
        Q_{\by,a+2}Q_{\bx,a}.
    \end{align*}
    In particular, we deduce that $Q_{\by,a+2}Q_{\bx,a}$ is self-adjoint, so
    \begin{align*}
       [Q_{\bx,a},Q_{\by,a+2}]=0, 
    \end{align*}
    as needed.
\end{proof}

We can now finalise the proof of~\Cref{thm_odd_cycles_qPol_equals_qcPol}.
\begin{proof}[Proof of~\Cref{thm_odd_cycles_qPol_equals_qcPol}]
Let $\bx$ and $\by$ be two vertices of $\rel C_m^n$ and let $a,b$ be two vertices of $\rel C_m$.
We need to show that $[Q_{\bx,a},Q_{\by,b}]=0$. We shall use induction on $\delta=\dist_{\rel C_m^n}(\bx,\by)$. If $\delta\in\{0,1\}$, the result is clear from the definition of quantum homomorphisms. If $\delta=2$, the result is~\Cref{lem_1146_0710}. Let then $\delta\geq 3$, and suppose that the result holds for all pairs of vertices at mutual distance at most $\delta-1$. Let $\bx,\bz_1,\bz_2,\dots,\bz_{\delta-1},\by$ be a shortest path from $\bx$ to $\by$ in $\rel C_m^n$. Since $\delta\geq 3$, we can assume without loss of generality that the vectors $\bz_1-\bx$ and $\bz_2-\bz_1$ are equal.\footnote{Note that this is false in general when $\dist_{\rel C_m^n}(\bx,\by)=2$. For instance, consider the case that $n=2$, $m\geq 5$, $\bx=(0,0)$, and $\by=(0,2)$.} We now expand the product $Q_{\bx,a}Q_{\by,b}$ as follows:
\begin{align*}
    Q_{\bx,a},Q_{\by,b}
    &=
    Q_{\bx,a}\left(\sum_{c_1\in\Z_m}Q_{\bz_1,c_1}\right)\left(\sum_{c_2\in\Z_m}Q_{\bz_2,c_2}\right)\dots\left(\sum_{c_{\delta-1}\in\Z_m}Q_{\bz_{\delta-1},c_{\delta-1}}\right)Q_{\by,b}\\
    &=
    \sum_{\bc\in\Z_m^{\delta-1}}Q_{\bx_\ba}Q_{\bz_1,c_1}Q_{\bz_2,c_2}\dots Q_{\bz_{\delta-1}c_{\delta-1}}Q_{\by,b}.
\end{align*}
Note now that $\dist(\bz_i,\by)\leq \delta-1$ for each $i\in[\delta-1]$. Hence, the inductive hypothesis allows us to freely permute the corresponding projectors in the expression above, thus yielding
\begin{align}
\label{eqn_1221_0710}
    Q_{\bx,a}Q_{\by,b}
    &=\sum_{\bc\in\Z_m^{\delta-1}}Q_{\bx_\ba}Q_{\by,b}Q_{\bz_1,c_1}Q_{\bz_2,c_2}\dots Q_{\bz_{\delta-1}c_{\delta-1}}.
\end{align}
Multiplying both terms by $Q_{\bx,a}$ on the right and further rearranging gives
\begin{align*}
    Q_{\bx,a}Q_{\by,b}Q_{\bx,a}
    &=\sum_{\bc\in\Z_m^{\delta-1}}Q_{\bx_\ba}Q_{\by,b}Q_{\bz_1,c_1}Q_{\bz_2,c_2}\dots Q_{\bz_{\delta-1}c_{\delta-1}}Q_{\bx,a}\\
    &=\sum_{\bc\in\Z_m^{\delta-1}}Q_{\bx_\ba}Q_{\by,b}Q_{\bx,a}Q_{\bz_1,c_1}Q_{\bz_2,c_2}\dots Q_{\bz_{\delta-1}c_{\delta-1}}.
\end{align*}
Observe now that $\bx$, $\bz_1$, and $\bz_2$ lie on a common $m$-cycle within $\rel C_m^n$ since $\bz_1-\bx=\bz_2-\bz_1$.
By~\Cref{lem_07101220}, it holds that $Q_{\bx,a}Q_{\bz_1,c_1}Q_{\bz_2,c_2}=Q_{\bz_1,c_1}Q_{\bz_2,c_2}$ for each $c_1,c_2\in\Z_m$. We deduce that 
\begin{align*}
    Q_{\bx,a}Q_{\by,b}Q_{\bx,a}
    &=\sum_{\bc\in\Z_m^{\delta-1}}Q_{\bx_\ba}Q_{\by,b}Q_{\bz_1,c_1}Q_{\bz_2,c_2}\dots Q_{\bz_{\delta-1}c_{\delta-1}}
    =
    Q_{\bx,a}Q_{\by,b},
\end{align*}
where the second equality is~\cref{eqn_1221_0710}. From this, we conclude that $Q_{\bx,a}$ and $Q_{\by,b}$ commute, as required.
\end{proof}

The undecidability of $\qCSP(\rel C_m)$ is now a simple consequence of the machinery developed in~\Cref{sec_quantum_algebraic_constructions} and~\Cref{sec_quantum_relational_constructions}.

\begin{proof}[Proof of~\Cref{thm_quantum_odd_cycles_is_undecidable}]
    $\rel K_m$ is pp-definable in $\rel C_m$ with the gadget consisting of a path of length $m-2$; hence, $\Pol(\rel C_m)\to\Pol(\K_m)\to\Pol(\operatorname{3SAT})$, and the result follows by combining~\Cref{thm_odd_cycles_qPol_equals_qcPol} and~\Cref{cor_np_complete_plus_no_contextual_means_undecidable_ALGEBRAIC}.
\end{proof}

We conclude this section by proving that quantum homomorphisms between different odd cycles can be contextual.
\begin{proposition}
\label{prop_contextual_homomorphisms_C7_to_C5}
    There is a quantum homomorphism $\rel C_7\qto\rel C_5$ over a 2-dimensional Hilbert space that is not in $\qc{\Hom(\rel C_7,\rel C_5)}$.
\end{proposition}
\begin{proof}
    Consider for this proof $\rel C_5$ to have domain $\{0,1,2,5,6\}$ with edges between adjacent numbers, and $\rel C_7$ to have domain $[7]$.
    Consider the two classical endomorphisms $h,g\colon \rel C_5\to\rel C_5$ such that
    $h(2)=6,h(5)=0$ and $g(5)=0,g(2)=1$.
    Then $h\oplus g$ is a quantum homomorphism $Q\colon \rel C_5\qto\rel C_5$ over a 2-dimensional Hilbert space.
    Extend $Q$ by setting $Q_{3,a}=Q_{5,a}$ for $a\in\{0,1,2,5,6\}$.
    Choose a basis in which all the projectors are diagonal, and in this basis define
    $Q_{4,1}=\frac12\begin{pmatrix}1 & 1\\1 &1\end{pmatrix}$ and $Q_{4,6}=\frac12\begin{pmatrix}1&-1\\-1&1\end{pmatrix}$.
    Then $Q\colon\rel C_7\qto\rel C_5$ and $[Q_{4,1},Q_{2,1}]\neq 0$ since $Q_{2,1}=\begin{pmatrix}
        0&0\\0&1
    \end{pmatrix}$.
\end{proof}

\section{Non-oracular polymorphisms} 
\label{sec_nonoracular}

In this section, we consider the \textit{non-oracular} version of the results described in~\Cref{sec_quantum_algebraic_constructions} and~\Cref{sec_quantum_relational_constructions}. A \textit{non-oracular quantum homomorphism} between two $\sigma$-structures $\A$ and $\B$, denoted by $Q\colon\A\qnoto\B$, is a quantum function $Q\colon A\qto B$ 
satisfying the condition~\eqref{quantum_homo_1} but possibly not~\eqref{quantum_homo_2}; see~\Cref{subsec_prelimns_quantum_CSPs}.
The following is the non-oracular version of the quantum CSP.
\begin{definition}[Non-oracular quantum CSP]
Fix a  $\sigma$-structure $\A$.
    $\qnoCSP(\rel A)$ is the following computational problem: Given as input a $\sigma$-structure $\rel X$, determine whether there exists a non-oracular quantum homomorphism $\rel X\qnoto\rel A$. 
\end{definition}

For CSPs parameterised by digraphs, non-oracular quantum homomorphisms capture perfect strategies for the variant of the 2-player CSP game where both cooperating players receive some variable ($x_A$ and $x_B$, respectively) from the verifier and reply with assignments to their variables with some values $a_A$ and $a_B$, respectively. The answers are deemed correct by the verifier if $a_A=a_B$ whenever $x_A=x_B$ (i.e., the game is synchronous), and $(a_A,a_B)$ is an edge in the template digraph whenever $(x_A,x_B)$ is an edge in the instance. The relation between this game and the type of quantum homomorphisms given above was first described in~\cite{MancinskaR16} (in the case of undirected graphs) as a generalisation of the quantum colouring game of~\cite{CameronMNSW07}, see also~\cite{MancinskaR16_oddities}.
Various other works have considered this setting for specific types of binary constraint systems; see, e.g., 
the non-oracular strategies for Unique Projection Games~\cite{Khot02stoc} considered in~\cite{kempe2010unique}.

We next define quantum polymorphisms in the non-oracular setting.

\begin{definition}
    Let $\rel A$ be a relational structure.
    A \emph{non-oracular quantum polymorphism} is a non-oracular quantum homomorphism $Q\colon\rel A^n\qnoto\rel A$.
    We denote by $\qnoPol(\rel A)$ the set of non-oracular quantum polymorphisms of $\rel A$.
\end{definition}
As the main result of this section, we shall prove that, just like in the oracular case~\Cref{prop_galois_oracular}, the Galois connection between classical polymorphism clones and pp-definitions~\cite{bodnarchuk1969galois,geiger1968closed} can be lifted to a Galois connection for non-oracular quantum CSPs. As a consequence, we will prove the following result.

\begin{restatable}{theorem}{mainthmgaloisconditionreductionsnonoracular}
\label{main_thm_galois_condition_reductions_nonoracular}
    Let $\rel A,\rel B$ be structures on potentially different signatures such that
    \begin{align*}
        \qcPol(\rel A)=\qnoPol(\rel A)\subseteq\qnoPol(\rel B).
    \end{align*} 
    Then $\qCSP(\rel B)$ reduces to $\qnoCSP(\rel A)$ in logspace.
\end{restatable}

It should be observed that, in this case, we obtain a reduction from the \textit{oracular} version of the quantum CSP parameterised by $\B$ to the \textit{non-oracular} version of the quantum CSP parameterised by $\A$. As we shall see, this is a consequence of how commutativity in the contexts is handled by the reduction.

We begin by defining the non-oracular version of q-definitions.

\begin{definition}\label{qgadget_non_oracular}
    Let $\rel A$ be a $\sigma$-structure    and let $S\subseteq A^r$ be a set of $r$-tuples of elements of $A$ for some $r\in\N$. Let $R$ be the $r$-ary symbol of the signature of $(A;S)$ (i.e., $R^{(A;S)}=S$), and 
    consider the $\{R\}$-structure
    $\rel R$ having domain $[r]$ and relation $R^{\rel R}=\{(1,\dots,r)\}$.
    A \emph{non-oracular q-definition} (in short, a \textit{q-no-definition}) for $\rel A$ and $S$ is a $\sigma$-structure $\GG$ with distinguished vertices $g_1,\dots,g_r$ satisfying the following properties:
    \begin{enumerate}
        \item[$(\mathbf{noq_1})$\labeltext{$\mathbf{noq_1}$}{itm:qgadget-extension_non_oracular}] For every quantum homomorphism $Q\colon\rel R\qto (A;S)$, 
        there exists a non-oracular quantum homomorphism $Q'\colon\GG\qnoto\rel A$ such that $Q'_{g_i,a} = Q_{i,a}$ for all $i\in[r]$ and all $a\in A$.
        \item[$(\mathbf{noq_2})$\labeltext{$\mathbf{noq_2}$}{itm:qgadget-restriction_non_oracular}] For every non-oracular quantum homomorphism $Q\colon\GG\qnoto\rel A$, 
        the quantum function $Q'$ defined by $Q'_{i,a}:=Q_{g_i,a}$ for all $i\in[r]$ is a non-oracular quantum homomorphism $\rel R\qnoto (A;S)$.
    \end{enumerate}
\end{definition}

Note that, in the first item of the definition above, $Q$ is required to be oracular (i.e., to satisfy~\eqref{quantum_homo_2}), while $Q'$ need not.
\begin{definition}
    Let $\A$ and $\B$ be structures on the same domain and potentially different signatures. We say that  $\A$ \textit{q-no-defines} $\B$ if there exists a non-oracular q-definition for $\A$ and $R^\B$ for every relation $R^\B$ of $\B$.
\end{definition}

In the following, it shall be useful to observe that~\Cref{sum-homomorphisms}, \Cref{lem_ess_classical_endo_automorphism},~\Cref{tensor-homomorphisms}, and~\Cref{composition-homomorphisms} also hold in the non-oracular setting. 
Furthermore, the definition of quantum closure of $\Hom(\A,\B)$ (\Cref{quantum-closure}) does not change in the non-oracular setting.
The next result is the non-oracular version of~\Cref{galois-connection-oracular}, and is the key to the quantisation of the relation/operation Galois connection in the non-oracular quantum setting.
\begin{theorem}\label{galois}
    Let $\rel A$ be a relational structure and let $S\subseteq A^r$ for some $r\in\N$.
    Then there exists a non-oracular q-definition for $\rel A$ and $S$ if, and only if, $\qnoPol(\rel A)\subseteq\qnoPol((A;S))$.
\end{theorem}
\begin{proof}
Henceforth in this proof, we shall denote $(A;S)$ by $\D$ for the sake of readability.
\begin{itemize}
    \item[$(\Rightarrow)$]
Let $\GG$ be a non-oracular q-definition for $\A$ and $S$. 
    This means that every quantum homomorphism $\rel R\qto\D$ extends to a non-oracular quantum homomorphism $\GG\qnoto\rel A$, and every non-oracular quantum homomorphism $\GG\qnoto\rel A$ restricts to a non-oracular quantum homomorphism $\rel R\qnoto\D$.

    Fix now a non-oracular quantum polymorphism  $Q$ of $\A$; i.e.,  $Q\colon\rel A^n\qnoto\rel A$. 
    We want to prove that $Q\in\qnoPol(\D)$.
    Note that $Q$ is a quantum function $A^n\qto A$, so we only need to prove~\eqref{quantum_homo_1} applied to the relation of $\D^n$; i.e., to the relation $R^{\D^n}$
    containing all $r$-tuples $(\bb^{(1)},\dots,\bb^{(r)})$ where each $\bb^{(i)}$ is an element of $A^n$, and for each $j\in[n]$ it holds that $(b^{(1)}_j,\dots,b^{(r)}_j)\in R^{\D}=S$.  

    Fix then such a tuple $(\bb^{(1)},\dots,\bb^{(r)})\in R^{\D^n}$, and choose a tuple $\bc\in A^r\setminus S$. For every $i\in [n]$, consider the deterministic homomorphism $f_i:\rel R\to\D$ defined by $f_i(j)=b^{(j)}_i$ for each $j\in[r]$. In particular, $f_i$ is a quantum homomorphism. Hence,~\eqref{itm:qgadget-extension_non_oracular} implies that there exists a non-oracular quantum homomorphism $W^{(i)}\colon\GG\qnoto\A$ extending $f_i$. Consider now the tensor product 
    \begin{align*}
        W=\bigotimes_{i\in[n]}W^{(i)}\colon\GG\qnoto \A^n,
    \end{align*}
    and compose it with $Q$ to obtain $Q\bullet W\colon\GG\qnoto\A$.
    Fix $\ell\in [r]$ and $a\in A$. Expanding the composition  and the tensor product  yields
    \begin{align*}
        (Q\bullet W)_{g_\ell,a}
        =
        \sum_{\ba\in A^n}Q_{\ba,a}\otimes W_{g_{\ell},\ba}=
        \sum_{\ba\in A^n}Q_{\ba,a}\otimes\bigotimes_{i\in[n]} W^{(i)}_{g_{\ell},a_i}.
    \end{align*}
    We know from~\eqref{itm:qgadget-extension_non_oracular} that $W^{(i)}_{g_{\ell},a_i}=(f_i)_{\ell,a_i}$. Since $f_i$ is deterministic, this means that in order for the $\ba$-th summand in the sum above to be non-zero it must be that $a_i=f_i(\ell)=b_i^{(\ell)}$ for each $i\in[n]$; i.e., $\ba=\bb^{(\ell)}$. Hence, 
    \begin{align*}
        (Q\bullet W)_{g_\ell,a}
        =
        Q_{\bb^{(\ell)},a}\otimes\id.
    \end{align*}
    On the other hand,~\eqref{itm:qgadget-restriction_non_oracular}
    implies that the assignment $(\ell,a)\mapsto (Q\bullet W)_{g_\ell,a}$ yields a non-oracular quantum homomorphism from $\rel R$ to $\D$. Since we are assuming that $\bc\not\in S=R^{\D}$, we deduce that
    \begin{align*}
        0=\prod_{\ell\in[r]}(Q\bullet W)_{g_\ell,c_\ell}
        =
        \prod_{\ell\in[r]}(Q_{\bb^{(\ell)},a}\otimes\id)
        =
        \left(\prod_{\ell\in[r]}Q_{\bb^{(\ell)},a}\right)\otimes\id.
    \end{align*}
    Therefore, we find $\prod_{\ell\in[r]}Q_{\bb^{(\ell)},a}=0$. It follows that $Q$ is a non-oracular quantum homomorphism from $\D^n$ to $\D$; i.e., an element of $\qnoPol(\D)$, as desired.\\

    \item[$(\Leftarrow)$] Suppose now that $\qnoPol(\rel A)\subseteq\qnoPol(\D)$. Let us enumerate the tuples in $S$ as $\ba^{(1)},\dots,\ba^{(n)}$, where $n=|S|$. We will show that $\A^n$ is a non-oracular quantum gadget for $\A$ and $S$, with distinguished vertices $g_1,\dots,g_r$ defined by $g_\ell=(a^{(1)}_\ell,\dots,a^{(n)}_\ell)$ for $\ell\in[r]$.
    We show that the two conditions in~\Cref{itm:qgadget-extension_non_oracular,itm:qgadget-restriction_non_oracular} hold.

    Let $Q\colon\rel R\qto\D$ be a quantum homomorphism over some Hilbert space $H$, and let $d$ be the dimension of $H$. Since $Q$ satisfies~\eqref{quantum_homo_2}, $Q$ must be classical, so by the spectral theorem we can decompose it as $Q=\bigoplus_{i\in[d]} h_i$, where $h_i\colon\rel R\to\D$ is a deterministic homomorphism for each $i\in[d]$. By the definition of $\rel R$, it must be that $(h_i(1),\dots,h_i(r))\in S$ for each $i\in [d]$.
    In other words, there exists a function $f:[d]\to[n]$ satisfying
    \begin{align}
    \label{eqn_1208_2111}
        (h_i(1),\dots,h_i(r))=\ba^{(f(i))}
    \end{align}
    for each $i\in[d]$.
    Consider now the classical projection polymorphisms $\pi^n_{j}:\A^n\to\A$ defined by $\pi^n_j(\ba)=a_j$ for each $j\in[n]$ and each $\ba\in A^n$.
    Define 
    \begin{align*}
    P=\bigoplus_{i\in[d]}\pi^n_{f(i)}.
    \end{align*}
    Note that $P$ is a classical (so, a fortiori, non-oracular quantum) homomorphism from $\A^n$ to $\A$. Take now $\ell\in[r]$ and $a\in A$, and observe that
    \begin{align*}
        P_{g_\ell,a}
        =
        \bigoplus_{i\in[d]}(\pi^n_{f(i)})_{g_\ell,a}
        =
        \bigoplus_{\substack{i\in[d]\\\pi^n_{f(i)}(g_\ell)=a}}\id
        =
        \bigoplus_{\substack{i\in[d]\\a^{f(i)}_\ell=a}}\id
        =
        \bigoplus_{\substack{i\in[d]\\h_i(\ell)=a}}\id
        =
        \bigoplus_{\substack{i\in[d]}}(h_i)_{\ell,a}
        =
        Q_{\ell,a}
    \end{align*}
    (where the fourth equality follows from~\Cref{eqn_1208_2111}). 
Hence, $P$ witnesses that~\eqref{itm:qgadget-extension_non_oracular} holds.

    Finally, let $Q\colon\rel \A^n\qnoto\rel A$ be a non-oracular quantum polymorphism, and consider the quantum function $Q'$ given by $Q'_{\ell,a}=Q_{g_\ell,a}$ for $\ell\in [r]$ and $a\in A$.
    Recall now that each $\ba^{(i)}$ is an element of $S=R^\D$ and, thus, the tuple $(g_1,\dots,g_r)$ belongs to $R^{\D^n}$ by the definition of categorical power of structures. Since $Q$ is a non-oracular quantum polymorphism of $\rel A$, by assumption it is also an $n$-ary non-oracular quantum polymorphism of $\D$.     
    Hence, for each $\bc\in A^r\setminus S$, we must have 
    \begin{align*}
        0=\prod_{\ell\in[r]}Q_{g_\ell,c_\ell}=\prod_{\ell\in[r]}Q'_{\ell,c_\ell},
    \end{align*}
    which means that $Q'$ is a non-oracular quantum homomorphism from $\R$ to $\D$. This shows that~\eqref{itm:qgadget-restriction_non_oracular} is also satisfied, and concludes the proof of the theorem.\qedhere 
    \end{itemize}
\end{proof}

The following result is an immediate consequence of~\Cref{galois}.

\corgaloisnonoracular*

Unlike in the oracular setting, non-oracular q-gadgets alone are not strong enough to guarantee reductions. To that end, they need to be strengthened via a non-oracular commutativity gadget, which we now define.
\begin{definition}\label{comm-gadget_non_oracular}
    Fix a $\sigma$-structure $\rel A$.
    A \emph{non-oracular commutativity gadget} for $\rel A$ is a $\sigma$-structure $\GG$ with two distinguished elements $u,v\in G$ satisfying the following properties:
    \begin{enumerate}
        \item[$(\mathbf{noc_1})$\labeltext{$\mathbf{noc_1}$}{itm:commgadget-extension_nonoracular}] For any quantum function $Q\colon\{u,v\}\qto A$ such that $Q_{u,a}$ and $Q_{v,b}$ commute for all $a,b\in A$, %
        there exists a quantum non-oracular homomorphism $Q'\colon\GG\qnoto\rel A$ extending $Q$.
        \item[$(\mathbf{noc_2})$\labeltext{$\mathbf{noc_2}$}{itm:commgadget-restriction_nonoracular}] For every quantum non-oracular homomorphism $Q\colon\GG\qnoto\rel A$ and all $a,b\in A$, it holds that $[Q_{u,a},Q_{v,b}]=0$.
    \end{enumerate}
\end{definition}
Note that, unlike in the oracular case~\Cref{comm-gadget}, a non-oracular commutativity gadget is not the same as a q-no-definition for $A^2$ in $\A$.
We now give the non-oracular analogue of~\Cref{gadget-characterization}.
Just like in the oracular setting, the characterisation of commutativity gadgets via quantum polymorphisms (stated in the result below) also holds for the weaker \textit{algebraic} non-oracular commutativity gadgets of~\cite{Zeman}, which are defined by replacing~\eqref{itm:commgadget-extension_nonoracular} in~\Cref{comm-gadget_non_oracular}  with the weaker condition stated below:
\begin{enumerate}
\item[$(\mathbf{noc_1'})$\labeltext{$\mathbf{noc_1'}$}{itm:commgadget-extension_nonoracular_zeman}] For every $a,b\in A$, there exists a quantum non-oracular homomorphism $Q\colon\GG\qnoto\rel A$ such that $Q_{u,a}=\id$ and $Q_{v,b}=\id$.
\end{enumerate}

\begin{theorem}\label{gadget-characterization_non_oracular}
Let $\rel A$ be a relational structure. The following are equivalent:
\begin{enumerate}
    \item $\rel A$ has a non-oracular commutativity gadget.
    \item $\rel A$ has an algebraic non-oracular commutativity gadget.
    \item $\rel A^n$ is an (algebraic) non-oracular commutativity gadget for all large enough $n\in\N$.
    \item $\qnoPol(\rel A)=\qcPol(\rel A)$.
\end{enumerate}
\end{theorem}

\begin{proof}
We prove the following implications.
    \begin{itemize}
        \item[$(4.\Rightarrow 3.)$] The proof follows exactly the same argument as for~\Cref{obtaining-commgadget}.
        \item[$(3.\Rightarrow 1.)$] Obvious.
        \item[$(1.\Rightarrow 2.)$] This is clear, since each non-oracular commutativity gadget is in particular an algebraic non-oracular commutativity gadget (via the same argument as in~\Cref{rem_algebraic_gadgets}).
        \item[$(2.\Rightarrow 4.)$] Take an algebraic non-oracular commutativity gadget $\GG$ with distinguished elements $u,v$, a non-oracular quantum polymorphism $Q\colon\A^n\qnoto\A$, two tuples $\ba,\bb\in A^n$, and two elements $c,d\in A$. We proceed as in~\Cref{commgadget-implies-closure}. We first build a non-oracular quantum polymorphism $W^{(i)}\colon\GG\qnoto\A$ for each $i\in[n]$ with the property that $W^{(i)}_{u,a_i}=W^{(i)}_{v,b_i}=\id$, using~\eqref{itm:commgadget-extension_nonoracular_zeman}. Then, we take the tensor product $W=\bigotimes_{i\in [n]}W^{(i)}$ of such homomorphisms, and we define $U\colon\GG\qnoto\A$ as the composition of $Q$ and $W$. On the one hand,~\eqref{itm:qgadget-restriction_non_oracular} guarantees that $[U_{u,c},U_{v,d}]=0$. On the other hand, expanding the composition yields $U_{u,c}= Q_{\tuple a,c}\otimes\id$ and $U_{v,d}=Q_{\tuple b,d}\otimes\id$. Combining the two facts, we conclude that $[Q_{\tuple a,c},Q_{\tuple b,d}]=0$, and we deduce that $Q$ is classical.\qedhere
    \end{itemize}
\end{proof}

The next result shows that non-oracular q-definitions, combined with non-oracular commutativity gadgets, give logspace reductions from $\qCSP$ to $\qnoCSP$.

\begin{proposition}\label{qgadget-reduction_non_oracular}
    Let $\rel A$ be a $\sigma$-structure and let $\rel B$ be a $\rho$-structure having the same domain (but potentially distinct signatures $\sigma$ and $\rho$), and suppose that 
    \begin{itemize}
        \item $\rel A$ has a non-oracular commutativity gadget, and
        \item $\A$ q-no-defines $\B$.  
    \end{itemize}
    Then $\qCSP(\rel B)$ reduces to $\qnoCSP(\rel A)$ in logspace.
\end{proposition}
\begin{proof}
For every symbol $R\in\rho$ of some arity $r$, let $\GG_R$ be a non-oracular q-definition for $\rel A$ and $R^{\rel B}$ (with its distinguished vertices $g_1,\dots,g_r$).
    Let also $\rel H$ be a non-oracular commutativity gadget for $\rel A$ with its distinguished vertices $u,v$.

Let $\rel X$ be an instance of $\qCSP(\rel B)$. %
Henceforth in this proof, we shall view $\Gaif(\X)$ as the directed graph obtained by viewing each edge $\{x,y\}$ as two directed edges $(x,y)$ and $(y,x)$.
We define an instance $\rel Y$ of $\qnoCSP(\rel A)$ as follows.
    Let the domain of $\rel Y$ originally consist of the disjoint union of $X$, together with a copy $\GG_{R}^{(\bx)}$ of $\GG_R$ for each symbol $R\in\rho$ and each 
    tuple $\tuple x\in R^{\rel X}$, as well as a copy $\rel H^{(x,y)}$ of $\rel H$ for each edge $(x,y)$ of $\Gaif(\X)$. 
    For $R\in\rho$ of some arity $r\in\N$, $\bx\in R^\X$, and $i\in [r]$, we denote by $g_i^{(\bx)}$ the copy of $g_i$ in $\GG_{R}^{(\bx)}$.
    Similarly,
    we denote by $u^{(x,y)}$ (resp. $v^{(x,y)}$) the copy of $u$ (resp. $v$) in $\rel H^{(x,y)}$.
    For every edge $(x,y)$ of $\Gaif(\X)$, glue $u^{(x,y)}$ with $x$ and glue $v^{(x,y)}$ with $y$.
    For each $R\in\rho$ of some arity $r$, each $\tuple x\in R^{\rel X}$, and each $i\in [r]$, glue $g_i^{(\bx)}$ with $x_i$.
    Call this instance $\rel Y$.\\

    (\underline{\textit{Completeness.}})
    Suppose that $Q\colon\rel X\qto\rel B$ is a quantum homomorphism.
    For every symbol $R\in\rho$ of some arity $r\in\N$ and every
    $\tuple x\in R^{\rel X}$, consider the $\{R\}$-structure $\A'=(A;R^\B)$ with domain $A$ and relation $R^{\A'}=R^\B$. Recall also the $\{R\}$-structure $\rel R$ described in~\Cref{qgadget_non_oracular}.
    Setting $W_{i,a}=Q_{x_i,a}$ for each $i\in[r]$ and each $a\in A=B$ yields a quantum homomorphism $W\colon\rel R\qto \A'$.
    (Observe that this would not be guaranteed if $Q$ was a non-oracular quantum homomorphism.)
    Then,~\eqref{itm:qgadget-extension_non_oracular} guarantees the existence of a non-oracular quantum homomorphism $Q^{(\tuple x,R)}\colon\GG_{R}^{(\bx)}\qnoto\rel A$ such that \[Q^{(\tuple x,R)}_{g_i^{(\bx)},a}=Q_{x_i,a}\] for each $i\in [r]$ and each $a\in A$.
    Fix now an edge $(x,y)$ of $\Gaif(\rel X)$.
    Using again the fact that $Q$ is oracular, we know that $[Q_{x,a},Q_{y,b}]=0$ for every $a,b\in A$.
    Hence, we can apply~\eqref{itm:commgadget-extension_nonoracular} to find a non-oracular quantum homomorphism \[Q^{(x,y)}\colon\rel H^{(x,y)}\qnoto\rel A\] extending $Q$ in $x$ and $y$.
    Overall, this is a non-oracular quantum homomorphism $\rel Y\qnoto\rel A$.\\

    (\underline{\textit{Soundness.}})
    Suppose that $Q\colon\rel Y\qnoto\rel A$ is a non-oracular quantum homomorphism. Consider the quantum function $Q':X\to B=A$ obtained by restriction of $Q$ to $X\subseteq Y$. We claim that $Q'$ is an (oracular) quantum homomorphism from $\X$ to $\B$. To that end, take a symbol $R\in\rho$ of some arity $r\in\N$ and a tuple $\bx\in R^\X$.
    Take $i,j\in [r]$. We claim that the PVMs $Q'_{x_i}$ and $Q'_{x_j}$ commute. This is easily true if $x_i=x_j$ (as in this case the PVMs coincide), so suppose that $x_i\neq x_j$. In this case, $(x_i,x_j)$ is an edge in $\Gaif(\X)$. Note now that, by construction of $\Y$, the vertices $x_i$, $u^{(x_i,x_j)}$, and $g_i^{(\bx)}$ are identified in $Y$. Similarly, the vertices $x_j$, $v^{(x_i,x_j)}$, and $g_j^{(\bx)}$ are identified in $Y$. Since $Q$ induces a non-oracular quantum homomorphism over $\rel H^{(x_i,x_j)}$, we deduce from~\eqref{itm:commgadget-restriction_nonoracular} that the PVMs $Q'_{x_i}$ and $Q'_{x_j}$ commute also in this case.
    Hence, $Q'$ satisfies~\eqref{quantum_homo_2}.
    We now apply~\eqref{itm:qgadget-restriction_non_oracular} to deduce that the quantum function $U$ defined by $(i,a)\mapsto Q_{g_i^{(\bx)},a}=Q'_{x_i,a}$ is a non-oracular quantum homomorphism $\rel R\qnoto\A'$, where, as before, $\A'=(A;R^\B)$.
    Consider now a tuple $\bb\in A^r\setminus R^\B$.
    Since $(1,\dots,r)\in R^{\rel R}$ and $\bb\not\in R^\B= R^{\A'}$, we conclude that 
     \begin{align}
        \label{eqn_2011_2025}
        \prod_{i\in[r]}Q_{x_i,b_i}=0.
    \end{align}
(Note that for this to be true we do not need $U$ to be oracular.)
This means that $Q'$ satisfies~\eqref{quantum_homo_1} and, thus, it is indeed a quantum (oracular) homomorphism from $\X$ to $\B$.
\end{proof}

We can finally complete the proof of~\Cref{main_thm_galois_condition_reductions_nonoracular}.
\begin{proof}[Proof of~\Cref{main_thm_galois_condition_reductions_nonoracular}]
  Since $\qcPol(\rel A)=\qnoPol(\rel A)$,~\Cref{gadget-characterization_non_oracular} guarantees that $\A$ has a non-oracular commutativity gadget. Moreover, we know from $\qnoPol(\rel A)\subseteq\qnoPol(\rel B)$ and~\Cref{cor_1308_2111} that
  $\A$ q-no-defines $\B$.
  Note also that $\A$ and $\B$ have the same domain since $\qnoPol(\rel A)\subseteq\qnoPol(\rel B)$. Hence, the result directly follows from~\Cref{qgadget-reduction_non_oracular}.
\end{proof}

We now show that \emph{tree gadgets} can be used to produce q-no-definitions.
A  relational structure $\rel X$ is a \emph{tree} if for every sequence of tuples $\tuple t^1,\dots,\tuple t^m$ that all belong to (potentially different) relations $R_i$ of $\rel X$ with arity $r_i$,  the total number of elements appearing in $\tuple t^1,\dots,\tuple t^m$ is at least $1+\sum_{i=1}^m(r_i-1)$.

\begin{lemma}\label{tree-homomorphisms}
    Let $\rel T$ be a tree structure, and let $x,y$ be elements of\/ $\rel T$.
    Let $Q\colon\rel T\qnoto\rel A$ be a non-oracular quantum homomorphism.
    If $Q_{x,a}Q_{y,b} \neq 0$, there exists a homomorphism $h\colon\rel T\to\rel A$ such that $h(x)=a$ and $h(y)=b$.
\end{lemma}
\begin{proof}
    We prove this by induction on the size of $\rel T$.
    Assume first that $\rel T$ consists of a single constraint $(z_1,\dots,z_r)\in R^{\rel T}$ with $x=z_i$ and $y=z_j$ for some $i,j\in[r]$.
    For the argument below, we can assume without loss of generality that $i<j$ up to swapping $x$ and $y$.
    Then we have
    \begin{align*}
        Q_{x,a}Q_{y,b} &= \sum_{c_1,\dots,c_r\in A: c_i=a, c_j=b}\prod_{k=1}^r Q_{z_k,c_k}\\
        &= \sum_{(c_1,\dots,c_r)\in R^{\rel A}: c_i=a,c_j=b}\prod_{k=1}^r Q_{z_k,c_k} & \text{by~\eqref{quantum_homo_1}}
    \end{align*}
    and since $Q_{x,a}Q_{y,b}\neq 0$, it must be that there exists $(c_1,\dots,c_r)\in R^{\rel A}$ such that $c_i=a$ and $c_j=b$.
    We obtain a homomorphism $h\colon\rel T\to\rel A$ by setting $h(z_k)=c_k$ for all $k\in[r]$.

    Otherwise, $\rel T$ consists of the union of some trees $\rel T_1,\dots,\rel T_m$ that share a single vertex $z$.
    Since $\sum_c Q_{x,a}Q_{z,c}Q_{z,c}Q_{y,b} = Q_{x,a}(\sum_{c} Q_{z,c})Q_{y,b} = Q_{x,a}Q_{y,b}\neq 0$, there exists a $c$ such that $Q_{x,a}Q_{z,c}\neq 0$ and $Q_{z,c}Q_{y,b}\neq 0$.
    By the induction hypothesis on the subtrees $\rel T_1,\dots,\rel T_m$, there exist homomorphisms $h_i\colon \rel T_i\to\rel A$ such that $h_i(z)=c$ for all $i$, and such that if $x\in T_i$ (resp.\ $y\in T_i$) then $h_i(x)=a$ (resp.\ $h_i(y)=b$).
    The union of these homomorphisms gives a homomorphism $h\colon\rel T\to \rel A$ with $h(x)=a$ and $h(y)=b$.
\end{proof}

\begin{proposition}\label{pp-def-gadget}
    Let $\rel A$ be a structure, and let $R\subseteq A^2$ be a relation with a classical gadget $\GG$ that is a tree.
    Then $\GG$ is a non-oracular q-definition for $R$.
\end{proposition}
\begin{proof}
    Suppose that $Q\colon\rel R\qto (A;R)$ is a quantum homomorphism.
    Then we have $Q=h_1\oplus\dots\oplus h_d$, where each $h_i$ is a map $[r]\to A$ such that $(h_i(1),h_i(2))\in R$.
    Thus, each $h_i$ can be extended to a homomorphism $g_i\colon \GG\to\rel A$ since $\GG$ is a classical gadget for $R$.
    It follows that $Q':= g_1\oplus\dots\oplus g_d$ is a (non-oracular) quantum homomorphism $\GG\qnoto\rel A$ satisfying~\eqref{itm:qgadget-extension_non_oracular}.

    To prove~\eqref{itm:qgadget-restriction_non_oracular}, let $Q\colon\GG\qnoto\rel A$ be a non-oracular quantum homomorphism.
    Let $x_1,x_2$ be the distinguished vertices of $\GG$.
    Suppose that $Q_{x_1,a_1} Q_{x_2,a_2}\neq 0$.
    Since $\GG$ is a tree, there exists by~\Cref{tree-homomorphisms} a homomorphism $h\colon\GG\to\rel A$ such that $h(x_i)=a_i$ for $i\in\{1,2\}$.
    Thus, $(a_1,a_2)\in R$, and therefore the quantum function $Q'\colon \{1,2\}\qto A$ defined by $Q'_{i,a}=Q_{x_i,a}$ satisfies~\eqref{quantum_homo_1}
    and is a non-oracular quantum homomorphism $\rel R\qnoto (A;R)$.
\end{proof}

The next result is the non-oracular version of~\Cref{qgadget-composition}.
\begin{proposition}\label{qgadget-composition_no}
    Let $\rel A,\rel B$ be structures, and let $R_A\subseteq A^r$ and $R_B\subseteq B^r$ be relations that admit a common non-oracular q-definition\/ $\GG$ over $\rel A$ and $\rel B$, respectively. 
    Then every non-oracular quantum homomorphism $Q\colon\rel A\qnoto\rel B$ is a non-oracular quantum homomorphism $(A;R_A)\qnoto (B;R_B)$.
\end{proposition}
\begin{proof}
    Let $g_1,\dots,g_r$ be the distinguished vertices of $\GG$.
    We only need to prove~\eqref{quantum_homo_1}. Let $(a_1,\dots,a_r)\in R_A$ and $(b_1,\dots,b_r)\not\in R_B$.
    By~\eqref{itm:qgadget-extension_non_oracular}, there exists a non-oracular quantum homomorphism $Q'\colon\GG\qnoto\rel A$ such that $Q'_{g_i,a_i}=\id$ for all $i\in[r]$.
    Let $S=Q \bullet Q'$, which is a non-oracular quantum homomorphism $\GG\qnoto\rel B$.
    By~\eqref{itm:qgadget-restriction_non_oracular}, we have that the quantum function $Q''$ defined by $Q''_{i,b}:= S_{g_i,b}$ for all $i\in[r]$ and all $b\in B$ is a non-oracular quantum homomorphism $\rel R\qnoto (B;R_B)$.
    Thus, $\prod_{i=1}^rQ''_{i,b_i}=0$ and therefore $\prod_{i=1}^r S_{g_i,b_i}=0$.
    Note that $S_{g_i,b_i}$ is by definition $\sum_{a\in A} Q_{a,b_i}\otimes Q'_{g_i,a} = Q_{a_i,b_i}\otimes\id$.
    Thus, $(\prod_{i=1}^r Q_{a_i,b_i})\otimes\id = 0$, from which we obtain that $\prod_{i=1}^r Q_{a_i,b_i}=0$.
    \end{proof}

\begin{remark}
Let $Q\colon \rel X\qnoto\rel Y$ be a non-oracular quantum homomorphism between two digraphs $\rel X$ and $\rel Y$. Suppose that $x,x'\in X$ and $y,y'\in Y$ are such that there exists a directed $\ell$-walk from $x$ to $x'$ in $\rel X$ but there is no directed $\ell$-walk from $y$ to $y'$ in $\rel Y$.
Consider the structure $\GG$ consisting of a path of length $\ell$, and call $u,v$ the two endpoints of this path.
    Since this structure is a tree, \Cref{pp-def-gadget} implies that $(\GG,u,v)$ is a q-no-definition for certain binary relations $R_X\subseteq X^2,R_Y\subseteq Y^2$ in $\rel X$ and $\rel Y$, respectively.
    It follows from~\Cref{qgadget-composition_no} that $Q$ is a non-oracular quantum homomorphism $(X;R_X)\qnoto(Y;R_Y)$.
    By assumption, $R_X$ contains $(x,x')$ and $R_Y$ does not contain $(y,y')$,
    and therefore $Q_{x,y}Q_{x',y'}=0$ by~\eqref{quantum_homo_1}.
    Hence,~\Cref{pp-def-gadget} provides an alternative proof of~\cite[Lemma~4.12]{MancinskaR16}.
\end{remark}

We now show that our undecidability result for odd cycles,~\Cref{thm_quantum_odd_cycles_is_undecidable}, also holds in the non-oracular setting, by virtue of a result from~\cite{Zeman}.

\begin{proposition}
For every odd $m\geq 3$, $\qnoCSP(\rel C_m)$ is undecidable.
\end{proposition}
\begin{proof}
It was observed in~\cite[Proposition~5.27 \& Lemma~4.9]{Zeman} that, if an undirected graph $\A$ having no 4-cycle admits a commutativity gadget, then it also admits a non-oracular commutativity gadget. (The result is phrased therein in terms of \textit{algebraic} commutativity gadgets; however, as noted in~\Cref{gadget-characterization} and~\Cref{gadget-characterization_non_oracular}, this is equivalent to our definitions.) Hence,~\Cref{thm_odd_cycles_qPol_equals_qcPol}  directly implies that $\rel C_m$ has a non-oracular commutativity gadget.
Note now that $\rel C_m$ trivially pp-defines itself via a tree gadget. Hence,~\Cref{pp-def-gadget} implies that $\rel C_m$ q-no-defines itself. Using~\Cref{qgadget-reduction_non_oracular}, we deduce that $\qCSP(\rel C_m)$ reduces to $\qnoCSP(\rel C_m)$ in logspace. Hence, it follows from~\Cref{thm_quantum_odd_cycles_is_undecidable} that $\qnoCSP(\rel C_m)$ is undecidable.
\end{proof}

We conclude this section by observing that the notion of q-no-definitions can be extended to a notion of \textit{q-no-constructions} in a similar way as in the oracular case considered in~\Cref{subsec_body_quantum_reductions}. 

\begin{definition}
\label{defn_q_construction_nonoracular}
Let $\A$ be a $\sigma$-structure and let $\B$ be a $\rho$-struture. We say that $\A$ \textit{q-no-constructs} $\B$ if there exists a $\rho$-structure $\C$ such that 
\begin{itemize}
    \item[$(i)$] the domain of $\C$ is $A^d$ for some $d\in\N$;
    \item[$(ii)$] each relation of $\rel C$ (of arity, say, $r$) is q-no-definable in $\A$ (identifying $C^r=(A^d)^r$ with $A^{dr}$);
    \item[$(iii)$] $\B\qto\C$ and $\C\qto\B$.
\end{itemize}
\end{definition}

Note that, in part $(iii)$, we require $\B$ and $\C$ to be quantum \textit{oracular} homomorphically equivalent. As we shall see, non-oracular homomorphic equivalence is not sufficient in order to obtain a q-no-construction version of~\Cref{qgadget-reduction_non_oracular}.

\begin{proposition}
\label{thm_q_constructions_means_reductions_nonoracular}
Let $\A$ and $\B$ be structures (with potentially different domains and signatures), and suppose that
\begin{itemize}
    \item $\rel A$ has a non-oracular commutativity gadget, and
    \item $\A$ q-no-constructs $\B$.
\end{itemize}
Then $\qCSP(\B)$ reduces to $\qnoCSP(\rel A)$ in logspace.
\end{proposition}
\begin{proof}

Let $\sigma$ and $\rho$ be the signatures of $\A$ and $\B$, respectively, and let $\C$ be the $\rho$-structure with domain $A^d$ witnessing that $\A$ q-no-constructs $\B$, as per~\Cref{defn_q_construction_nonoracular}.
We now construct a structure $\rel D$ in the same way as in the proof of~\Cref{q_construction_implies_qPol_hom} in~\Cref{subsec_body_quantum_reductions}.

Given a tuple $\ba=((a_{1,1},\dots,a_{1,d}),\dots,(a_{r,1},\dots,a_{r,d}))\in  (A^d)^r$, let $\bar\ba\in A^{dr}$ be the tuple defined by 
\[\bar\ba=(a_{1,1},\dots,a_{1,d},\dots,a_{r,1},\dots,a_{r,d}).\]

Consider the signature $\tau$ having the same symbols as $\rho$ such that, if a symbol has arity $r$ in $\rho$, it has arity $dr$ in $\tau$. We let $\rel D$ be the $\tau$-structure with domain $D=A$ and relations defined as follows: For any $R\in\rho$, $R^{\rel D}=\{\bar\bc:\bc\in R^{\rel C}\}$. 

Using part $(ii)$ of~\Cref{defn_q_construction_nonoracular}, we know that $\A$ q-no-defines $\D$. Since, by assumption, $\A$ has a non-oracular commutativity gadget, we deduce from~\Cref{qgadget-reduction_non_oracular} that $\qCSP(\D)$ reduces to $\qnoCSP(\A)$ in logspace. On the other hand, we know from the proof of~\Cref{q_construction_implies_qPol_hom} that $\qPol(\D)\to\qPol(\C)$ and thus, by~\Cref{thm_main_qPol_homo_reductions}, $\qCSP(\C)$ reduces to $\qCSP(\D)$ in logspace. In addition, by part $(iii)$ of~\Cref{defn_q_construction_nonoracular}, $\qCSP(\B)$ and $\qCSP(\C)$ are trivially interreducible in logspace. 
Combining the three reductions yields the required logspace reduction from $\qCSP(\B)$ to $\qnoCSP(\A)$. 
\end{proof}

\begin{remark}
We point out that the Galois connection for q-no-definitions of~\Cref{cor_1308_2111} does not appear to lift to q-no-constructions. In particular,~\Cref{q_construction_implies_qPol_hom} does not immediately extends to the non-oracular setting. The reason is that the order we choose while identifying the elements of $A^{dr}$ with elements of $C^r$ (see the proof of~\Cref{thm_q_constructions_means_reductions_nonoracular}) is now important, as in non-oracular homomorphisms we are not allowed to permute projectors associated with the same constraint. This, in turns, prevents from building  a minion homomorphism $\qnoPol(\D)\to\qnoPol(\C)$, which is needed in order to conclude that q-no-constructions yield minion homomorphisms between the minions of quantum non-oracular polymorphisms.
\end{remark}

\section{Quantum polymorphisms of Boolean structures}
\label{sec_boolean}

In this section, we consider the case of quantum polymorphisms of Boolean structures, and we prove the following result.

\thmbooleancommgadgetclassification*

As a direct consequence of this result, we obtain a $\mathsf{P}$ vs. undecidable complexity dichotomy for quantum CSPs parameterised by Boolean structures.
We let $\rel{XOR}$ be the Boolean structure having two ternary relation symbols $R_+$ and $R_{-}$, whose interpretations are
\begin{align*}
    R_+^{\rel{XOR}}&=\{(1,0,0),(0,1,0),(0,0,1),(1,1,1)\},\\
    R_-^{\rel{XOR}}&=\{(0,1,1),(1,0,1),(1,1,0),(0,0,0)\}.
\end{align*}
Note that $\CSP(\rel{XOR})$ encodes the satisfiability problem for systems of linear equations mod $2$.

\corcomplexitydichotomyboolean*
\begin{proof}
If $\A$ does not pp-define $\rel{XOR}$, it was shown in~\cite{AtseriasKS19} that $\qCSP(\A)=\CSP(\A)$, and the latter problem is solvable in polynomial time via the bounded-width algorithm~\cite{Barto14:jacm}.
Suppose now that $\A$ pp-defines $\rel{XOR}$. This implies that $\Pol(\A)\subseteq\Pol(\rel{XOR})$.
Since $\rel{XOR}$ is not preserved by majority, it follows that $\A$ has no majority polymorphism. Then, it follows from~\Cref{thm_boolean_comm_gadget_classification} that $\A$ admits a commutativity gadget.
Applying~\Cref{prop_pp_plus_comm_equals_q}, we deduce that $\A$ q-defines $\rel{XOR}$. This implies via~\Cref{qgadget-reduction} that $\qCSP(\rel{XOR})$ reduces to $\qCSP(\A)$ in logspace. Since $\qCSP(\rel{XOR})$ is undecidable (as proved in~\cite{slofstra2019set}), we conclude that $\qCSP(\A)$ is undecidable.
\end{proof}

\begin{remark}
\label{remark_comparison_boolean_paddock_slofstra}
    \Cref{cor_complexity_dichotomy_boolean} should be compared with~\cite[Theorem~5.11 (a)]{paddock2025satisfiability}. The result established therein concerns Boolean constraint systems \textit{with contexts}, where each context is a subset of variables of the CSP instance whose associated projectors are required to commute, see~\cite[Definition~3.1]{paddock2025satisfiability}. 
By modifying the contexts in the reductions used to prove undecidability, classical pp-definitions can be employed to obtain sound and complete reductions between such problems.
Syntactically, this can be enforced by requiring that the language admits a full binary relation, see for example~\cite[\S6.1]{AtseriasKS19}. To extend the result to arbitrary Boolean languages, the use of commutativity gadgets appears to be necessary; see also~\cite{culf2024re}. On the other hand, we point out that the dichotomy for Boolean constraint systems with contexts in~\cite{paddock2025satisfiability} is incomparable to ours, in that it also addresses other variants of quantum satisfiability. In contrast,~\Cref{cor_complexity_dichotomy_boolean} applies only to finite-dimensional quantum strategies in the tensor-product model (see also the discussion in~\Cref{subsec_overview_outlook}).
\end{remark}

\subsection{The case of structures without majority polymorphism}

In this subsection, we study the quantum polymorphism minions of Boolean structures that are not preserved under majority.
As a consequence, we give an alternative proof of the fact, originally proved in~\cite[Section 7]{culf2024re}, that any Boolean relational structure that is not invariant under the majority polymorphism has a commutativity gadget. 

If a quantum homomorphism $Q$ has range $\{0,1\}$, then $Q_{x,0}$ is determined by $Q_{x,1}$ since $Q_{x,0}=\id-Q_{x,1}$.
Moreover, the set $\{0,1\}^n$ is naturally in bijection with the powerset of $[n]$.
Given $S\subseteq[n]$, we let $s\in\{0,1\}^n$ be the tuple such that $s_i=1$ if, and only if, $i\in S$.
In the following, if $\rel A$ is a Boolean structure, we  write the projectors of a quantum homomorphism $\rel A^n\qto\rel A$ by $Q_S$ where $S\subseteq[n]$.

We begin our analysis by focussing on the following class of Boolean structures.
\begin{definition}
    Let $t\in\{0,1\}^k$ be an arbitrary Boolean tuple. Then the $t$-translate of $1$-in-$k$ is the relational structure $\rel O_t=(\{0,1\};\{x+_2 t\mid x\in R_{1/k}\})$, where
    \begin{enumerate}
        \item $R_{1/k}=\{x\in \{0,1\}^k\mid \sum_{i=1}^kx_i=1\}$ is the relation of $1$-in-$k$ SAT and,
        \item $+_2$ is componentwise addition modulo 2.

    \end{enumerate}
    The relation of $\rel O_t$ is named $R_t$.
\end{definition}
One can easily see that the ternary majority function on $\{0,1\}$ does not preserve $R_t$; the majority of any three distinct tuples in the relation yields some $t$ that is not in $R_t$, witnessing the non-preservation. 

We show that if we translate $1$-in-$k$ by a tuple $t\in R_{1/k}$, then $\qcPol(\rel O_t)=\qnoPol(\rel O_t)$. First consider the case where $k=3$ and $t=(1,0,0)$:
\begin{proposition}\label{prop:polys-100}
    For every non-oracular polymorphism $Q\colon \rel O_{(1,0,0)}^n\qnoto \rel O_{(1,0,0)}$, the following properties hold:
\begin{enumerate}
    \item If $S\cap T=\emptyset$, then $Q_SQ_T=0$.
    \item If $S\cap T=\emptyset$, then $Q_{S\cup T}=Q_{S\cup T}(Q_S+Q_T)$.
    \item $Q_S=Q_{S\cup T}Q_S$ for all $S$
    \item\label{itm:uniqueness} $Q_S=\sum_{i\in S} Q_{\{i\}}$.
\end{enumerate}
\end{proposition}
\begin{proof}
    Consider $S,T\subseteq [n]$ with $S\cap T=\emptyset$ and let $R:=S\cup T$.
    Then $(r,s,t)$ is in $R_{(1,0,0)}$, which gives the following identities by property~\eqref{quantum_homo_1} of $Q$:
    \begin{itemize}
        \item $Q_{S\cup T}Q_SQ_T = 0$, since $(1,1,1)$ is not in $R_{(1,0,0)}$,
        \item $(\id-Q_{S\cup T})Q_SQ_T = 0$, since $(0,1,1)$ is not in $R_{(1,0,0)}$,
        \item $Q_{S\cup T}(\id-Q_S)(\id-Q_T) = 0$, since $(1,0,0)$ is not in $R_{(1,0,0)}$.
        \item $(\id-Q_{S\cup T})Q_S(\id-Q_T) = 0$, since $(0,1,0)$ is not in $R_{(1,0,0)}$.
    \end{itemize}
    We can add the first two identities to get $Q_SQ_t=0$ so, (1) holds. Moreover, adding the first and third identities yields (2).

    To derive $Q_S=Q_{S\cup T}Q_S$ we add the first and fourth identities to get $Q_S= Q_{S\cup T}Q_S+Q_SQ_T$. Application of (1) then yields (3).

    Finally, from the equalities
    \[Q_{S\cup T}=Q_{S\cup T}(Q_S+Q_T)=Q_{S\cup T}Q_S+Q_{S\cup T}Q_T=Q_S+Q_T\]
    we can derive (4) by inductively building any $Q_S$ from sums of the individual elements of $S$. 
\end{proof}
In particular, it follows that for every non-oracular quantum polymorphism 
\[
Q\colon\rel O_{(1,0,0)}^n\qnoto \rel O_{(1,0,0)}\] and every $S,T\subseteq[n]$, we have $[Q_S,Q_T]=0$, which implies that $\qnoPol(\rel O_{(1,0,0)})$ is the quantum closure of $\Pol(\rel O_{(1,0,0)})$.
\begin{lemma}\label{lem: polys-10...0}
    For any natural number $k\ge 3$, the translation $\rel O_{(1,0,\dots,0)}$ of\/ $1$-in-$k$ by $(1,0,\dots,0)$ has the property that $\qnoPol(\rel O_{(1,0,\dots,0)})=\qcPol(\rel O_{(1,0,\dots,0)})$.
\end{lemma}
\begin{proof}
    Since $(0,0,\dots, 0)=(1,0,\dots, 0)+_2 (1,0,\dots,0)\in R_{(1,0,\dots,0)}$ and $(1,1,\dots,1)\notin R_{(1,0,\dots,0)}$, we can pp-define $0$ in $\rel O_{(1,0,\dots,0)}$ by
    \[a=0\Leftrightarrow R_{(1,0,\dots,0)}(a,\dots, a)\]
    This relation is preserved by every non-oracular quantum polymorphism of $R_{(1,0,\dots,0)}$ since,
    \[Q_{\ba,1}=Q_{\ba,1}\dots Q_{\ba,1}=0\]
    for every $\ba\in \{0,1\}^n$ and $Q\colon \rel O_{(1,0,\dots,0)}^n\to \rel O_{(1,0,\dots,0)}$.
    We then proceed to define $R_{(1,0,0)}$ in $(\{0,1\},R_{(1,0,\dots,0)},0)$ by
    \[R_{(1,0,0)}(x,y,z)\Leftrightarrow R_{(1,0,\dots,0)}(x,y,z,0,\dots,0)\]
    We prove that $R_{(1,0,0)}$ is preserved by every non-oracular $Q\colon \rel O_{(1,0,\dots,0)}^n\to \rel O_{(1,0,\dots,0)}$. Assume we have tuples $\ba,\bb,\bc\in \{0,1\}^n$ with $(\ba,\bb,\bc)$ in the $n$th direct power of $R_{(1,0,0)}$ and variables $x,y,z\in \{0,1\}$ such that $(x,y,z)\notin R_{(1,0,0)}$. 
    Then  $(\ba,\bb,\bc,\mathbf{0},\dots,\mathbf{0})$ is in the $n$th power $R_{(1,0,\dots,0)}$ and $(x,y,z,0,\dots,0)$ is not an element of $R_{(1,0,\dots,0)}$
    
    Now let $Q\colon \rel O_{(1,0,\dots,0)}^n\qnoto \rel O_{(1,0,\dots,0)}$ be a non-oracular quantum polymorphism and consider:
    \begin{align*}
        Q_{\ba,x}Q_{\bb,y}Q_{\bc,z} &=Q_{\ba,x}Q_{\bb,y}Q_{\bc,z}\sum_{w\in \{0,1\}}Q_{\mathbf{0},w}\\
        &=\sum_{w\in \{0,1\}}Q_{\ba,x}Q_{\bb,y}Q_{\bc,z}Q_{\mathbf{0},w}\\
        &=\sum_{w\in \{0,1\}}(Q_{\ba,x}Q_{\bb,y}Q_{\bc,z}Q_{\mathbf{0},w}\dots Q_{\mathbf{0},w})(Q_{\mathbf{0},w}\dots Q_{\mathbf{0},w})\\
        &=(Q_{\ba,x}Q_{\bb,y}Q_{\bc,z}Q_{\mathbf{0},0}\dots Q_{\mathbf{0},0})(Q_{\mathbf{0},0}\dots Q_{\mathbf{0},0})\\
        &+(Q_{\ba,x}Q_{\bb,y}Q_{\bc,z}Q_{\mathbf{0},1}\dots Q_{\mathbf{0},1})(Q_{\mathbf{0},1}\dots Q_{\mathbf{0},1})\\
        &= 0(Q_{\mathbf{0},0}\dots Q_{\mathbf{0},0})+(Q_{\ba,x}Q_{\bb,y}Q_{\bc,z}Q_{\mathbf{0},1}\dots Q_{\mathbf{0},1})0=0
    \end{align*}
    where the penultimate equality follows from a combination of~\eqref{quantum_homo_1} and the fact that $(1,\dots,1)\notin R_{(1,0,\dots, 0)}$.
    
   This proves that every non-oracular quantum polymorphism of $\rel O_{(1,0,\dots,0)}$ preserves $R_{(1,0,0)}$ and thus
       \[\qnoPol(\rel O_{(1,0,\dots,0)})\subseteq \qnoPol(\rel O_{(1,0,0)})=\qcPol(\rel O_{(1,0,0)}).\]
       It follows that every non-oracular quantum polymorphism of $\rel O_{(1,0,\dots,0)}$ is non-contextual, which completes the proof.
\end{proof}
For a generalisation of~\Cref{lem: polys-10...0} to an arbitrary $t\in R_{1/k}$, we can pp-define $R_t$ in $\rel O_{(1,0,0\dots,0)}$ by
\[R_t(x_1,\dots,x_k)\Leftrightarrow R_{(1,0,0\dots,0)}(x_l,x_1\dots,x_{l-1},x_{l+1},\dots,x_k)\]
where $1\le l\le k$ is the singular non-zero component of $t$.

\begin{corollary}\label{I1}
    If $k\ge 3$ and $t\in R_{1/k}$, the non-oracular polymorphisms of $\rel O_t$ are non-contextual.
\end{corollary}
Now for the when $t\notin R_{1/k}$. We reduce it to the case from~\Cref{I1}.
\begin{lemma}\label{I2}
    Assume $k\ge 3$, and neither $t$ nor $\overline{t} $ ($= (1,\dots,1)+_2 t$) are in $ R_{1/k}$. Then $\qnoPol(\rel O_t)=\qcPol(\rel O_t)$.
\end{lemma}
\begin{proof}
    Without loss of generality, assume that $t=t^l$ is the tuple with $l$ 1-entries at the start, one can do this from an arbitrary $t$ by just permuting the entries.
    
    If $l\notin \{1,k-1\}$ (i.e., neither $t$ nor $\overline{t}$ are in $ R_{1/k}$) then $(0,\dots,0)=t+_2 t$ and $(1,\dots,1)=t+_2 \overline{t}$  are not in $R_t$. We first pp-define $\neq$ in $\rel O_{t}$ via
    \[x\neq y \Leftrightarrow R_t(x,\dots,x,y,\dots,y)\]
    where we fill the first $l+1$ entries with $x$'s. This relation is preserved by every non-oracular $Q\colon\rel O_{t}^n\to \rel O_{t}$ since for every $\ba\neq\bb\in \{0,1\}^n$ and $x\in\{0,1\}$ we have
    \[Q_{\ba,x}Q_{\bb,x}= Q_{\ba,x}\dots Q_{\ba,x}Q_{\bb,x}\dots Q_{\bb,x}=0.\]
    Let $t'=(0,1\dots,1,0,\dots,0)$ ($t^l$ with the first entry switched).
    We define $R_{t'}$ by the formula
    \[R_{t'}(s_1,\dots,s_k)\Leftrightarrow\exists q\left(R_t(q,s_2,\dots,s_k)\wedge q\neq s_1 \right).\]
    Assume we have tuples $\bs_1,\dots, \bs_k\in \{0,1\}^n$ with $(\bs_1,\dots, \bs_k)$ in the $n$th direct power of $R_{t'}$ and variables $x_1,\dots,x_k\in \{0,1\}$ such that $(x_1,\dots,x_k)\notin R_{t'}$. Then we have $\bq\in \{0,1\}^n$ such that $(\bq,\bs_2,\dots,\bs_k)$ are elements of the $n$th power of $R_{t}$ and $\bq\neq\bs_1$ . Also, there is no $y\in \{0,1\}$ such that both $(y,x_2,\dots,x_k)\in R_{t}$ and $y\neq x_1$, i.e., we know that $(\neg x_1,x_2,\dots,x_k)\notin R_t$.

    Now let $Q\colon \rel O_{t}^n\qnoto \rel O_{t}$ be a non-oracular polymorphism of $\rel O_t$.
    Then we have
    \begin{align*}
        Q_{\bs_1,x_1}\dots Q_{\bs_k,x_k} &=  Q_{\bs_1,x_1}\left(\sum_{y\in\{0,1\}}Q_{\bq,y}\right)Q_{\bs_2,x_2}\dots Q_{\bs_k,x_k}\\
        &=  \sum_{y\in\{0,1\}}Q_{\bs_1,x_1}Q_{\bq,y}Q_{\bs_2,x_2}\dots Q_{\bs_k,x_k}\\
        &= \sum_{y\in\{0,1\}}(Q_{\bs_1,x_1}Q_{\bq,y})(Q_{\bq,y}Q_{\bs_2,x_2}\dots Q_{\bs_k,x_k})\\
        &= (Q_{\bs_1,x_1}Q_{\bq,x_1})(Q_{\bq,x_1}Q_{\bs_2,x_2}\dots Q_{\bs_k,x_k})+(Q_{\bs_1,x_1}Q_{\bq,\neg x_1})(Q_{\bq,\neg x_1}Q_{\bs_2,x_2}\dots Q_{\bs_k,x_k})\\
        &= (0)(Q_{\bq,x_1}Q_{\bs_2,x_2}\dots Q_{\bs_k,x_k})+(Q_{\bs_1,x_1}Q_{\bq,\neg x_1})(0)=0
    \end{align*}
   where the penultimate equality follows from~\eqref{quantum_homo_1}.  This proves that every non-oracular quantum polymorphism of $\rel O_{t^l}$ preserves $R_{t'}$ and thus
   \[\qnoPol(\rel O_{t^l})\subseteq \qnoPol(\rel O_{t'})=\qnoPol(\rel O_{t^{l-1}})\]
    by induction, we then get that
     \[\qnoPol(\rel O_{t})\subseteq \qnoPol(\rel O_{(1,0,\dots,0)})\]
    completing the proof after an application of~\Cref{lem: polys-10...0}.
    
    Note that this theorem also covers the case when $t=(0,\dots,0)$ i.e. $R_t=R_{1/k}$ (or $l=0$).  In this case $R_{t'}$ as pp-defined above is itself $R_{(1,0,\dots,0)}$ and we are done immediately.
\end{proof}

Thus $\qPol(\rel O_{t})=\qcPol(\rel O_{t})$ for every $k$-tuple $t$. We now argue that this result covers every Boolean structure that is of interest to us. Specifically, we establish the following result, using similar ideas as in~\cite[Proposition~7.19]{culf2024re}. 
\begin{theorem}\label{nonprojection<=>oneink}
    For any relation $R\subseteq \{0,1\}^k$ the following statements are equivalent:
    \begin{enumerate}
        \item\label{itm:} $R$ is not invariant under majority, majority preserves every proper projection of  $R$, and every binary projection of $R$ is not full.
        \item $R=R_t$ for some tuple $t$.
    \end{enumerate}
\end{theorem}

\begin{proof}${}$

\begin{itemize}
    \item[$2.\Rightarrow 1.$]
    Note that translating a relation by a tuple does not change the size of the binary projections and does not change preservation by majority, since majority is a self-dual operation.
 It follows that it suffices to prove the implication in the case of $R_{1/k}$, for which all the properties in $1.$ are clear.
    \item[$1.\Rightarrow 2.$]
    Suppose \( R \subseteq \{0,1\}^k \) is not invariant under majority but majority does preserve every proper projection and no binary projection is full. Suppose $t\in \{0,1\}^k$ is a witness to the fact that $R$ is not invariant under majority. We first show that $\{x+_2 t\mid x\in R_{1/k}\}\subseteq R$.

        For $i\in[k]$, let $x\in R_{1/k}$ be the tuple with a 1 in the $i$th position.
        We know that $\pi_{[k]\setminus\{i\}}(R)$ is invariant under majority, so there is an $r_i\in R$ with $\pi_{[k]\setminus \{i\}}(r_i)=\pi_{[k]\setminus \{i\}}(t)$.
        Now, since $t\notin R$, we can conclude that $\pi_{\{i\}}(r_i)\neq\pi_{\{i\}}(t)$
        and therefore $x+_2 t = r_i$, so $x+_2 t\in R$.

        Finally, we use the non-fullness of binary projections to argue that any $r\in R$ must match $\overline{t}:=(1,\dots,1)+_2 t$ in exactly one argument, placing it in $\{x+_2 t\mid x\in R_{1/k}\}$.

        Assume that $r$ matches $\overline{t}$ in two arguments $1\le i<j\le k$. We know that $\pi_{i,j}(R)$ contains $\pi_{i,j}(t)$, $(\pi_i(t),\pi_j(\overline{t}))$, and $(\pi_i(\overline{t}),\pi_j(t))$; then $\pi_{i,j}(R)$ is full, a contradiction. Additionally, $r$ cannot match $\overline{t}$ on zero arguments since then it would equal $t$, which is not in $R$. Thus every $r$ is in $\{x+_2 t\mid x\in R_{1/k}\}$, completing the proof.\qedhere
\end{itemize}

\end{proof}
Since a structure is TVF if, and only if it doesn't have any full binary projections,~\Cref{prop_nonTVF_implies_non_contextual} allows us to apply~\Cref{nonprojection<=>oneink} and prove the following:
\begin{theorem}\label{thm_nomaj_contextual}
    If $\rel A$ is a Boolean relational structure such that $\Pol(\rel A)$ does not contain majority, then $\qPol(\rel A)=\qcPol(\rel A)$.
\end{theorem}
\begin{proof}
     Consider a relational structure $\rel A$; without loss of generality, every projection of a relation of $\rel A$ is also a relation of $\rel A$, since this does not change the hypothesis or the conclusion of the statement.
     Let $R$ be a relation of $\rel A$ that is not preserved by majority and of minimal arity.
     By~\Cref{prop_nonTVF_implies_non_contextual}, if a binary projection of $R$ is full then we are done, so assume that this is not the case. Then~\Cref{nonprojection<=>oneink}
     implies that $R=R_t$ for some $k$-tuple $t$.

     Finally, we know that $k\ge 3$ since every binary relation of arity 2 is preserved under majority and so, most of the result follows from~\Cref{I1} and~\Cref{I2}. All that remains is for us to argue that~\Cref{I1} also covers the case where $\overline{t}\in R_{1/k}$.
     
    This is since, $\rel O_t$ has `essentially' the same polymorphisms as $\rel O_{\overline{t}}$ for any $t$. The formal fact we use is that $\{Q_S\mid S\subseteq [n]\}$ is a non-oracular quantum polymorphism of $\rel O_t$ if, and only if, $\{\id-Q_{\overline{S}}\mid S\subseteq [n]\}$ is a (non-oracular) quantum polymorphism of $\rel O_{\overline{t}}$. In particular $\rel O_t$ has contextual polymorphisms if, and only if, $\rel O_{\overline{t}}$ does. Thus we can assume that $\overline{t}\notin R_{1/k}$, completing the proof. 
\end{proof}
\subsection{The minimal contextual minion}
\label{subsec_body_minimal_contextual_clone}
In this section, we provide a minion with contextual polymorphisms that is contained within the quantum polymorphisms of any Boolean relational structure that has contextual polymorphisms.

In the previous section, we have seen that every Boolean relational structure with contextual polymorphisms has non-full binary projections (is TVF) and contains majority in the minion of (classical) polymorphisms. 

Let $\minimalClone$ be the structure $\minimalClone:=(\{0,1\}; S_{00}, S_{11}, S_{10})$ where $S_{ab}=\{0,1\}^2\setminus\{(a,b)\}$. In the classical setting, $\Pol(\minimalClone)$ is the minimal clone containing the majority operation.
We first show that this property extends to oracular quantum polymorphisms.
\begin{theorem}\label{thm_DM_quantum_minimality}
    Let $\rel A$ be a Boolean TVF structure that admits majority as a polymorphism.
    Then $\qPol(\minimalClone)\subseteq\qPol(\rel A)$.
\end{theorem}
\begin{proof}
    Consider a relational structure $\rel A:=(\{0,1\}; R)$ that is preserved by majority, then it is well known that $R$ (of arity $r$) can be pp-defined in $\rel A_2:=(\{0,1\};\{\pi_{i,j}(R)\mid 1\le i,j \le r\})$ via the conjunction of its binary projections $\bigwedge_{1\le i,j\le r}\pi_{\{i,j\}}(R)(x_i,x_j)$.
    We show that $R$ is preserved by every quantum polymorphism $Q\colon\rel A_2^n\qto \rel A_2$. Assume that $(\ba_1,\dots,\ba_r)$ is in the $n$th direct power of $R$ while in $(x_1,\dots, x_r)$, we have indices $i,j$ such that $(x_i,x_j)\notin \pi_{\{i,j\}}(R)$. Then
    \[Q_{\ba_1,x_1}\dots Q_{\ba_r,x_r}= (Q_{\ba_i,x_i}Q_{\ba_j,x_j})Q_{\ba_1,x_1}\dots Q_{\ba_r,x_r}=(0)Q_{\ba_1,x_1}\dots Q_{\ba_r,x_r}=0\]
    where the second equality is by duplicating the $Q_{\ba_i,x_i}$ and $Q_{\ba_j,x_j}$ projectors and then using property~\eqref{quantum_homo_2} to move them to the front. Moreover, if $(\ba,\ba')\in E(\Gaif(\rel A^n))$ then  $(\ba,\ba')\in E(\Gaif(\rel A_2^n))$ so, for any $x,x'\in\{0,1\}$ $[Q_{\ba,x},Q_{\ba',x'}]=0$ by property~\eqref{quantum_homo_2}. $Q$ is then an oracular quantum polymorphism of $\rel{A}$ and so $\qPol(\rel A_2)\subseteq\qPol(\rel A)$. 

    By assumption, no binary projection of a relation of $\rel A$ is full, thus,
    each relation $R$ in $\rel A_2$ can be defined as the (non-empty) conjunction $\bigwedge_{(c,d)\not\in R} S_{cd}(x,y)$,
    where we use the shorthand $S_{01}(x,y)$ for the formula $S_{10}(y,x)$.
    We show that each relation $R$ of $\rel A_2$ is preserved by the quantum polymorphisms of $\minimalClone$.
    
    Let $Q\colon \minimalClone^n\qto \minimalClone$ be a quantum polymorphism, $(\ba,\bb)$ be in the $n$th direct power of $R$ and $(c,d)\notin R$.
    We then have $Q_{\ba,c}Q_{\bb,d}=0$ 
    by~\eqref{quantum_homo_1} since $(c,d)\not\in S_{cd}$ and $(\ba,\bb)$ is in the relation $S_{cd}$ in the $n$th power of $\minimalClone$.
    Thus,~\eqref{quantum_homo_1} holds.

    Finally, if $(\ba,\bb)$ is a pair in a relation $R$ in the $n$th power of $\rel A_2$, then it is necessarily a pair in one of the relations $S_{cd}$ in the $n$th power of $\minimalClone$.
    Thus, $Q_{\ba}$ and $Q_{\bb}$ commute, so that~\eqref{quantum_homo_2} holds.
\end{proof}

Now we show that the minion defined here, actually contains some contextual polymorphisms.

\begin{theorem}\label{thm_DM_contextuality}
For $n\le 3$, every quantum polymorphism $Q\colon\minimalClone^n\to \minimalClone$  is non-contextual. 
For all $n\geq 4$, $\minimalClone$ has contextual quantum polymorphisms of arity $n$.
\end{theorem}
\begin{proof}
    We start by proving that every quantum polymorphism $Q\colon\minimalClone^n\qto\minimalClone$ is non-contextual for $n\leq 3$.
    For this, it suffices by~\eqref{quantum_homo_2} to prove that for every pair $S,T$ of subsets of $\{1,\dots,n\}$, either $(S,T)$ or $(T,S)$ belongs to one of the relations of $\minimalClone^n$.
    If $(S,T)$ does not belong to the relation $S_{00}$ in $\minimalClone^n$, then there exists $i\in\{1,\dots,n\}$ such that $i\in \overline S\cap\overline T$.
    Similarly, for it to not belong to the relation $S_{11}$ in $\minimalClone^n$, there must exist $j\in\{1,\dots,n\}$ such that $j\in S\cap T$.
    If $n\leq 3$, there remains at one most other element $k$ in $\{1,\dots,n\}\setminus\{i,j\}$.
    If $k$ is not in $S\cap T$ nor in $\overline S\cap \overline T$, then it is the case that $(S,T)$ or $(T,S)$ belongs to the relation $S_{10}$ of $\minimalClone^n$, so we are done.

    We now build a concrete class of contextual polymorphisms of arity $n=4$. We provide~\Cref{fig:O_b^4} as a visual aid in this example. The figure is a simplified presentation of the relatedness of tuples in $\{0,1\}^4$ with the arrows indicating the poset $S_{10}$ and the colouring emphasizing that $S_{11}$ (and $S_{00}$) relate all opposite tuples. In particular, the only unrelated tuples are those $s,t$ such that $|S|=|T|=2$ and $S\neq \overline{T}$.
   
    We define $Q\colon \minimalClone^4\qto \minimalClone$ by listing the projectors. Let  $T\subseteq [4]$.
    If $|T|\geq 3$, define $Q_T=\id$.
    Pick now arbitrary matrices $A,B,C$.
    Define $Q_{\{1,2\}}:=A$, $Q_{\{1,3\}}=B$, and $Q_{\{1,4\}}=C$.
    In all the other cases, define $Q_T = \id - Q_{[4]\setminus T}$.
    Note that if at least two of the matrices $A,B,C$ do not commute, then $Q$ is contextual.

    It remains to prove that $Q$ preserves the relations of $\minimalClone^4$.  
    Assume we have $S,T\subseteq [4]$ with $(s,t)$ in the fourth power of $S_{11}$ and consider $Q_SQ_T$ (since $(1,1)$ is the only Boolean pair not in $S_{11}$). $S$ and $T$ are related if, and only if, $S\cap T=\emptyset$; hence, we have two cases: Either $|S|=|T|=2$ and $S=\overline{T}$, or one of $S,T$ has cardinality less than $2$.
    
    If $S\neq \overline{T}$, then $|S|\le 1$ or $|T|\le 1$ in this case at least  $Q_S$ or $Q_T$ equals $0$ and we are done immediately. If $S=\overline{T}$, then 
    $Q_SQ_T=Q_SQ_{\smash{\overline{S}}}=Q_S(\id-Q_S)=0$ and we are done.
    The case of $S_{00}$ is symmetric.

    Now assume we have $S,T\subseteq [4]$ such that the corresponding pair of tuples $(s,t)$ in the fourth power of $S_{10}$ ($S\subseteq T$).  If both $|S|$ and $|T|$ do not equal $2$, we know that all projectors $Q_S$, $Q_T$, $(\id-Q_S)$, and $(\id-Q_T)$ are either $0$ or $\id$; thus, we immediately find that $[Q_S,{\id-Q_T}]=0$.
    The final case is when at least one of $|S|$ or $|T|$ equals $2$. In this case, we deduce that either
    \[Q_S(\id-Q_T)=0(\id-Q_T)=0\]
    when $|T|=2$ (and $S\neq T$) or,
    \[Q_S(\id-Q_T)=Q_S(\id-\id)=0\]
     when $|S|=2$ (and $S\neq T$). Thus all relations are preserved, and $Q\colon \minimalClone^4\qto \minimalClone$ is a quantum polymorphism.

     Note that one can easily extend $Q$ to a contextual polymorphism of arbitrary arity $\ge 4$ by adding dummy variables.
    For example, define $\tau\colon [4]\to [n]$ by $\tau(i):=i$ for $i\in [4]$.
    By~\Cref{prop_qPol_abstract_minion},
    \[Q_{/\tau}=Q\bullet
             \bigotimes_{j\in[4]}
             \pi_{\tau(j)}^n\colon M^n\qto M\]
    is in $\qPol(\minimalClone)$. We just need to check that $Q_{/\tau}$ is non-contextual, but this is immediate from the definition of $Q$, since e.g. $(Q_{/\tau})_{\{1,2\}}=Q_{\{1,2\}}=A$, $(Q_{/\tau})_{\{1,3\}}=Q_{\{1,3\}}=B$, and $[A,B]\neq 0$.
\end{proof}

We can now conclude by proving the main result on this section, giving a full classification of the existence of commutativity gadgets for Boolean structures. As remarked in~\Cref{subsec_overview_boolean}, the $(ii)\Rightarrow (i)$ implication was first proved in~\cite{culf2024re}, which described explicit commutativity gadgets for $\A$ having no majority and for any $\A$ that is not TVF.

\begin{proof}[Proof of~\Cref{thm_boolean_comm_gadget_classification}]\mbox{}
\begin{itemize}
    \item [$(ii)\Rightarrow (i)$] 
    If $\rel A$ has no majority polymorphism, then it admits a commutativity gadget by~\Cref{thm_nomaj_contextual} and~\Cref{gadget-characterization_overview_friendly}. Otherwise, the result follows from~\Cref{prop_nonTVF_implies_non_contextual} and~\Cref{gadget-characterization_overview_friendly}.

    \item[$(i)\Rightarrow (ii)$] By contrapositive, if $\A$ admits a majority polymorphism and is TVF, then  it holds that $\qPol(\minimalClone)\subseteq\qPol(\rel A)$ by~\Cref{thm_DM_quantum_minimality}. Thus $\A$ has at least one contextual polymorphism by~\Cref{thm_DM_contextuality}. So, it has no commutativity gadget.\qedhere
\end{itemize}
\end{proof}

\section{Acknowledgments}

The first and last author are thankful for the attentive listening and comments by Peter Zeman, to whom they explained the definition of quantum polymorphisms and the main results in this paper concerning the characterisation of the existence of commutativity gadgets on the occasion of the CSP World Congress in September 2025.
The second and last author acknowledge funding from the project “Hamburg
Quantum Computing”, co-financed by ERDF of the European Union and by the Fonds of the Hamburg
Ministry of Science, Research, Equalities and Districts (BWFGB).

\addcontentsline{toc}{section}{References}
\printbibliography

\end{document}